\documentclass[aps,pre,reprint,superscriptaddress,longbibliography]{revtex4-2}

\usepackage{amsmath}
\usepackage{amssymb}
\usepackage{amsthm}
\usepackage{graphicx}
\usepackage{booktabs}
\usepackage{tikz}

\makeatletter
\g@addto@macro\@arrayparboxrestore{\raggedright\let\\\tabularnewline}
\makeatother

\theoremstyle{plain}
\newtheorem{maintheorem}{Main Theorem}

\newtheorem{theorem}{Theorem}
\newtheorem{lemma}{Lemma}
\newtheorem{proposition}{Proposition}
\newtheorem{corollary}{Corollary}
\newtheorem{conjecture}{Conjecture}
\theoremstyle{definition}
\newtheorem{definition}{Definition}
\theoremstyle{remark}
\newtheorem{remark}{Remark}

\newcommand{\Phifit}{\Phi_{\mathrm{fit}}}
\newcommand{\JD}{J_{D}}
\newcommand{\dpJD}{\Delta^{+}\!J_{D}}
\newcommand{\dJD}{\Delta J_{D}}
\newcommand{\V}{V}
\newcommand{\Winformed}{W_{\mathrm{informed}}}
\newcommand{\Wblind}{W_{\mathrm{blind}}}
\newcommand{\sigmaM}{\sigma_{M}}
\newcommand{\Sigmatot}{\Sigma_{\mathrm{total}}}
\newcommand{\etacap}{\eta_{\mathrm{cap}}}
\newcommand{\Lgen}{L_{\mathrm{gen}}}
\newcommand{\rhogen}{\rho_{\mathrm{gen}}}
\newcommand{\flatstar}{\mathrm{flat}^{*}}
\newcommand{\kT}{k T}
\newcommand{\Gap}{\mathrm{Gap}}

\begin{document}

\title{Thermodynamics of Learning:\\ A Typed Four-Component Accounting of Memory, Fit, and Value}

\author{Akihito Sudo}
\affiliation{%
    ZeroStruct Inc.\\
    sudo.akihito@shizuoka.ac.jp
}

\date{\today}

\begin{abstract}
What a finite learning device has recorded and what will hold value for it
on future tasks are not the same quantity.
We develop a typed accounting for finite-state learning devices that
separates four components: a training-side fit functional $\Phifit$, the
record-correlation stock $\JD=I(M;D)$, an update-side search ledger
$\sigmaM$, and an operational capital value $\V(M;T,b)$.
This value is the work gap between an informed protocol class and a blind
class obtained by deleting the memory-read port and re-optimizing from
scratch. Its explicit arguments are the future task distribution $T$, the
access structure, and a per-run work budget $b$.
Three theorem groups organize the accounting.
(I)~\emph{Separation}: $\V$ is well posed, invariant under bijective
re-labeling of the memory, and reduces to the single-task value of a
companion framework. For every $n$, there is a device family on which
record correlation and world correlation grow by $n\ln 2$ while the
capital gain is exactly zero. In the $\flatstar$ regime, data-free
updates never increase $\V$.
(II)~\emph{Capitalization ledger}: an exact $\flatstar$ extraction
identity and a universal ledger identity give, for (F5$'$)-stable
$M$-local updates under a no-discarded-record-correlation
condition~(f), the bound $\etacap\le 1$
for the capitalization efficiency $\etacap=\Delta \V/(\kT\,\sigmaM)$,
together with necessary and sufficient conditions for equality. Outside
this regime, the bound fails in identified ways: pure entropy-production
denominators, finite-budget gate enablement, and recycling credit. The
regime map organizes these failures. A blank-start cumulative bound
survives without~(f).
(III)~\emph{Value retention}: for the retention gap $\Lgen$ and retention
ratio $\rhogen$ (the former carries no sign constraint; the latter is
defined for positive training-side value and is not confined to $[0,1]$)
we give a two-layer alignment domain: an exact exchange rate between
value and the side-information-adjusted record fit $I(M';D\mid Y)$
without any record--side-information independence assumption, and a raw
record-stock exchange rate under a joint side-information neutrality
condition $(M,D)\perp Y$, whose boundary is marked by an explicit
one-time-pad witness. We also give a four-coordinate intervention
theorem. Shift, budget, access, and record content each admit a paired
setting in which changing that coordinate alone reverses or restores
the fit--value ranking. The neutrality coordinate alone carries a
genuine one-condition-drop necessity witness.
Corollaries include a subject-difference identity for the retention gap
and the two-way failure of ``overfitting equals low efficiency'':
$\etacap$ and $\rhogen$ admit no functional or monotone relation. Every
pair in a full rectangle is realized by an explicit device and shift.
These are statements about finite-device value retention under
task-distribution shift, not a theory of statistical generalization.
\end{abstract}

\maketitle

\section{Introduction}
\label{sec:intro}

Does what a learning device has memorized coincide with what will hold
value for it in the future? For a finite physical device, the two
sides of this question belong to different accounts. Memorization is
a correlation inventory: the information carried by the device's
memory about its training record or the variables of the world.
Future value is operational: the additional work that readable memory
is worth on a stated distribution of future tasks, under a stated
access structure and a stated per-run work budget. Information
thermodynamics has developed the execution layer of learning in depth.
This layer covers the entropy production of memory updates and its
bound on information acquisition, ledger identities for bipartite
information flows, the dissipative price of retaining nonpredictive
correlation, and work values of stored
information~\cite{Sagawa2010Generalized,Goldt2017PRL,Goldt2017NJP,%
Horowitz2014,Hartich2014,Still2012,Kolchinsky2018Semantic,Boyd2022TML}.
This paper develops the complementary account through typed
bookkeeping. What was recorded and what will hold value remain
quantities of different type, connected by theorems rather than merged
by definition. Section~\ref{sec:related} tabulates what each
neighboring lineage already provides and what this paper imports from
it.

We fix one boundary at the outset. All results below are exact
statements about finite-state learning devices in a fixed operational
framework: work values on explicit task distributions, under explicit
budgets and access structures. They concern finite-device value
retention under task-distribution shift
(Sec.~\ref{sec:retention}), and they are not a theory of statistical
generalization. No sample-complexity, risk-bound, or
i.i.d.-asymptotic content is claimed. The retention quantities defined
below neither reduce to nor estimate a generalization error
(Sec.~\ref{sec:discussion}).

The accounting is typed (Table~\ref{tab:types}). It keeps four
component currencies separate: acquisition cost (C1), physical
dissipation (C2), transport action (C3), and task value (C4). On the
training side, we distinguish three quantities that the single word
``fit'' tends to merge: a training-side fit functional $\Phifit$,
which ranks candidate updates by training loss, training likelihood,
or training work; the record-correlation stock $\JD(M):=I(M;D)$; and
the gross acquisition $\dpJD:=I(M';D\mid M)$ of a single update. On
the update side, the central bookkeeping quantity is the search ledger
$\sigmaM=\Delta I(M;D)+\Sigmatot$, a memory-side subsystem account
that is not a pure-dissipation quantity (Sec.~\ref{sec:ledger}); on
the future side, it is the capital value $\V(M;T,b)$
(Sec.~\ref{sec:value}). The four components are not a table for
classifying learning phenomena. They form a type system. The theorems
below specify the regimes in which exchange laws between the
sub-accounts hold and identify how those laws fail outside those
regimes. We do not claim that every theorem requires all four
components; each theorem uses an explicitly stated subset.

The central definition uses a deletion counterfactual. The capital
value of a memory state $M$ is
\begin{equation}
\V(M;T,b)=\Winformed(M;T,b)-\Wblind(T,b):
\label{eq:V-intro}
\end{equation}
the optimal expected work extractable over the future task
distribution $T$ when the memory-read port is available, minus the
optimum of a blind agent from whom every read port on $M$ has been
deleted and who re-optimizes from scratch under the same tasks, the
same access structure, and the same per-run budget $b$.
\emph{Learning}, in this paper, is an admissible memory-local update
with $\Delta\V>0$; the predicate is relative to the stated task
distribution, access structure, and evaluation budget, and it differs
both from a record-correlation increase, $\Delta I(M;D)>0$, and from
an improvement of any training-side fit functional. The distinction
between these predicates is not merely definitional. It
is the subject of the first main theorem.

The results fall into three theorem groups. Individual devices and
auxiliary lemmas are subordinate to those groups rather than parallel
claims.
\emph{Main Theorem~\ref{thm:main-I} (separation;
Sec.~\ref{sec:value})}: $\V$ is well posed, nonnegative, affine in
$T$, invariant under bijective recoding of the finite memory, and
reduces at a single task to the informed--blind value of the
companion framework~\cite{Sudo_MPU}. For every $n$ there is a device
family on which the record correlation and the world correlation grow
by $n\ln2$ along an admissible update sequence while the capital gain
is exactly zero. In the $\flatstar$ regime, data-free updates never
increase $\V$.
\emph{Main Theorem~\ref{thm:main-II} (capitalization ledger;
Secs.~\ref{sec:ledger} and~\ref{sec:regime})}: the central question is
which part of the update account is capitalized into future value.
This is the capitalization of the search ledger. An exact
$\flatstar$ extraction identity and a universal ledger identity
combine, for (F5$'$)-stable $M$-local updates under a
no-discarded-record-correlation condition~(f), into
the bound $\etacap\le1$ for the capitalization efficiency
$\etacap=\Delta\V/(\kT\,\sigmaM)$, with necessary and sufficient
conditions for equality. Outside this regime, the bound fails in
identified ways: pure entropy-production denominators,
finite-budget gate enablement, and recycling credit. A regime map
organizes these failures, while a blank-start cumulative bound
survives without~(f).
\emph{Main Theorem~\ref{thm:main-III} (value retention;
Sec.~\ref{sec:retention})}: for the retention gap $\Lgen$ and
retention ratio $\rhogen$ under task-distribution shift, there is a
two-layer alignment domain. Within this domain, value equals the
side-information-adjusted record fit $\kT\,I(M';D\mid Y)$ exactly,
with no independence assumption between record and side information.
Under a joint side-information neutrality condition
$(M,D)\perp Y$, it also equals the raw record stock
$\kT\,I(M';D)$. An explicit one-time-pad witness exhibits the boundary
of this condition. The theorem also gives a four-coordinate
intervention theorem: shift, budget, access, and record content each
admit a paired setting in which changing that coordinate alone
reverses or restores the fit--value ranking. Only the neutrality
coordinate carries a genuine one-condition-drop necessity witness.
Two corollaries and an application follow (Secs.~\ref{sec:accounting}
and~\ref{sec:gate}): a subject-difference identity for the retention
gap; the two-way failure of ``overfitting equals low efficiency''
($\etacap$ and $\rhogen$ admit no functional or monotone relation);
and a gate-family contrast between per-bit linear retention and a
threshold transition on a common redraw-depth axis.

The claim level for real learning systems is fixed once. Every
theorem concerns explicitly constructed device classes inside a fixed
protocol framework. These are isolation results in the sense of
device theory. When the phenomena resemble machine learning, as in
memorization without value, forgetting, or value enabled by an
evaluation budget, we assert only a \emph{correspondence} of
accounting structure. We claim neither a model nor an explanation of
any specific biological or artificial system.

On novelty we make one dated and scoped statement: within the fence
of prior work surveyed up to 2026-07-12, we did not find a prior
framework simultaneously equipped with the same operational deletion
value, evaluation budget, access structure, search ledger, and
four-axis alignment map. Each ingredient separately has close
relatives, and several identities used below are imports; the
comparison tables and the attribution of every imported piece are in
Sec.~\ref{sec:related}.

Section~\ref{sec:setting} fixes the operational setting and the type
table. Section~\ref{sec:value} defines the training-side and
value-side quantities and proves Main Theorem~\ref{thm:main-I}.
Sections~\ref{sec:ledger} and~\ref{sec:regime} develop the
capitalization ledger and its regime map (Main
Theorem~\ref{thm:main-II}). Section~\ref{sec:retention} develops
value retention (Main Theorem~\ref{thm:main-III}),
Sec.~\ref{sec:accounting} its accounting corollaries, and
Sec.~\ref{sec:gate} the gate-family application.
Section~\ref{sec:related} locates the framework among its neighbors,
and Sec.~\ref{sec:discussion} records limitations. The appendices
collect device definitions and numerical verification, proofs,
witness details, self-contained fixed-content statements, and the
notation table (Table~\ref{tab:notation}, Appendix~\ref{app:notation}).
This table fixes the primary term for every central quantity of the
accounting; the body uses these terms and no synonyms.

\section{Operational setting and typed accounting}
\label{sec:setting}

\subsection{Imported framework and conventions}
\label{sec:setting-imports}

We import the protocol framework unchanged from the companion
manuscript~\cite{Sudo_MPU} and restate the parts used here.
Throughout, $\kT=1/\beta$, information is measured in nats, and bit
values are displayed as explicit multiples of $\ln2$.

A \emph{budgeted protocol class} $\mathrm{Prot}_b(A,q)$ consists of
compositions of at most $K_{\mathrm{hor}}$ single-step multipartite
maps subject to the following conditions. (a) Under \emph{mechanism
locality}, each designed mechanism acts only on its declared targets
as a function of its declared parents. (b) For the
\emph{dynamic access structure}
$A=(\mathrm{Acc}_0,R)$, the parent set of every stage-$k$
mechanism is contained in the currently licensed set $\mathrm{Acc}_k$
issued by the rule $R$ on the record history. (c) The protocol carries
a \emph{model} $q$; $q=p$ throughout this paper, so the knowledge axis
is not exercised. (d) Under the \emph{per-run pathwise drawdown cap},
on every positive-probability trajectory $\omega$ and at every stage
$k$, the running account satisfies $E_k(\omega)\le b$.
Auxiliaries are priced one-shot at their $D_{\max}$ value, and an
average-budget constraint is not admitted because biased-coin
strategies defeat it in the imported framework. (e) No \emph{unpriced
resources} are allowed. (f) Under \emph{refinancing}, banked balance
lowers the running account. (g) An \emph{Assumption-U class restriction}
imposes a conditional-independence property that the imported
framework builds into the class definition. Whether it holds without
loss of generality is an open question there
(Sec.~\ref{sec:setting-asterisk}).

We fix two bookkeeping conventions of the imported framework once and
use them throughout. ($\alpha$) For memory-reading protocols, clause
(g) is imposed \emph{conditionally on the memory value}. Under this
convention, the per-realization and single-protocol forms of the
informed branch in Definition~\ref{def:gap} agree. An unconditional
clause (g) implies the $M$-conditional form for a frozen $M$ by
including $M$ among the quantified registers. Thus, the blind class
embeds unchanged into the memory-reading class. ($\beta$) The extracted
work $W_{\mathrm{ext}}$ is the \emph{terminal balance of the account}.
Nonequilibrium resources remaining at the horizon are forfeited. This
convention closes the loophole of parking value outside the books. It
also makes the net contribution of any imported auxiliary nonpositive
because its one-shot $D_{\max}$ price is not recovered.

The imported \emph{deletion counterfactual} is the comparison object of
this paper. Deleting a register from the access structure removes every
designable read port from that register to the mechanisms. The deprived
agent then re-optimizes \emph{from scratch} within the reduced class.
The fixed wiring of the task remains untouched. Deletion changes the
agent, not the physics~\cite{Sudo_MPU}.

\subsection{Tasks, memory, and the future task distribution}
\label{sec:setting-tasks}

\begin{definition}[Task tuple]
\label{def:task}
A \emph{task} $\tau$ is a specification consisting of world-side data
only:
\begin{equation}
\tau=\bigl(W_\tau,\ p_\tau(w),\ \{E_i\}_\tau,\ G_\tau,\
\mathrm{Man}_\tau,\ K_\tau,\ A_\tau\bigr),
\label{eq:task-tuple}
\end{equation}
where $W_\tau$ is the set of world registers, which does \emph{not}
contain the memory $M$; $p_\tau(w)$ is the distribution of the world
registers at the start of a run; $\{E_i\}_\tau$ are the energy
functions; $G_\tau$ is a (possibly empty) set of fixed gate devices;
$\mathrm{Man}_\tau$ is the manipulable set, i.e., the registers on
which designed mechanisms may act; $K_\tau$ is the horizon (maximal
number of composed steps); and $A_\tau=(\mathrm{Acc}_0,R)_\tau$ is
the access structure, which does \emph{not} include a read port on
$M$.
\end{definition}

Two exclusions in Definition~\ref{def:task} are consequential. $M$ is
not a world register. Otherwise, every update would change the task
tuple itself, and ``$\Delta\V$ under the same $T$'' would be a type
error. The $M$-read port is also not part of $A_\tau$. The
informed/blind contrast of Sec.~\ref{sec:value} is introduced by
\emph{adding} or \emph{deleting} that port. The port is therefore an
attribute of the agent, not of the task. We obtain the initial joint
distribution of a run by gluing the tuple's $p_\tau(w)$ to the
learning state of Definition~\ref{def:memory}. The bank and
residual-exergy conventions and the budget semantics are fixed
theory-wide, not per task.

\begin{definition}[Memory register and learning state]
\label{def:memory}
The memory $M$ is a separate finite register, common to all tasks,
subject to four clauses: (i) $M$ is energetically decoupled in every
task ($E_\tau$ does not depend on $M$); (ii)
$M\notin\mathrm{Man}_\tau$; (iii) no fixed gate in $G_\tau$ reads
$M$; (iv) no fixed gate in $G_\tau$---and no task-fixed channel of
any kind---acts \emph{on} $M$. A \emph{learning state} of $M$ is a
family of joint distributions
\begin{equation}
\pi_M=\bigl\{\,p_\tau(m,d,w)\ :\ \tau\in\operatorname{supp}T\,\bigr\}
\label{eq:learning-state}
\end{equation}
over $(M,D,W_\tau)$, with $D$ the external data record of
Definition~\ref{def:update}, subject to (a) \emph{world-marginal
consistency}: the $W$-marginal of $p_\tau(m,d,w)$ equals the tuple's
$p_\tau(w)$; and (b) \emph{draw exogeneity} of the training closure
(Definition~\ref{def:exogeneity}).
\end{definition}

Clauses (ii)--(iv) jointly \emph{derive} that the in-run dynamics of
$M$ is the identity ($M$ is frozen). This is not an extra assumption.
On their own, clauses (ii) and (iii) would not exclude a fixed gate
with a world register as control and $M$ as target. Such a gate could
inject value into a blank memory during a run; clause (iv) closes that
hole. The clauses have distinct roles. Clauses (ii)$+$(iv) provide
the frozen datum under which the per-realization and single-protocol
forms of the informed branch agree (with convention $\alpha$).
Clause (iii) underlies the $M$-independence of the blind branch
(Remark~\ref{rem:blind}). In the imported framework, a fixed gate
that reads the datum deliberately re-attributes value. Here, $M$ is
the protagonist of the ledger, so the re-attribution channel is removed
by definition. The evaluation of $\V$ in Sec.~\ref{sec:value} uses
only the $(M,W_\tau)$-marginals of $\pi_M$. The extended family over
$(M,D,W_\tau)$ is nevertheless the state proper; without it, the
update dynamics does not close. We leave couplings between world
variables of different tasks unspecified. Because $\V$ is a
$T$-expectation of per-task quantities, its value does not depend on
these couplings. Global quantities of the type $I(M;W)$ are used only
when all tasks in the support share one world register set (a
\emph{shared world}, as in the device LB of Sec.~\ref{sec:value}).

\emph{Convention (probability space).}---Conditional quantities such
as $p(m\mid\tau)$ and $\mathbb{E}[W_{\mathrm{ext}}\mid M=m,\tau]$ are
evaluated on the single probability space obtained by placing
$T(\tau)\,p_\tau(m,d,w)$ on the bundle
$\bigsqcup_{\tau}\{\tau\}\times(M\times D\times W_\tau)$.

\begin{definition}[Draw exogeneity]
\label{def:exogeneity}
The task draw is jointly independent of the entire training closure:
\begin{equation}
\tau\ \perp\ \bigl(M,\ D,\ \text{all ancillas used by
updates}\bigr),
\label{eq:exogeneity}
\end{equation}
equivalently, the joint marginal of $(M,D)$ together with the update
ancillas is common to all $\tau$.
\end{definition}

Independence of $M$ alone ($p(m\mid\tau)=p(m)$) would not suffice. If
the task index is drawn as $M_0\oplus D$, then $\tau\perp M_0$ and
$\tau\perp D$ hold separately. Yet the single admissible update
$M_1:=M_0\oplus D$ destroys the independence. Joint independence of
the closure is what survives updates:

\begin{lemma}[Update invariance of exogeneity]
\label{lem:exo-invariance}
Under Definition~\ref{def:exogeneity}, $\tau\perp M_k$ holds at
every point of every sequence of admissible $M$-local updates
(Definition~\ref{def:update}).
\end{lemma}

\begin{proof}
$M_k$ is a function of $(M_0,D,\text{ancillas})$: update kernels are
functions of $(m,d,a)$ only, and ancillas are independently
initialized. Hence
$\sigma(M_k)\subseteq\sigma(M_0,D,\text{ancillas})\perp\tau$.
\end{proof}

\begin{definition}[Future task distribution]
\label{def:taskdist}
$T$ is a probability distribution on a finite support
$\mathcal{T}=\{\tau_1,\dots,\tau_m\}$; the state space of $M$ is
also finite.
\end{definition}

Finiteness of the support is a choice of the present paper. Under an
extension to countable supports or general measures, the affinity in
Main Theorem~\ref{thm:main-I}(i) survives. Finiteness of the Lipschitz
constant additionally requires uniform supply boundedness
($\sup_\tau\Gap<\infty$; e.g., uniformly bounded prizes, horizons,
and initial total correlations). Any such extension would carry this
condition explicitly.

\begin{remark}[The meaning of ``future'']
\label{rem:future}
The temporal structure of the separation theorems is: update
(training) first, draw and run second. $\V(M;T,b)$ at a given update
step is the counterfactual evaluation ``if the draw happened now'',
and its soundness at every step of an update sequence is exactly
Definition~\ref{def:exogeneity} plus
Lemma~\ref{lem:exo-invariance}. Time order alone does not yield
statistical independence (in the example above, $M_1$ is fixed
before the draw); what the definition encodes is exogeneity---the
drawing mechanism does not read the training closure.
Self-selection, where the content of $M$ influences which tasks are
attempted, is outside the scope of this paper.
\end{remark}

\subsection{Evaluation budget}
\label{sec:setting-budget}

One draw $\tau\sim T$ is one run. The argument $b$ of $\V$ is
the per-run pathwise drawdown cap specified in clause (d). On every
positive-probability trajectory and at every stage, the running
account satisfies $E_k(\omega)\le b$, with auxiliaries priced one-shot
at $D_{\max}$. The costs of the update (training) process itself remain
in a \emph{separate account}. They are measured in the
acquisition-cost and dissipation currencies (C1, C2) of
Table~\ref{tab:types} and never enter $b$, which prices only the future
run side. This separation distinguishes what was paid to acquire a record
from what the record can be cashed for under a future budget. It is
the skeleton of the four-currency accounting. In the gate-resolvable
subclass appearing in Sec.~\ref{sec:regime}, paid queries made during a
run are charged to the same running account $E_k(\omega)$.

\subsection{Typed accounting: the four component currencies}
\label{sec:setting-types}

Table~\ref{tab:types} fixes the four component currencies of the
accounting. The table is a contract on types, not a classification of
learning phenomena. For each component, it fixes whether the quantity
is actual or counterfactual, its probabilistic and functional type,
its unit, and its composition rule. The contract prevents any argument
below from silently converting one component into another. The search ledger
$\sigmaM=\Delta I(M;D)+\Sigmatot$ of Sec.~\ref{sec:ledger} is
\emph{not} C2 itself. It is a memory-side subsystem account; only its
$\Sigmatot$ entry is pure dissipation. The hypotheses of each theorem
name the subset of components the theorem uses. Fig.~\ref{fig:accounting}
charts the accounts and the theorems that connect them.

\begin{table*}[t]
\caption{\label{tab:types}%
The C1--C4 type table: the four component currencies of the
accounting. The table is a contract on types, not a classification
of learning phenomena; each theorem in this paper uses an explicitly
stated subset of the components. The search ledger
$\sigmaM=\Delta I(M;D)+\Sigmatot$ (Sec.~\ref{sec:ledger}) is not C2
itself.}
\begin{ruledtabular}
\begin{tabular}{p{1.7cm}p{2.1cm}p{1.7cm}p{1.6cm}p{1.7cm}p{1.4cm}p{1.6cm}p{2.4cm}}
Component & Realization & Actual / counterfactual & Probabilistic
type & Functional type & Unit & Composition & Cautions \\
\colrule
C1 acquisition cost & K1/K2 replacement
cost~\cite{Sudo_Replacement} & counterfactual optimum & expected or
worst-path & protocol-class optimum & energy & generally nonadditive
& paid-query dissipation can overlap with C2 \\
C2 physical dissipation & marginal entropy production $\Sigmatot$ of
the implementation & actual & expected path quantity & path
functional & nats; energy equivalent $\kT\,\Sigmatot$ & additive under stage sums & $\Sigmatot$ is the
pure-dissipation entry; every energy-typed use below carries the
explicit $\kT$; the signed $\Delta I$ telescopes exactly
along update sequences, $\dpJD$ in general does not \\
C3 transport action & metric-dependent path action & actual &
pathwise or expected & path functional & metric dependent & depends
on path composition & connects to a physical cost only when a
physical metric is specified \\
C4 task value & informed--blind optimal-value difference $\V$ &
counterfactual optimum & expectation over $T$ & difference of
protocol-class optima & energy & affine in $T$ & relative to $T$,
access structure, and evaluation budget \\
\end{tabular}
\end{ruledtabular}
\end{table*}

\subsection{The U-restriction asterisk}
\label{sec:setting-asterisk}

All protocol classes in this paper carry the imported Assumption-U restriction
of clause (g), with convention ($\alpha$) for
memory-reading protocols. The imported framework leaves open whether
the restriction is without loss of generality. We therefore build it
into the class rather than list it among the hypotheses. The result is
a standing asterisk on interpretation: $\V$ is a
restricted-class--relative quantity. While the U-reduction remains
open, it may deviate from an unrestricted operational value. Every
theorem below is unconditionally valid on the restricted class. The
asterisk attaches to interpretation, and we return to it in
Sec.~\ref{sec:discussion}.

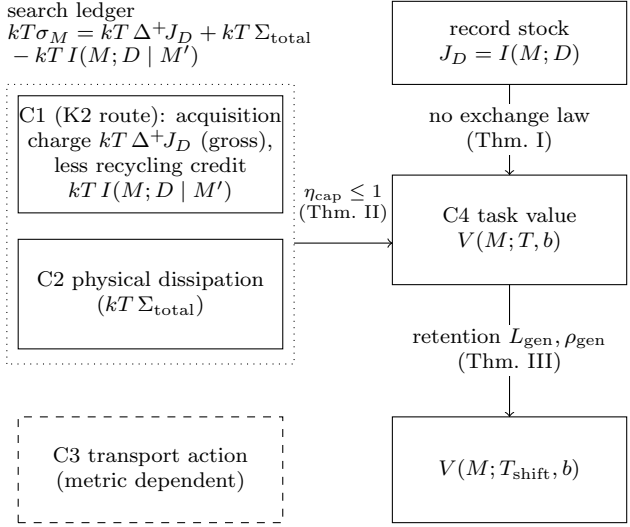
\begin{figure}[t]
\begin{tikzpicture}[every node/.style={font=\footnotesize}]
\node[anchor=north west,align=left,inner sep=0pt] at (0,7.4)
{search ledger\\[-1pt]
 $\kT\sigmaM=\kT\,\dpJD+\kT\,\Sigmatot$\\[-1pt]
 $\mathord{}-\kT\,I(M;D\mid M')$};
\draw[dotted] (0,2.6) rectangle (3.8,6.3);
\draw (0.15,4.6) rectangle (3.65,6.15);
\node[align=center] at (1.9,5.375)
{C1 (K2 route): acquisition\\ charge $\kT\,\dpJD$ (gross),\\
 less recycling credit\\ $\kT\,I(M;D\mid M')$};
\draw (0.15,2.8) rectangle (3.65,4.25);
\node[align=center] at (1.9,3.525)
{C2 physical dissipation\\ ($\kT\,\Sigmatot$)};
\draw[dashed] (0.15,0.5) rectangle (3.65,1.9);
\node[align=center] at (1.9,1.2)
{C3 transport action\\ (metric dependent)};
\draw (5.1,6.3) rectangle (8.2,7.4);
\node[align=center] at (6.65,6.85)
{record stock\\ $\JD=I(M;D)$};
\draw (5.1,3.65) rectangle (8.2,5.1);
\node[align=center] at (6.65,4.375)
{C4 task value\\ $\V(M;T,b)$};
\draw (5.1,0.5) rectangle (8.2,1.9);
\node[align=center] at (6.65,1.2)
{$\V(M;T_{\mathrm{shift}},b)$};
\draw[->] (3.8,4.2) -- (5.1,4.2);
\node[align=center,font=\scriptsize,fill=white,inner sep=1pt] at (4.45,4.72)
{$\etacap\le1$\\ (Thm.~\ref{thm:main-II})};
\draw[->] (6.65,6.3) -- (6.65,5.1);
\node[align=center,fill=white,inner sep=1pt] at (6.65,5.7)
{no exchange law\\ (Thm.~\ref{thm:main-I})};
\draw[->] (6.65,3.65) -- (6.65,1.9);
\node[align=center,fill=white,inner sep=1pt] at (6.65,2.775)
{retention $\Lgen,\rhogen$\\ (Thm.~\ref{thm:main-III})};
\end{tikzpicture}
\caption{\label{fig:accounting}%
The typed four-component accounting. Boxes are the component
currencies (Table~\ref{tab:types}) together with the record stock
$\JD$ and the shifted value; arrows denote theorems or routes that
hold between the accounts---or, where so labeled, their proven
absence---\emph{not} additions. The ledger frame displays the
universal identity in its energy-typed form,
$\kT\sigmaM+\kT\,I(M;D\mid M')=\kT\,\dpJD+\kT\,\Sigmatot$
[Eq.~\eqref{eq:ledger-universal}]: the C1 entry on the K2 route is
the \emph{gross} acquisition charge $\kT\,\dpJD$ (distinct from the
signed stock change $\kT\,\Delta\JD$, which is not an acquisition
charge), the recycling credit vanishes under condition~(f), and all
information-typed quantities enter with their explicit $\kT$
conversion. The
capitalization arrow holds in the $\flatstar+$(f) regime with the
$\sigmaM$ account (Main Theorem~\ref{thm:main-II}); the absent
$\JD\to\V$ exchange law is the separation of Main
Theorem~\ref{thm:main-I}; the retention arrow is the alignment
schema of Main Theorem~\ref{thm:main-III}. C3 connects to a
physical cost only when a physical metric is specified.}
\end{figure}

\section{Fit, record correlation, and capital value}
\label{sec:value}

\subsection{Training-side quantities}
\label{sec:value-fit}

\begin{definition}[Training-side fit, record correlation, gross
acquisition]
\label{def:fit}
A \emph{training-side fit functional} $\Phifit$ is any functional
that ranks candidate updates from the training side---training loss,
training likelihood, training work, or value on the training
distribution. The \emph{record-correlation stock} of a memory state
is
\begin{equation}
\JD(M):=I(M;D),
\label{eq:JD}
\end{equation}
with $D$ the external data record of Definition~\ref{def:update}.
The \emph{gross acquisition} of an update $M\to M'$ is
\begin{equation}
\dpJD:=I(M';D\mid M),
\label{eq:dpJD}
\end{equation}
distinguished from the \emph{signed increment}
$\dJD:=I(M';D)-I(M;D)$: the signed increment telescopes along update
sequences, the gross acquisition in general does not (under the
growth clause of Definition~\ref{def:update} the two coincide).
\end{definition}

$\JD$ is memorization, a stock of record correlation, and is not
identical to training-side fit. An early working document used the
word ``fitting'' for the predicate $\Delta I(M;D)>0$. Throughout
this paper, that predicate is called \emph{record-correlation
increase}, whereas $\Phifit$ always denotes a training-side fit
functional (Appendix~\ref{app:notation}).

\subsection{The informed--blind gap and the capital value}
\label{sec:value-def}

\begin{definition}[Per-task informed--blind gap]
\label{def:gap}
For a task $\tau$ and budget $b\ge0$,
\begin{equation}
\begin{split}
\Gap(M;\tau,b):={}&\mathbb{E}_{m\sim p(M)}\Bigl[\\[-2pt]
&\quad\sup_{P\in\mathrm{Prot}_b(A_\tau+M,\,p)}
\mathbb{E}\bigl[W_{\mathrm{ext}}\bigm|M=m,\tau\bigr]\Bigr]\\[2pt]
&-\sup_{P\in\mathrm{Prot}_b(A_\tau-M,\,p)}
\mathbb{E}\bigl[W_{\mathrm{ext}}\bigm|\tau\bigr],
\end{split}
\label{eq:gap}
\end{equation}
where $A_\tau+M$ is the access structure with a read port on $M$
added, and $A_\tau-M$ is the access structure with every designable
read port on $M$ deleted, in the sense of the imported deletion
counterfactual (Sec.~\ref{sec:setting-imports}).
\end{definition}

The informed branch is per-realization. The protocol may be
re-chosen for each memory value $m$. Because $M$ is frozen during
runs (derived from Definition~\ref{def:memory}(ii)$+$(iv)), the
imported calibration lemma, together with convention ($\alpha$),
identifies this branch with the supremum over single protocols
carrying an $M$-read port. The blind branch re-optimizes with full
knowledge of the public task tuple. Only the read port on $M$ is
removed.

\begin{remark}[The blind branch does not depend on $M$]
\label{rem:blind}
A blind protocol has no read port on $M$, and by
Definition~\ref{def:memory}(i)--(iv) the memory enters neither the
energetics nor the dynamics; every stage map of a blind protocol
therefore factorizes as $\mathrm{id}_M\otimes\Lambda_W$, and its
$\mathbb{E}[W_{\mathrm{ext}}]$ is a functional of the non-$M$
marginals alone. The remaining conceivable channel---importing an
auxiliary whose initial state is correlated with $M$---is closed not
by pricing but by type: the coupling at $t_0$ is task data rather
than an object of protocol choice, and mid-run auxiliaries are
adjoined as fresh registers in product state; no preparation map
acting on $M$ exists in the class. Hence $\Wblind$ carries no $M$
argument. The full argument is reproduced in
Appendix~\ref{app:equality}.
\end{remark}

\begin{definition}[Capital value]
\label{def:value}
\begin{equation}
\begin{split}
\V(M;T,b)&:=\mathbb{E}_{\tau\sim T}\bigl[\Gap(M;\tau,b)\bigr]\\
&=\Winformed(M;T,b)-\Wblind(T,b),
\end{split}
\label{eq:value}
\end{equation}
where $\Winformed$ and $\Wblind$ are the $T$-expectations of the two
branches of Eq.~(\ref{eq:gap}). The unit is work; \emph{bits of
value} means $\V/(\kT\ln2)$. $\V$ carries the U-restriction asterisk
of Sec.~\ref{sec:setting-asterisk}.
\end{definition}

Measuring the value of information by an expected work difference on
tasks lies, by itself, in the lineages of the value of information
and semantic information. What is specific to $\V$ is the
counterfactual type (deletion followed by re-optimization from
scratch), the budget and access arguments, and the affine structure
in $T$. The theorems below are built on this affine structure. The
attributions, and the exact differences from the scramble-type
counterfactual, are tabulated in Sec.~\ref{sec:related}.

\subsection{Admissible updates and the learning predicate}
\label{sec:value-updates}

\begin{definition}[Admissible updates: $M$-local updates]
\label{def:update}
An \emph{update} is a physical process whose kernel is a function of
$(m,d,a)$ only; in particular it neither reads nor writes any world
register (this excludes the weaker reading ``does not act on $W$'').
Here: (i) $D$ is the \emph{external data record} (the training
record), a register that may carry a coupling $p_\tau(d,w)$ with the
world registers and lies outside $W_\tau$ for every $\tau$; during
updates $D$ is \emph{read-only}. (ii) Ancillas $A$ are initialized
in a product state \emph{jointly} independent of $(W,M,D)$ (fresh
randomness), a condition imposed per update. (iii) After an update
the memory may be re-designated as $M'=(M,\text{absorbed ancilla
registers})$ (growth clause).
\end{definition}

The two type restrictions in (i) and (ii) each close a specific
hole. Read-only $D$ closes the first hole. Allowing writes to $D$
would let an update copy memory content into empty cells of $D$ and
register $\Delta I(M;D)>0$ without acquiring anything. Allowing
erasure of $D$ would hide perfect memorization as
$\Delta I(M;D)\le0$. The record-correlation predicate would fail in
both directions. (The accounting of legitimate record rewriting and
forgetting of $D$ is a future slot.) Joint independence of ancillas
closes the second hole. Marginal independence would not suffice:
$A:=M\oplus W$ is marginally independent of $W$, yet the update
$M':=M\oplus A=W$ would break data processing. Joint independence
excludes such key composition. The condition is imposed per update.
Along a sequence, earlier stages may create registers correlated
through $D$, which is the designed channel.

\begin{lemma}[Two consequences of the update class]
\label{lem:update-consequences}
For every admissible update and every task $\tau$ in the support:
(1) (\emph{blind invariance}) the joint law of the world registers
is unchanged; hence $\Wblind(T,b)$ is invariant under updates and
$\Delta\V=\Delta\Winformed$.
(2) (\emph{per-update data processing})
$I(M';W_\tau)\le I\bigl((M,D);W_\tau\bigr)$; on a shared-world
support the same bound holds for the full world register set $W$.
\end{lemma}

\begin{proof}
(1) The kernel acts on $(M,D,A)$ only; in a classical system without
post-selection, a local operation leaves the marginal of the
untouched registers unchanged. (2) $M'$ is a function of $(M,D,A)$
with $A$ jointly independent of $(W,M,D)$, so
$W_\tau\to(M,D)\to M'$ is a Markov chain and the data-processing
inequality applies.
\end{proof}

\begin{definition}[Learning; record-correlation increase]
\label{def:learning}
An admissible update is \emph{learning} (\emph{capitalization}) if
\begin{equation}
\Delta\V:=\V(M';T,b)-\V(M;T,b)>0 .
\label{eq:learning}
\end{equation}
It is a \emph{record-correlation increase} if $\Delta I(M;D)>0$
(meaningful because $D$ is read-only, hence identical in content
before and after the update).
\end{definition}

These are two predicates on the same update, relative to the stated
task distribution, access structure, and evaluation budget. They are
neither exclusive nor exhaustive. Both can hold, which is the
typical intent of honest training, and both can fail. They are a
pair of predicates, not a chain-rule decomposition. Main
Theorem~\ref{thm:main-I}(ii)--(iii) shows that they separate at the
theorem level.

\begin{remark}[Value accounting, not computation accounting]
\label{rem:scope}
$\V$ prices what a readable memory is worth on future tasks; it does
not price computation. In the informed branch, run-time protocol
design and the functional composition of reads are free. In the
fully proved $\flatstar$ regime, an update that merely re-encodes
existing correlation into a more usable form gains no $\V$ (Main
Theorem~\ref{thm:main-I}(iv)); for general tasks the same statement
is a heuristic supported only by the proof sketch of
Sec.~\ref{sec:value-blackwell} and is not asserted as a theorem.
The value of representation learning, compression, distillation, and
consolidation---of making the same correlation easier to use---is
not what $\V$ measures, for a reason independent of that scope
limit: $\V$ does not price the computation such updates save. That
value belongs to an accounting for computation-bounded agents,
outside the scope of this paper.
\end{remark}

\subsection{Main Theorem I}
\label{sec:value-main}

\begin{maintheorem}[Operational capital value and separation]
\label{thm:main-I}
Let the setting be that of Sec.~\ref{sec:setting}, with $q=p$.
\begin{enumerate}
\item[(i)] (\emph{Basic properties, invariance, correspondence.})
$\V$ is well posed: each branch of Eq.~(\ref{eq:gap}) is a finite
supremum of a scalar functional over a fixed class, and the outer
expectations are finite sums. Moreover:
(a)~$\V(M;T,b)\ge0$;
(b)~$T\mapsto\V(M;T,b)$ is an affine functional of the distribution,
with
\begin{equation}
\begin{aligned}
\lvert\V(M;T,b)-\V(M;T',b)\rvert
&\le\lVert T-T'\rVert_{\mathrm{TV}}\\
&\quad\times\sup_{\tau}\Gap(M;\tau,b),
\end{aligned}
\label{eq:tv-bound}
\end{equation}
where
$\lVert T-T'\rVert_{\mathrm{TV}}:=\tfrac12\sum_\tau
\lvert T(\tau)-T'(\tau)\rvert$;
(c)~if $A_\tau\preceq A'_\tau$ for all $\tau$---componentwise:
$\mathrm{Acc}_0\subseteq\mathrm{Acc}'_0$, and $R'$ licenses supersets
on identical record histories---then both the informed and the blind
branch are nondecreasing; no monotonicity is claimed for $\Gap$
itself;
(d)~for every bijection $\varphi$ of the finite state space of $M$,
with learning state
$\pi_{\varphi(M)}:=(\varphi\otimes\mathrm{id})_{*}\pi_M$,
\begin{equation}
\begin{aligned}
\V(\varphi(M);T,b)&=\V(M;T,b),\\
\Gap(\varphi(M);\tau,b)&=\Gap(M;\tau,b)\qquad\forall\tau
\end{aligned}
\label{eq:gauge}
\end{equation}
(scope: bijective recodings of a finite register; this is an
invariance of the C4 account and is not asserted for the C1 or C2
accounts);
(e)~for $T=\delta_\tau$,
\begin{equation}
\V(M;\delta_\tau,b)=\kT\ln2\cdot M_{\mathrm{val}}(M;b,\tau,A_\tau),
\label{eq:mpu-anchor}
\end{equation}
the informed--blind datum value of the companion
framework~\cite{Sudo_MPU} in bits, on the common domain of data
conforming to Definition~\ref{def:memory}.
\item[(ii)] (\emph{Separation.}) For every $n\in\mathbb{N}$ there
exist a task distribution $T_n$, a budget $b\ge0$, a data record
$D_n$, and a sequence $M_0\to M_1\to\dots\to M_n$ of admissible
$M$-local updates such that $I(M_k;D_n)$ is strictly increasing in
$k$ with $I(M_n;D_n)=n\ln2$, while
\begin{equation}
\V(M_k;T_n,b)=\V(M_0;T_n,b)\qquad(\forall\,k\le n).
\label{eq:separation}
\end{equation}
\item[(iii)] (\emph{Strengthened separation.}) The family in (ii)
can be chosen with all tasks in the support sharing one world
register set $W$, such that in addition $I(M_k;W)$ increases by
$n\ln2$ along the sequence while Eq.~(\ref{eq:separation}) still
holds.
\item[(iv)] (\emph{$\flatstar$ D-free nonincrease.}) For updates
whose kernel is a function of $(m,a)$ alone (\emph{D-free}), and for
every task $\tau$ that is $\flatstar$ relative to the pre-update
state (Definition~\ref{def:flatstar}, Sec.~\ref{sec:ledger}),
$\Gap(M';\tau,b)-\Gap(M;\tau,b)\le0$; hence $\Delta\V\le0$ whenever
$T$ is supported on $\flatstar$ tasks.
\end{enumerate}
\end{maintheorem}

\begin{proof}[Proof of (i)]
\emph{Well-posedness and (a).} Each branch is the supremum of a
scalar functional over a fixed protocol class, with the bank,
residual-exergy, and budget conventions fixed theory-wide; the
suprema are finite because each task has finite supply (finite
prizes, finite horizon $K_\tau$, finite total correlation at
$t_0$), as in the imported framework, and the support and the state
space of $M$ are finite (Definition~\ref{def:taskdist}). For (a):
with $q=p$ the blind class embeds into the informed class by
ignoring the $M$ port, so $\Gap(M;\tau,b)\ge0$ for every $\tau$, and
$\V\ge0$ follows by averaging.

\emph{(b).} Affinity is
$\V(M;T,b)=\sum_\tau T(\tau)\,\Gap(M;\tau,b)$; the bound is
total-variation duality together with $\Gap\ge0$ from (a). We record
Eq.~(\ref{eq:tv-bound}) as a basic property, not as a theorem of
independent content: what it is worth on a concrete shift family
depends on the inner-product structure between the gap profile
$\tau\mapsto\Gap(M;\tau,b)$ and the measure movement, which is the
subject of Sec.~\ref{sec:retention}.

\emph{(c).} The record history of a fixed protocol is unchanged
under access enlargement, and by induction over stages the
qualification of each record (a write from a readable source) is
preserved, so every parent-set constraint
$\mathrm{Pa}\subseteq\mathrm{Acc}_k$ satisfied under $A$ is
satisfied under $A'$; each branch's class only grows, and its
supremum is nondecreasing. The two branches may profit
unequally---there are enlargements from which only the blind branch
profits---which is why no claim is made for $\Gap$.

\emph{(d).} The blind branch does not depend on $\pi_M$
(Remark~\ref{rem:blind}) and is invariant. For the informed branch,
map $P\mapsto P^{\varphi}$ by replacing the parent-dependent
function $f(m,\dots)$ of every mechanism reading $M$ with
$f(\varphi^{-1}(m'),\dots)$; the imported class admits arbitrary
parent-dependent functions in a single step and does not charge
reads within a stage as adaptivity, so $P^{\varphi}$ is admissible
with the same depth. Couple the trajectories of $P$ (memory value
$m$) and $P^{\varphi}$ (memory value $\varphi(m)$): if all record
values agree up to stage $k$, the mechanisms of $P^{\varphi}$
compute $f(\varphi^{-1}(\varphi(m)))=f(m)$, so the record values
remain equal; the access rule $R$ sees identical histories and
licenses identical $\mathrm{Acc}_k$; the parent-set constraints hold
on one side iff they hold on the other; and the agreement propagates
to stage $k+1$. By induction, the trajectory laws of
$W_{\mathrm{ext}}$ and of the running account
$E_k(\omega)$---invested work, one-shot $D_{\max}$ prices,
withdrawals, banked work---coincide. Class membership is also
preserved: the imported U restriction is a family of vanishing
conditions on conditional mutual informations, i.e., properties of
joint distributions rather than of channels; such conditions are
invariant under bijective relabeling of an argument,
$\sigma(\varphi(M))=\sigma(M)$, and the per-realization conditional
laws coincide under the coupling. Since $P\mapsto P^{\varphi}$ is a
bijection between the classes, the suprema agree, and the outer
expectation transforms by the change of variables $m'=\varphi(m)$.

\emph{(e).} With $T=\delta_\tau$ the expectation collapses to
$\Gap(M;\tau,b)$. The informed branch of Definition~\ref{def:gap} is
the per-realization form and the blind branch is deletion followed
by re-optimization from scratch; this is the companion framework's
informed--blind datum value with the datum read as $M$, normalized
there by $\beta/\ln2$ to bits, so the energy-unit gap equals
$\kT\ln2\cdot M_{\mathrm{val}}$. Two design choices carry the
identity: the per-realization form of the informed branch (a
single-protocol form would agree only through the calibration
lemma), and per-task re-optimization in the blind branch (a
task-agnostic blind would not reduce to the companion definition).
\end{proof}

\begin{proof}[Proof of (ii) and (iii): the device family LB]
The witness is the \emph{ledger-blocked correlation} device LB,
defined in full in Appendix~\ref{app:devices}; we summarize the
construction and the two blocking conditions. Registers: $X_1,X_2$,
one uniform independent bit each (the working media of two tasks);
$Y=(Y_1,\dots,Y_n)$, i.i.d.\ fair bits independent of $(X_1,X_2)$
(world registers outside every manipulable set); and
$M_{\mathrm{core}}$, a perfect copy of $X_1$ (one bit of genuine
capital, so that $\Delta\V=0$ is exhibited in a nondegenerate
situation with $\V>0$). All energies are degenerate ($E\equiv0$).
The support is $\mathcal{T}=\{\tau_1,\tau_2\}$: $\tau_1$ has the
shared world $W=(X_1,X_2,Y)$, uniform product $p(w)$, $G=\emptyset$,
$\mathrm{Man}=\{X_1\}\cup\text{ancillas}$, $K_\tau=3$, and static
all-read access; $\tau_2$ is identical with
$\mathrm{Man}=\{X_2\}\cup\text{ancillas}$; $T_n=(1/2,1/2)$; and
$b\ge0$ is arbitrary---the informed protocols require no upfront
investment, so the budget axis is deliberately not exercised
(budget-driven devices appear in Sec.~\ref{sec:regime}). The data
record is $D:=$ the observation record of $Y$ ($D_j$ a copy of
$Y_j$, read-only during updates), and the update sequence is
$M_k:=(M_{\mathrm{core}},\text{copies of }Y_1,\dots,Y_k)$, each step
a single controlled-copy from $D_k$ into fresh memory---an
admissible $M$-local growth update.

For each $\tau_i$ the per-task gap is pinned exactly: the informed
supremum equals $\kT[\ln\lvert X_i\rvert-H(X_i\mid
M_k)]=\kT\,I(M_k;X_i)$ and the blind supremum equals $0$. The
converse (no admissible protocol exceeds these values) is carried by
a three-step argument reproduced in Appendix~\ref{app:devices}: a
conditional chain bound of the imported framework applied to the
$M=m$ conditional law; invariance of the non-manipulable marginals
(mechanism locality with $G=\emptyset$); and auxiliary netting under
the terminal-balance convention ($\beta$). Achievability is a
conditional permutation ($X_1$ conditioned on $M_{\mathrm{core}}$)
followed by Szilard extraction. Hence
\begin{equation}
I(M_k;D)=k\ln2,\qquad I(M_k;W)=(k+1)\ln2,
\label{eq:lb-correlations}
\end{equation}
\begin{equation}
\Gap(M_k;\tau_1,b)=\kT\ln2,\qquad \Gap(M_k;\tau_2,b)=0,
\label{eq:lb-gaps}
\end{equation}
so $\V(M_k;T_n,b)=\tfrac12\,\kT\ln2$ for every $k$: along the
sequence, $I(M_k;D)$ rises strictly to $n\ln2$ and $I(M_k;W)$ rises
by $n\ln2$, while $\Delta\V=0$. This proves (ii) and (iii).

The blockage is the conjunction of two independent conditions:
(i)~$X_i\perp Y$---the acquired correlation targets registers
independent of the working media, which closes the laundering route
(using $M$--$Y$ correlation to reset $X_i$ cheaply and re-extract);
and (ii)~$Y\notin\mathrm{Man}_\tau$ for every $\tau$ in the
support---the cashing route is severed. Read through condition (i),
LB is a thermodynamic implementation of the classical statement that
memorizing environmental noise creates no value; read through condition
(ii), the correlation is genuine correlation with world variables of
every task in the support, and it is the absence of any manipulation
route that blocks its capitalization. A cartridge variant, in which
$Y$ is replaced by keys whose gates are absent from the support,
realizes the same blockage through gate absence
(Appendix~\ref{app:devices}). Numerical verification of the exact
gap values, of the achievability transformation, and of saturation
within the read family is reported in Appendix~\ref{app:devices}.
\end{proof}

\begin{proof}[Proof of (iv)]
By the $\flatstar$ extraction identity
(Lemma~\ref{lem:flat-identity}, Sec.~\ref{sec:ledger}) the
pre-update informed branch equals
$\kT[\ln\lvert X_\tau\rvert-H(X_\tau\mid M,Y_\tau)]$, while the
informed-branch converse---which uses none of the achievability
normalization---bounds the post-update informed branch by
$\kT[\ln\lvert X_\tau\rvert-H(X_\tau\mid M',Y_\tau)]$; the blind
branch is unchanged under updates
(Lemma~\ref{lem:update-consequences}(1)) and cancels in the
difference, so
$\Gap(M';\tau,b)-\Gap(M;\tau,b)\le\kT[\,I(M';X_\tau\mid
Y_\tau)-I(M;X_\tau\mid Y_\tau)\,]$. The chain rule
gives $I(M';X\mid Y)\le I(M;X\mid Y)+I(M';X\mid M,Y)$ (writing
$X:=X_\tau$, $Y:=Y_\tau$), and conditionally on $(M,Y)$ the D-free
update output $M'=f(M,A)$ is independent of $X$: the ancilla is
jointly independent of $(W,M,D)$, and joint independence is
inherited under conditioning (Appendix~\ref{app:equality}). Hence
$I(M';X\mid M,Y)=0$, so
$\Gap(M';\tau,b)-\Gap(M;\tau,b)\le0$; averaging over a
$\flatstar$-supported $T$ gives $\Delta\V\le0$.
\end{proof}

\begin{remark}[Quantifier structure of the separation]
\label{rem:quantifier}
Parts (ii)--(iii) are stated in the weak form
$\forall n\ \exists\,(T_n,D_n,\text{sequence})$: the environment
register $Y$ has $n$ bits, so the support and the record depend on
$n$ (the finiteness required by Definition~\ref{def:taskdist}
constrains the number of tasks, not the size of $Y$). The strong
form---a single $(T,D)$ carrying an infinite update sequence with
$I(M;D)\uparrow\infty$---would require countably many registers, an
extension of the imported class that we do not make; it is not
claimed.
\end{remark}

\begin{remark}[Scope of the recoding invariance]
\label{rem:gauge-scope}
Part (i)(d) covers bijective recodings of a finite register.
Reparametrizations of a continuous parameter space composed with
discretization merge and split cells and induce no bijection; they
are outside its scope (a continuous single-step map class is not
part of the imported framework). Independently of the lemma-level
statement, a structural fact holds: no parameter coordinate appears
in the definition of $\V$ at all---the objects entering
Definition~\ref{def:value} are the task tuple, the learning state,
and the budget. The invariance also depends on the design of
Definition~\ref{def:memory} ($M$ read-only, outside every
manipulable set, energetically decoupled, unpriced): extensions that
let $M$ be manipulated, or that charge its residual exergy, would
make a uniform criterion representation dependent.
\end{remark}

\begin{remark}[What part (i)(e) does and does not assert]
\label{rem:mpu-scope}
The companion value $M_{\mathrm{val}}$ is defined for arbitrary
identification subsystems, including data read by fixed gates, data
carrying energy, manipulable data, and non-frozen data. The
single-task slice of $\V$ agrees with $M_{\mathrm{val}}$ restricted
to data conforming to Definition~\ref{def:memory}(i)--(iv). $\V$ is
therefore not a generalization of the companion quantity but an
extension along the task axis combined with a restriction along the
datum axis; the four restricted directions are deliberate design
choices matching the character of a learned memory (the protagonist
of the ledger, frozen during runs, energetically decoupled).
\end{remark}

\begin{remark}[$T$-relativity of capital; the affine anchor]
\label{rem:t-relativity}
$\V$ is a functional of $T$: no absolute value of a memory exists in
this accounting. The device LB exhibits this concretely: append
$\tau_3$ (extraction from $\mathrm{Man}=\{Y\}$, where
$Y=(Y_1,\dots,Y_n)$ is read as a single bundled register over the
alphabet $\{0,1\}^n$, keeping the flat single-manipulable-register
scope) to the
support with $T'=(1/3,1/3,1/3)$; then
$\Gap(M_k;\tau_3,b)=\kT\,I(M_k;Y)=k\,\kT\ln2$, so the same update
sequence that capitalizes nothing under $T$ capitalizes every step
under $T'$ (Appendix~\ref{app:devices}). Whether an update is
learning is not an intrinsic property of the update; it is a
relation to the future task distribution. We emphasize that
``correlation outside the support of $T$ is not capital'' is a
theorem-level consequence here, not a restatement of the
definition: by part (i)(e) the single-task values are anchored in an
independently motivated operational counterfactual, and by part
(i)(b) together with draw exogeneity the evaluation is the
risk-neutral (affine) average of the anchored values, so $\V$ is the
unique $T$-affine extension of the anchor. The risk-neutral choice
itself is a design decision: a risk-averse functional (a minimum or
a conditional value-at-risk over tasks) would assign positive value
to insurance-type correlations; such variants are an explicit
non-goal of this paper.
\end{remark}

\subsection{Discussion: D-free updates on general tasks}
\label{sec:value-blackwell}

Part (iv) of Main Theorem~\ref{thm:main-I} is stated and proved,
for the $\flatstar$ regime. For general tasks, we give a proof
sketch and deliberately do not promote it to a theorem. In the
informed branch, write $\mathbb{E}_m[g_b(\rho_m)]$ with
$g_b(\rho):=\sup_{P\in\mathrm{Prot}_b}\mathbb{E}_\rho
[W_{\mathrm{ext}}]$ and $\rho_m$ the world distribution conditioned
on $M=m$. The functional $g_b$ is convex in the initial
distribution. For each fixed $P$, the corresponding functional is linear.
The supremum is taken over an admissible class that shrinks under
support enlargement through the pathwise cap. Mixing therefore lowers
the supremum. A D-free update is a garbling of $M$ with respect to the
world. Under a product-state ancilla, the conditional world distributions
of $M'$ are mixtures of those of $M$. Jensen's inequality therefore gives
$\mathbb{E}_{m'}[g_b(\rho_{m'})]\le\mathbb{E}_m[g_b(\rho_m)]$.
The blind branch remains unchanged. Accordingly, the statement
``updates that do not touch data are not learning'' is asserted as a
theorem in this paper only over the fully proved $\flatstar$ scope
of part (iv); on general tasks its status is that of this sketch.
Active updates, which intervene on the world itself, lie outside the
update class of Definition~\ref{def:update} altogether. They
correspond to a fourth axis (activity, cf.\
Ref.~\cite{Fiderer2025WorkCapacity}) recorded as a future slot.

\section{Capitalization ledger}
\label{sec:ledger}

\subsection{Update implementations and the search ledger}
\label{sec:ledger-sigma}

The ledger prices \emph{implementations} of updates.

\begin{definition}[Physical implementation of an update]
\label{def:implementation}
An \emph{implementation} of an $M$-local update $M\to M'=f(M,D,A)$
(Definition~\ref{def:update}) is a physical process on the registers
$(M,D,A)$ and a single heat bath at temperature $T$ such that:
(i)~(\emph{logical consistency}) the total effect of the process on
$(M,D,A)$ equals the update kernel---a conditional law
$K(m'\mid m,d,a)$; the two-time coupling
$p_\tau(m,m',x,y,d)$ is the pushforward of the initial learning
state $\pi_M$ and $p(a)\,K$, hence is determined by the
\emph{logical specification} (kernel and $\pi_M$) and does not
depend on the implementation.
(ii)~$D$ is read-only: its marginal and its content are unchanged,
and it participates only as a control.
(iii)~the ancilla $A$ is initialized fresh (jointly independent, in
product state; Definition~\ref{def:update}) and has exactly two
admissible terminal fates: \emph{absorption}---re-designation as
part of $M'$ under the growth clause---or \emph{erasure}, a priced
reset to a standard state whose heat enters
$\Delta S_{\mathrm{bath}}$. \emph{Totality convention}: the $M'$
appearing in the ledger statements below and in $\V(M')$ is one and
the same register, namely every memory-side register left at the
end of the process---$M$ together with all absorbed
ancillas---\emph{excluding} the read-only source $D$, the world $W$,
and the bath ($D$ and $W$ also end outside the bath, but as control
and spectator respectively, not as parts of $M'$; the terminal system
is the disjoint pair $(M',D)$); a split reading (``ledger over the
residuals, value over a declared part'') is not admitted.
(iv)~bath contacts follow standard stochastic thermodynamics: the
bath is initially equilibrated, initially uncorrelated with all
registers (including the world), and memoryless (each elementary
process satisfies local detailed balance), with no external feedback
controller; where the event-ledger decomposition of the companion
framework~\cite{Sudo_PaperI} applies, it is inherited.
The \emph{entropy production} of the implementation is the marginal
quantity
\begin{equation}
\Sigmatot:=\Delta S(M,D,A)+\Delta S_{\mathrm{bath}}\ \ge\ 0 .
\label{eq:sigmatot}
\end{equation}
\end{definition}

Nonnegativity in Eq.~(\ref{eq:sigmatot}) follows from the second law
for the marginal process. Neither the kernel nor the implementation
reads $W$. The $(M,D,A)$-marginal therefore follows the same Markov
evolution whatever its initial correlation with the world, so the
standard relative-entropy argument applies. The $W$-correlation
remains a spectator subject to data processing. Two cautions attach
to the definition. The first concerns the layer structure. The
residual subjects of the decomposition below (Lemma~\ref{lem:decomp})
and the acquisition $\dpJD$ are functionals of the two-time coupling
and hence properties of the \emph{logical specification}. Only the
reversibility entry $\Sigmatot=0$ is a property of the
\emph{implementation}. The second is a naming caution. $\Sigmatot$
is the \emph{marginal} entropy production of $(M,D,A)$ plus bath. The
global entropy production including the world adds the
\emph{destroyed} world correlation. This term is the time difference
of the mutual informations, not their initial total:
\begin{equation*}
\begin{aligned}
\Sigma_{\mathrm{global}}
&=\Sigmatot+\bigl[I_0(MDA;W)-I_1(M'D;W)\bigr]\\
&\ge \Sigmatot ,
\end{aligned}
\end{equation*}
with $I_0$ evaluated on the initial state and $I_1$ on the terminal
state. The destroyed world correlation and its initial
total coincide only when the correlation is destroyed in full,
$I_1=0$. An identity update, with nothing destroyed, has
$\Sigma_{\mathrm{global}}=\Sigmatot$. The theorems below are
internally consistent in the marginal version. The global
quantity is not used in this ledger. An erasure of memory correlated
with the world can have $\Sigmatot=0$; the loss then appears on the
value side, as the forgetting subject $g_a$ of
Lemma~\ref{lem:decomp}.

\begin{definition}[Search ledger]
\label{def:sigmaM}
The \emph{search ledger} of an implementation is
\begin{equation}
\sigmaM:=\bigl[S(M')-S(M)-S(A)\bigr]+\Delta S_{\mathrm{bath}}
\qquad(\text{nats}) .
\label{eq:sigmaM-def}
\end{equation}
The first bracket is the effective entropy change of the memory-side
registers, with the initial entropy $S(A)$ of the imported
ancillas---absorbed or erased---deducted as imported resource (a
fresh ancilla carrying uniform bits may not inflate $S(M')$ for
free; for an erased ancilla the deduction cancels against the
Landauer heat of its reset, so a borrowed page returned blank is not
billed); the second term is the heat
discharged to the bath. Along an update sequence
$M_0\to\dots\to M_n$ the ledger is read as the stage sum
$\Sigma_{\mathrm{search}}^{(n)}:=\sum_k\sigmaM^{(k)}$, and likewise
for $\Delta I$ and $\Sigmatot$.
\end{definition}

Three facts justify taking $\sigmaM$, rather than $\Sigmatot$, as
the denominator account. (1)~It abstracts the memory-subsystem
ledger $\Delta S(\omega)+\Delta Q$ of
Ref.~\cite{Goldt2017PRL}. This makes the efficiencies directly
comparable (Table~\ref{tab:correspondence}). (2)~In Bennett
accounting, a blank, low-entropy register is a thermodynamic
resource, and $\sigmaM$ charges its consumption. A reversible
controlled-copy with $\Sigmatot=0$ still spends one blank page and
is billed $\sigmaM=\ln2$ (Sec.~\ref{sec:regime}). (3)~The ledger
identity below has the energy-typed four-term form
$\kT\sigmaM+\kT\,I(M;D\mid M')=\kT\,\dpJD+\kT\,\Sigmatot$. It
embeds the search ledger into the type table as the \emph{gross}
Landauer acquisition charge $\kT\,\dpJD$ (the K2-route entry of C1)
plus the pure dissipation $\kT\,\Sigmatot$ (C2), less the recycling
credit $\kT\,I(M;D\mid M')$, which vanishes under condition~(f).
The C1 entry is the gross charge, not the signed increment
$\Delta I$ (which can be negative and is not an acquisition charge).
Each information-typed term carries its explicit $\kT$. This
identity is the form quoted
in Secs.~\ref{sec:intro} and~\ref{sec:setting} and in
Appendix~\ref{app:notation}.

\begin{remark}[The choice of account decides the theorem]
\label{rem:ledger-choice}
The reading of the denominator is fixed here once, and the theorems
of this section are sensitive to it: with the denominator read as
$\Sigmatot$, the capitalization efficiency below admits no upper
bound (route B1 of Sec.~\ref{sec:regime}); read as $\sigmaM$, the
bound $\etacap\le1$ is a theorem on the $\flatstar$$+$(f) regime
(Main Theorem~\ref{thm:main-II}); and without~(f) the bound
reappears for blank-start cumulative ledgers (Main
Theorem~\ref{thm:main-II}(v)). Recording this sensitivity is part of
the content of the operational definition.
\end{remark}

\subsection{Capitalization efficiency}
\label{sec:ledger-eta}

\begin{definition}[Capitalization efficiency]
\label{def:etacap}
For an $M$-local update (or update sequence) with implementation,
and $\Delta\V:=\V(M';T,b)-\V(M;T,b)$,
\begin{equation}
\etacap:=\frac{\Delta\V}{\kT\,\sigmaM}
\label{eq:etacap-def}
\end{equation}
(dimensionless; $\sigmaM$ in nats), defined on updates with
$\sigmaM>0$.
\end{definition}

The boundary and the outside of the domain are organized according to
condition~(f) of Lemma~\ref{lem:ledger-identity} below. For updates
satisfying~(f) (no discarded record correlation; automatic under the
growth clause), $\sigmaM=\dpJD+\Sigmatot\ge0$. If $\sigmaM=0$, then
$\dpJD=0$ and $\Sigmatot=0$, and the acquisition cap of
Lemma~\ref{lem:updateDP} gives $\Delta\V\le0$. This is a harmless
boundary on which efficiency is not in question. For updates
violating~(f) (reversible recycling of stored record correlation),
$\sigmaM<0$ is possible, as is $\sigmaM=0$ with $\Delta\V>0$. This is
not a defect of the definition; it is the signature of the recycling
regime (route B4 of Sec.~\ref{sec:regime}). Per-update statements about
$\etacap$ are therefore read within~(f). Updates with $\Delta\V<0$
(forgetting, impairment) are assigned $\etacap<0$. The numerator and
denominator both have units of energy ($\kT\times$ nats). For a
display in bits, divide both by $\kT\ln2$.

\begin{remark}[Management of the reading ``efficiency'']
\label{rem:etacap-naming}
The name licenses a ratio $\le1$ only inside the regime
$\flatstar+\sigmaM+(\mathrm f)$ (Main
Theorem~\ref{thm:main-II}(iii)). In the other cells of the regime
map (gate tasks with finite budget; $\Sigmatot$ accounting;
recycling) $\etacap$ can exceed $1$ and is read as a
\emph{conversion multiplier}---the same discipline by which
Ref.~\cite{Kolchinsky2018Semantic} distinguishes a multiplier
$\kappa$ that may exceed one from a ratio $\eta\le1$
(Table~\ref{tab:correspondence}). The numerator $\Delta\V$ is a
restricted-class--relative quantity
(Sec.~\ref{sec:setting-asterisk}), so the \emph{interpretation} of
$\etacap$ and of its equality conditions (``full capitalization'')
carries the same standing asterisk; the validity of the theorems on
the restricted class is unconditional.
\end{remark}

\subsection{The $\flatstar$ class and the extraction identity}
\label{sec:ledger-flat}

\begin{definition}[$\flatstar$ tasks]
\label{def:flatstar}
A task $\tau=(W_\tau,p_\tau,\{E_i\},G_\tau,\mathrm{Man}_\tau,
K_\tau,A_\tau)$ is \emph{$\flatstar$} (relative to a learning state
$\pi_M$) if:
(F1)~\emph{degenerate energies}: $E_i\equiv0$ for all variables;
(F2)~\emph{no gates}: $G_\tau=\emptyset$;
(F3)~\emph{single manipulable register}:
$\mathrm{Man}_\tau=\{X_\tau\}\cup\text{ancillas}$ with $X_\tau$ a
single register over a finite alphabet (possibly non-binary), and
$Y_\tau:=W_\tau\setminus\{X_\tau\}$ (possibly several units);
(F4)~\emph{static all-read access}: $A_\tau=(\mathrm{Acc}_0=
\text{all variables},\,R\ \text{trivial})$;
(F5$'$)~\emph{flat-conditional attainability}: under $\pi_M$ and
$p_\tau$, every conditional distribution $p(x_\tau\mid y,m)$ and
$p(x_\tau\mid y)$ is a point mass or a uniform distribution on a
subset of the alphabet, and $K_\tau\ge3$.
\end{definition}

Condition (F5$'$) is a joint property of $(\pi_M,p_\tau)$, not of
the task tuple alone. It closes achievability
\emph{exactly, in finitely many stages, at every budget $b\ge0$,
with pathwise nonnegative work on every trajectory}
(Appendix~\ref{app:equality}). A general dyadic condition would not
suffice. Through the fluctuation-theorem equality constraint on exact
blind attainment, a conditional atom deeper than uniform (say $1/8$
inside $p(x\mid y)=(1/2,1/4,1/8,1/8)$) forces a work-injecting branch
$w<0$, which contradicts the pathwise cap at $b=0$. The identity below
then fails in both directions, depending on which branch holds the
deep atom. Thus no one-sided repair exists. Equality restoration for
general dyadic states above a budget floor $b^{*}(\tau,\pi_M)$, and
the $K_\tau\to\infty$ asymptotics of general distributions, are not
claimed in this paper.

\begin{lemma}[$\flatstar$ extraction identity]
\label{lem:flat-identity}
Let $\tau$ be a $\flatstar$ task, $q=p$, and $M$ a memory conforming
to Definition~\ref{def:memory}. Then for every $b\ge0$,
\begin{equation}
\Gap(M;\tau,b)=\kT\,I(M;X_\tau\mid Y_\tau)
\label{eq:flat-identity}
\end{equation}
---an equality, independent of the budget. More precisely, the
informed branch equals $\kT[\ln\lvert X_\tau\rvert-
H(X_\tau\mid M,Y_\tau)]$ and the blind branch equals
$\kT[\ln\lvert X_\tau\rvert-H(X_\tau\mid Y_\tau)]$. The proof is in
Appendix~\ref{app:equality}. Its informed-branch converse uses none
of the achievability normalization (F5$'$) and holds for every
state---the fact used by Main Theorem~\ref{thm:main-I}(iv);
combined with blind attainment, carried by the $y$-clause of
(F5$'$) (a property of the task tuple alone), it bounds
$\Gap(M;\tau,b)\le\kT\,I(M;X_\tau\mid Y_\tau)$ for \emph{every}
state (Lemma~\ref{lem:flattask} below).
\end{lemma}

\begin{corollary}[Value identity]
\label{cor:v-identity}
If the support of $T$ consists of $\flatstar$ tasks, then
$\V(M;T,b)=\kT\,\mathbb{E}_{\tau\sim T}
[\,I(M;X_\tau\mid Y_\tau)\,]$.
\end{corollary}

\begin{remark}[The separation theorem re-derived; blocked
correlation quantified]
\label{rem:lb-rederived}
On the device LB (Sec.~\ref{sec:value-main};
Appendix~\ref{app:devices}), $X_i\perp Y$ and the $Y$-components of
$M_k$ are independent of $X_i$, so
$I(M_k;X_1\mid Y)=I(M_{\mathrm{core}};X_1)=\ln2$ for every $k$ and
$I(M_k;X_2\mid Y)=0$; $\Delta\V=0$ becomes a one-line consequence of
Eq.~(\ref{eq:flat-identity}). The identity names the ``usable''
correlation exactly: it is the conditional mutual information
$I(M;X_\tau\mid Y_\tau)$, and the difference from $I(M;W)$
quantifies the ledger-blocked correlation.
\end{remark}

\begin{remark}[Scope of the identity]
\label{rem:flat-scope}
(i)~\emph{Load distribution}: the converse is carried by analysis
(the conditional chain bound, stated as a standalone lemma in
Appendix~\ref{app:equality} so that the weight-bearing path does not
depend on the interior of an imported proof), achievability by
(F5$'$). (ii)~\emph{Multi-register $\mathrm{Man}$} requires handling
the internal correlation $TC(X_{\mathrm{Man}})$ (initial correlation
becomes a consumable resource) and is outside the scope of this
version; tasks whose manipulable registers can be bundled into one
(a joint permutation writable as a single mechanism) are covered as
stated. (iii)~Relaxing (F4) to partial read yields a variant in
which only the readable part of $Y$ enters the conditioning; it is
not developed here. (iv)~With gates ($G\ne\emptyset$) both the
equality and the $\le\kT\,I$ upper bound fail in general---the
enablement route of Sec.~\ref{sec:regime}; this is a regime
boundary, not a defect. (v)~Outside (F5$'$) the equality fails at
small budgets, where the pathwise cap blocks exact blind attainment;
a finite budget wounds the blind side asymmetrically through
extraction fluctuations---a phenomenon of the same family as
enablement, arising inside the gate-free $\flatstar$ world. A
general dyadic version above a budget floor $b^{*}$ is a recorded
future slot (Sec.~\ref{sec:discussion}).
\end{remark}

\subsection{Update data processing}
\label{sec:ledger-dp}

\begin{lemma}[Update data processing]
\label{lem:updateDP}
Let $M'=f(M,D,A)$ be an $M$-local update
(Definition~\ref{def:update}) and $\tau$ a $\flatstar$ task, and
suppose the post-update state is again (F5$'$)---\emph{(F5$'$)-%
stability}, an explicit hypothesis, since (F5$'$) is a joint
property of $(\pi_M,p_\tau)$ and is not automatically preserved by
$M$-local updates (the copy-type updates of the device suite, whose
records are exact copies or independent coins, preserve it).
Then
\begin{equation}
\Delta\Gap_\tau\ \le\ \kT\,\min\bigl\{\,I(D;X_\tau\mid M,Y_\tau),\
\dpJD\,\bigr\},
\label{eq:updateDP}
\end{equation}
where $\dpJD=I(M';D\mid M)$; under the growth clause
$\dpJD=\dJD$. Averaging over $T$ supported on $\flatstar$ tasks:
$\Delta\V\le\kT\,\mathbb{E}_\tau[I(D;X_\tau\mid M,Y_\tau)]$ and
$\Delta\V\le\kT\,\dpJD$ (the second bound is $\tau$-uniform, hence
survives the average in the same form). Proof in
Appendix~\ref{app:equality}.
\end{lemma}

The two bounds carry independent meanings. The first, the
\emph{supply cap}, is the task-relevant information held by the
record; correlation absent from $D$ cannot be learned. The second,
the \emph{acquisition cap}, is the information that the memory
actually absorbed from $D$; correlation not absorbed does not become
capital.

\subsection{The universal ledger identity}
\label{sec:ledger-identity}

\begin{lemma}[Universal ledger identity]
\label{lem:ledger-identity}
For every $M$-local update and every implementation
(Definition~\ref{def:implementation}),
\begin{equation}
\begin{aligned}
\sigmaM &= \Delta I(M;D)+\Sigmatot ,\\
\Delta I(M;D) &:= I(M';D)-I(M;D),
\end{aligned}
\label{eq:ledger-universal}
\end{equation}
equivalently, in the four-term form obtained from the chain rule
($\dpJD-\Delta I=I(M;D\mid M')\ge0$),
\begin{equation}
\sigmaM+I(M;D\mid M')=\dpJD+\Sigmatot .
\label{eq:ledger-4term}
\end{equation}
Under the condition
\begin{equation*}
(\mathrm f)\colon\quad I(M;D\mid M')=0
\end{equation*}
(\emph{no discarded record correlation}; automatic under the growth
clause $M'=(M,R)$; unrelated to the refinancing clause (f) of
Sec.~\ref{sec:setting-imports}), and only under it,
$\sigmaM=\dpJD+\Sigmatot$, in particular $\dpJD\le\sigmaM$.
\emph{Side note (attribution)}: the identity
(\ref{eq:ledger-universal}) is not new---it is the static-source
degeneration of the bipartite information-flow entropy balances of
Refs.~\cite{Horowitz2014,Hartich2014} (the read-only $D$ has zero
own entropy production), plus the ancilla bookkeeping of
Definition~\ref{def:sigmaM}. Proof in
Appendix~\ref{app:equality}.
\end{lemma}

\begin{remark}[The role of read-only $D$]
\label{rem:read-only}
$S(D')=S(D)$ in the proof is the direct dividend of the read-only
clause. Allowing erasure or rewriting of $D$ would inject a
$\Delta S(D)$ term and open the loophole of lowering the apparent
$\sigmaM$ by consuming the record. The stronger content clause
(``identical in content'', not merely marginal-preserving)
guarantees in addition that the $\dpJD$ of
Lemma~\ref{lem:updateDP} and the $\dpJD$ here denote correlation
with \emph{the same random variable} $D$; an implementation that
preserved only the marginal while swapping $D$ for fresh uniform
bits would satisfy $S(D')=S(D)$ and still break the composition
theorem, and is excluded by type.
\end{remark}

Along update sequences the universal form telescopes exactly:
$\Sigma_{\mathrm{search}}^{(n)}
=[\,I(M_n;D)-I(M_0;D)\,]+\sum_k\Sigmatot^{(k)}$. The signed
increments $\Delta I$ cancel stage by stage, whereas the gross
acquisitions $\dpJD$ in general do not. This telescoping underlies the
cumulative bound in Main Theorem~\ref{thm:main-II}(v). The credit
$I(M;D\mid M')$ is the record correlation \emph{not} inherited by
$M'$. A non-growth update can recover this credit without heat by
reversible uncomputation against $D$ (Bennett's page recovery).
This recovery is what makes the per-update bound fail without~(f)
(route B4 of Sec.~\ref{sec:regime}).

\subsection{Main Theorem II}
\label{sec:ledger-main}

The remaining result is the per-task residual decomposition.

\begin{lemma}[Per-task decomposition of the update gap]
\label{lem:decomp}
Let $M'=f(M,D,A)$ be an $M$-local update and $\tau$ a $\flatstar$
task with both states (F5$'$). Then the \emph{equality}
\begin{equation}
\Delta\Gap_\tau=\kT\bigl[\,\dpJD-g_a(\tau)-g_b(\tau)-g_c(\tau)
\,\bigr]
\label{eq:decomp}
\end{equation}
holds, with the three nonnegative residual accounting items
(hereafter \emph{subjects})
\begin{align}
g_a(\tau)&:=I(M;X_\tau\mid Y_\tau,M')
&&\text{(\emph{forgetting}),}
\nonumber\\
g_b(\tau)&:=I(M';D\mid X_\tau,Y_\tau,M)
&&\text{(\emph{waste}),}
\label{eq:subjects}\\
g_c(\tau)&:=I(M';Y_\tau\mid M)
&&\text{(\emph{$Y$-contamination}).}
\nonumber
\end{align}
Proof in Appendix~\ref{app:equality}. The subjects read: task
information of $M$ not inherited by $M'$; correlation with $D$ held
in excess of $(X_\tau,Y_\tau,M)$; and the re-purchase of readable
variables that the blind side holds for free.
\end{lemma}

\begin{maintheorem}[Capitalization ledger]
\label{thm:main-II}
Let the setting be that of Secs.~\ref{sec:setting}
and~\ref{sec:value}, with implementations as in
Definition~\ref{def:implementation} and the ledger $\sigmaM$ of
Definition~\ref{def:sigmaM}.
\begin{enumerate}
\item[(i)] (\emph{$\flatstar$ extraction identity.}) For every
$\flatstar$ task and every $b\ge0$,
$\Gap(M;\tau,b)=\kT\,I(M;X_\tau\mid Y_\tau)$
(Lemma~\ref{lem:flat-identity}); hence
$\V=\kT\,\mathbb{E}_\tau[I(M;X_\tau\mid Y_\tau)]$ on
$\flatstar$-supported $T$ (Corollary~\ref{cor:v-identity}).
\item[(ii)] (\emph{Universal ledger identity.}) For every $M$-local
update and every implementation,
$\sigmaM=\Delta I(M;D)+\Sigmatot$; equivalently
$\sigmaM+I(M;D\mid M')=\dpJD+\Sigmatot$
(Lemma~\ref{lem:ledger-identity}).
\item[(iii)] (\emph{Upper bound.}) If the support of $T$ is
$\flatstar$, the update is $M$-local and (F5$'$)-stable, the
condition (f) holds, and $\sigmaM>0$, then
\begin{equation}
\etacap=\frac{\Delta\V}{\kT\,\sigmaM}\ \le\ 1 .
\label{eq:etacap-bound}
\end{equation}
\item[(iv)] (\emph{Equality conditions.}) Under the hypotheses of
(iii), the following are equivalent:
\begin{equation*}
\begin{gathered}
\etacap=1 \iff \\
\begin{cases}
\text{(a) no forgetting: } g_a(\tau)=0 \text{ $T$-a.s.,}\\
\text{(b) no waste: } g_b(\tau)=0 \text{ $T$-a.s.,}\\
\text{(c) no $Y$-contamination: } g_c(\tau)=0 \text{ $T$-a.s.,}\\
\text{(e) reversible implementation: } \Sigmatot=0 .
\end{cases}
\end{gathered}
\end{equation*}
Conditions (a)--(c) jointly assert full task-effectiveness of the
acquired correlation and are properties of the logical
specification; (e) alone is a property of the implementation. (The
labeling gap is deliberate: a candidate condition (d)---%
$T$-uniformity of the supply---is not independent, being absorbed by
the $T$-almost-sure quantification of (a)--(c).)
Without (f) the general equality condition is
$\mathbb{E}_\tau[g_a+g_b+g_c]+\Sigmatot=I(M;D\mid M')$---the total
breakage offsetting the recycling credit; (f) shuts this offsetting
flow and purifies the equality conditions.
\item[(v)] (\emph{Blank-start cumulative bound; (f) not needed,
(F5$'$)-stability not needed.})
Let $M_0$ be blank (deterministic, independent of all registers),
let every task in the support of $T$ be a flat task
(Definition~\ref{def:flattask} of
Sec.~\ref{sec:retention}, used here by forward reference---the
task-level $y$-clause of Lemma~\ref{lem:flattask}; equivalently,
$\flatstar$ with respect to the blank $M_0$---no (F5$'$) is
assumed at intermediate or terminal states), and let
$M_0\to\dots\to M_n$ be any sequence of $M$-local updates (growth,
non-growth, or (f)-violating) with implementations as in
Definition~\ref{def:implementation} and cumulative ledger
$\Sigma_{\mathrm{search}}=\sum_k\sigmaM^{(k)}>0$. Then
\begin{equation}
\eta_{\mathrm{cum}}:=\frac{\V(M_n;T,b)-\V(M_0;T,b)}
{\kT\,\Sigma_{\mathrm{search}}}\ \le\ 1 .
\label{eq:cumulative}
\end{equation}
\item[(vi)] (\emph{Regime map.}) Outside the hypotheses of (iii)
the bound (\ref{eq:etacap-bound}) fails along three identified
routes, each realized by an explicit device
(Appendix~\ref{app:devices}): substituting the denominator account
$\sigmaM\to\Sigmatot$ admits no bound
(device~E: $\Sigmatot=0$, $\Delta\V=\kT\ln2$); gate tasks under a
finite budget give $\etacap\gg1$, unboundedly---on device~G,
$\etacap\to\infty$ as $\beta F_{\mathrm{cart}}\to\infty$ at fixed
key length, horizon, and budget within partial suppression
($K_{\mathrm{hor}}2^{-n}e^{\beta b}<1$,
Appendix~\ref{app:fixedI}, B1), using only that premise and the
finiteness and uniformity of the imported bookkeeping constant
(Sec.~\ref{sec:regime}); and violating (f) gives
$\etacap>1$ at $\sigmaM>0$ and divergent families as the excess
tends to zero (device~CX: $\etacap=2$). Together with the
uncapitalized-dissipation valley $\etacap=0$ (device~LB$+$diss,
common to both accounts), these organize into the regime map of
Sec.~\ref{sec:regime} (Table~\ref{tab:regime}).
\end{enumerate}
\end{maintheorem}

Proofs are in Appendix~\ref{app:equality}; the devices are in
Appendix~\ref{app:devices}. They include the achievement of equality
in (iv) (device~E) and the systematic non-achievements that ignite
the subjects: (b), (c), and (e) each in isolation, and (a) jointly
with the conversion entry (c) (devices W, F, Y, LB$+$diss, E$+$ex).

Superposing Lemma~\ref{lem:decomp} ($T$-averaged) on
Lemma~\ref{lem:ledger-identity} yields the complete accounting
decomposition of the search ledger:
\begin{multline}
\sigmaM+I(M;D\mid M')
=\frac{\Delta\V}{\kT}\\
{}+\mathbb{E}_{\tau}\bigl[g_a(\tau)+g_b(\tau)+g_c(\tau)\bigr]
+\Sigmatot .
\label{eq:full-ledger}
\end{multline}
The five-term equality states that the search ledger plus recycling
credit equals the capitalized part, plus the three subjects of
acquisition breakage, plus pure dissipation. Under (f), the credit vanishes, and
Eq.~(\ref{eq:full-ledger}) reduces to the four-term form
$\sigmaM=\Delta\V/\kT+\mathbb{E}_\tau[g_a+g_b+g_c]+\Sigmatot$.
\emph{Full capitalization of the search ledger} means precisely
that, under (f), the second and third terms vanish simultaneously.
The equality conditions of part (iv) then close in the language of
ledger subjects.

\begin{remark}[Honesty box: where the content sits]
\label{rem:ledger-honesty}
The inequality (iii) is a composition of two imported moves: a
Sagawa--Ueda-type step (value $\le\kT\times$ information;
Lemma~\ref{lem:flat-identity} plus Lemma~\ref{lem:updateDP},
cf.~Ref.~\cite{Sagawa2010Generalized}) and a Goldt--Seifert-type
step (information $\le$ subsystem ledger;
Lemma~\ref{lem:ledger-identity} under (f),
cf.~Ref.~\cite{Goldt2017PRL}); no novelty is claimed for the
inequality itself, and the universal identity is attributed to
Refs.~\cite{Horowitz2014,Hartich2014}. Likewise the algebra of
Lemma~\ref{lem:decomp} is three chain-rule identities and two
Markov properties. The content of this section sits elsewhere: in
the domain of validity of the composition and its boundary
(Sec.~\ref{sec:regime}); in the operational choice of the
denominator account (Remark~\ref{rem:ledger-choice}); in the
identification of condition (f) and of the recycling credit; in the
equality conditions (iv), whose decomposition applies not to an
arbitrary information quantity but to increments of the
operationally anchored $\V$, and whose subjects are each ignited
by an explicit device---(b), (c), (e) in isolation, (a) jointly
with (c), the shared-world conversion effect
(Appendix~\ref{app:devices});
and in the cumulative bound (v). Attributions are tabulated in
Sec.~\ref{sec:related}.
\end{remark}

\section{Regime map}
\label{sec:regime}

Main Theorem~\ref{thm:main-II}(iii) is a regime theorem. Each way of
leaving its hypotheses is a signature, not an artifact. We trace this
boundary with minimal devices; their full definitions and numerical
verification appear in Appendix~\ref{app:devices}. The resulting map
is assembled in Table~\ref{tab:regime}.

\emph{(B1) Breaking the account: with $\Sigmatot$ in the
denominator, no bound exists.}---Device~E is a reversible copy. It
comprises a single uniform bit $X_1$ as world, $D$ as its perfect
copy, and blank $M_0$. The update is a controlled-copy of $D$ into
fresh memory, implemented as a deterministic permutation on
degenerate registers. Thus $\Sigmatot=0$, whereas
$\Delta\V=\kT\ln2$ (Lemma~\ref{lem:flat-identity}) and
$\Delta I(M;D)=\ln2$. Hence
for every $C>0$ the putative bound
``$\Delta\V\le C\,\kT\,\Sigmatot$'' is false, and the excess
variant E$+$ex gives a divergent family
$\Delta\V/(\kT\,\Sigmatot)=\ln2/s\to\infty$ as $s\to0^{+}$.
Correlation can be created at zero \emph{marginal entropy
production}. The Landauer fee sits on the erasure side, while the
process still consumes one blank-memory unit charged by $\sigmaM$.
A total-entropy-production denominator therefore supports no
capitalization bound. Under the $\sigmaM$ account, the same device
is billed $\sigmaM=\ln2$ for its blank page and achieves $\etacap=1$
exactly, saturating the equality conditions of Main
Theorem~\ref{thm:main-II}(iv). This contrast motivates
Definition~\ref{def:sigmaM}.

\emph{(B2) Breaking the task class: gates plus finite budget give
$\etacap\gg1$ (enablement).}---Device~G violates (F2). It has a key
register $B$ of $n$ uniform bits, a tamper-proof pass/fail gate that
releases a prize $F_{\mathrm{cart}}$ when opened, a horizon
$K_{\mathrm{hor}}$, and a per-run budget $b$.
Appendix~\ref{app:fixedI} states the imported bounds
self-containedly. With $D$ a perfect copy of $B$, $M_0$ blank, and
the update a reversible copy ($\sigmaM=n\ln2$), the informed side
holds the key and collects $F_{\mathrm{cart}}$ by a zero-work
conditional swap. It draws no external work and therefore runs
already at $b=0$. The blind side is capped by the imported ledger bound
$\mathbb{E}[W_{\mathrm{ext}}]\le
F_{\mathrm{cart}}\min(1,K_{\mathrm{hor}}2^{-n}e^{\beta b})
+b+C_G\,\kT$, where $C_G$ denotes the bookkeeping remainder of the
imported bounds and is assumed finite and uniform in the device parameters.
Its proven numerical upper bound is not available from the
companion, so no finite numerical claim below rests on its value.
Consider partial suppression,
$K_{\mathrm{hor}}2^{-n}e^{\beta b}<1$, the premise of
Appendix~\ref{app:fixedI}, B1. It turns the ledger cap's
$\min(1,\cdot)$ into a slope strictly below one. Outside it, the
blind side also collects the prize's leading term, and no
unboundedness holds. Under this premise, $\beta\,\Delta\V\ge
\beta F_{\mathrm{cart}}(1-K_{\mathrm{hor}}2^{-n}e^{\beta b})
-\beta b-C_G$. Moreover, $\etacap\to\infty$ as
$\beta F_{\mathrm{cart}}\to\infty$ at fixed $(n,K_{\mathrm{hor}},b)$
satisfying that premise. The efficiency is therefore unbounded at
fixed information content. This conclusion needs only partial
suppression and the finiteness and uniformity of $C_G$.
(As an illustration, not a theorem: at $(n,\beta F_{\mathrm{cart}},
\beta b,K_{\mathrm{hor}})=(10,100,0,3)$ the arithmetic reads
$\beta\,\Delta\V\ge 99.7-C_G$, which gives $\etacap\ge14.2$ under
the script convention $C_G=1.0$ of
Appendix~\ref{app:verification}. This convention is equivalent to
assuming $C_G\le1.28$~nats, not a proven bound.) One nat of
acquired information is worth more than one nat of extractable
work. This is not a hole in the accounting; it is the efficiency
version of the enablement theorem of the companion
framework~\cite{Sudo_MPU} (Appendix~\ref{app:fixedI}). Under a
finite budget, specification information gates a prize that the
budget cannot buy. At $b=\infty$, the premium closes. The blind side
purchases the key state by outright erasure, the gap collapses to
$n\,\kT\ln2$, and $\etacap=1$ exactly on the cartridge family. The
saturation direction is carried by an exact per-realization
cancellation, free of auxiliary hypotheses. A bound
$\etacap\le1$ for \emph{general} gate tasks at $b=\infty$ is not
derived in this paper. The exact
key length at which the blind side is suppressed all the way to its
budget floor is the full-suppression threshold restated in
Appendix~\ref{app:fixedI} (B2), with a transition window of width
$\sim\ln(\beta F_{\mathrm{cart}})$ between the loss of the premium
and full suppression.

The two cashing routes behind this asymmetry deserve names on the
subclass where they are well defined.

\begin{definition}[K1/K2-resolvable subclass]
\label{def:resolvable}
A task $\tau$ is \emph{resolvable} if its world variables decompose
as $W_\tau=(\Theta_\tau,X_\tau)$, where $\Theta_\tau$ is a
\emph{specification variable}---an opaque keyed target
specification in the sense of Ref.~\cite{Sudo_Replacement},
indexing the pass state of a tamper-proof, pass/fail-only gate in
$G_\tau$---and $X_\tau$ is a \emph{state variable} (the working
medium). On this subclass the memory's task correlation decomposes
by the chain rule as
\begin{equation}
I(M;\Theta_\tau,X_\tau)=
\underbrace{I(M;\Theta_\tau)}_{\text{K1}}
+\underbrace{I(M;X_\tau\mid\Theta_\tau)}_{\text{K2}} .
\label{eq:k1k2}
\end{equation}
\end{definition}

K2 correlation is acquired by the reversible-preparation--Landauer
route. Work buys it at the $\kT\ln2$ exchange rate, and the
Sagawa--Ueda-type commutation holds in this
sector~\cite{Sagawa2010Generalized}. K1 correlation is acquired by
the paid-query--guesswork route priced in
Ref.~\cite{Sudo_Replacement}. Within the finite-resource,
single-shot, K1/K2-delimited regime of these papers, work does not
buy it. That route asymmetry is what device~G converts into
$\etacap\gg1$. This statement is always to be read against the
asymptotic setting of resource theory~\cite{Brandao2013Resource},
where interconversion rates become reversible. The asymmetry is a
finite-run, single-shot claim. Beyond its appearance as an imported
name in Table~\ref{tab:types} and Fig.~\ref{fig:accounting}, the
vocabulary K1/K2 is used \emph{only} on the resolvable subclass. A
general task tuple carries no specification/state split. This
delimitation is inherited from Ref.~\cite{Sudo_Replacement}, which
prices only opaque keyed target specifications and asserts nothing
about general structural knowledge.

\emph{(B3) Dissipation without capitalization: the $\etacap=0$
valley.}---Device~LB$+$diss re-implements the update of device~LB
(copying the record of environment bits $Y$) in overwrite form.
Each stage first erases the target cell, releasing heat $\kT\ln2$
to the bath, and then writes the value of $D_k$ reversibly. Thus
$\sigmaM=\ln2>0$ per stage (plus any finite-time excess), whereas
Main Theorem~\ref{thm:main-I}(ii) pins $\Delta\V=0$. This is the
minimal implementation of ``dissipated but not
capitalized''. It is invisible to a raw-acquisition surrogate: the
ratio $\dpJD/\sigmaM$ equals $1$ on both device~E and device~LB$+$diss,
while $\etacap$ separates them as $1$ versus $0$. A caution on
attribution is needed. The numerator of Ref.~\cite{Goldt2017PRL}
is the mutual information with the teacher side, a downstream
quantity. A faithful transplant of such a numerator can also
discriminate through the definitional choice of what counts as the
task-relevant numerator. The difference here is that the
discrimination arrives as a \emph{theorem}: $\etacap=0$ is a
consequence of the extraction identity for the operationally
defined $\V$ (deletion, re-optimization, $T$, $b$), not of a naming
decision. The account also comes with a budget slot and with the
$\etacap>1$ side of the map, which no $\le1$-by-construction ratio
carries. The valley $\etacap=0$ is common to both denominator
accounts. In the $\Sigmatot$ account, it is read on the finite-time
excess variant; the quasistatic limit there is $0/0$.

\emph{(B4) Breaking (f): recycling credit, the drawdown of prepaid
capital.}---Device~CX starts from $M_0=$ a copy of an independent
junk coin $D_J$ (a natural state, one controlled-copy away from
blank). It then performs a non-growth update. First, it uncomputes
the junk against $D_J$ (a matched erasure, heat-free); next, it
writes the record bit $D_1$ into the freed cell; finally, it absorbs
one blank ancilla and writes $D_2$. All stages are deterministic
permutations: $\Sigmatot=0$,
$\sigmaM=\ln2$, $\dpJD=2\ln2$, and the discarded correlation is
$I(M_0;D\mid M_1)=\ln2$, the recycling credit. By
Corollary~\ref{cor:v-identity}, $\Delta\V=2\,\kT\ln2$, so
$\etacap=2>1$. Meanwhile, Lemmas~\ref{lem:updateDP}
and~\ref{lem:decomp} hold intact; only the (f)-form of the ledger
identity breaks. A one-bit variant with excess $s$ gives
$\etacap=\ln2/s\to\infty$ as $s\to0^{+}$. Thus, within
$\flatstar+\sigmaM$, the per-update bound is unbounded without (f).
The physics is the memory-side reverse of Bennett accounting.
$\sigmaM$ charges imported order through the $-S(A)$ deduction, but
order already stored as $M$--$D$ mutual information is outside the
account. An update that draws it down therefore receives a
denominator subsidy. This is not an artificial exploit. Uncomputing
an old note against the record to reuse the page is replay- or
consolidation-type memory management. The per-update excess is the
deferred ledger of an earlier stage. Extending device~CX to the
blank-start two-stage sequence (absorb junk at $\etacap=0$, then
recycle at $\etacap=2$) gives cumulative
$\eta_{\mathrm{cum}}=2\ln2/2\ln2=1$ exactly. The per-update
breakage is a drawdown of the opening balance, and the history's
books are honest. This is Main Theorem~\ref{thm:main-II}(v).

Table~\ref{tab:regime} assembles the map. Vertical failures
(gate$+$finite $b$) are the signature of the K1/budget regime.
Horizontal failures ($\Sigmatot$ account) are the signature of an
accounting that leaves memory resources off the books. The (f)
failures are the signature of drawing down stored record
correlation, whereas the $\etacap=0$ valley is account-independent.
The small-budget failure of the extraction identity outside
(F5$'$) (Remark~\ref{rem:flat-scope}(v)) marks one further stretch
of the same boundary, inside the gate-free world.

\begin{table*}[t]
\caption{\label{tab:regime}%
Regime map of the capitalization efficiency $\etacap$
(Main Theorem~\ref{thm:main-II}(vi)). Rows move through the task
classes, columns through the two readings of the denominator
account. The $\etacap=0$ valley (device LB$+$diss: dissipation
without capitalization) is common to both columns and account
independent (in the $\Sigmatot$ column, read on the finite-time
excess variant; the quasistatic limit is $0/0$). Devices are
defined in Appendix~\ref{app:devices}.}
\begin{ruledtabular}
\begin{tabular}{p{4.2cm}p{6.2cm}p{5.6cm}}
Task class $\backslash$ account &
$\Sigma_{\mathrm{search}}=\sigmaM$ &
$\Sigma_{\mathrm{search}}=\Sigmatot$ \\
\colrule
$\flatstar$, (F5$'$)-stable update, (f) holds &
$\etacap\le1$ [Main Thm.~\ref{thm:main-II}(iii)]; equality
$\iff$ (a)(b)(c)(e) [(iv)] &
unbounded (device E refutes every linear bound; E$+$ex diverges as
$s\to0^{+}$) \\
$\flatstar$, (f) broken (recycling) &
unbounded (device CX: $\etacap=2$; one-bit variant diverges).
Blank-start cumulative $\eta_{\mathrm{cum}}\le1$ recovers
[(v)] &
same \\
gate (K1), finite $b$ &
$\etacap\gg1$ possible (device G: enablement premium) &
same; moreover unbounded (the reversible implementation of device G
has $\Sigmatot=0$) \\
gate (K1), $b=\infty$ &
$\etacap=1$ exactly on the cartridge family (per-realization
cancellation); a general-gate $\le1$ statement is not derived
here &
unbounded \\
\end{tabular}
\end{ruledtabular}
\end{table*}

\section{Value retention and fit--value alignment}
\label{sec:retention}

\subsection{Shift topology and the retention quantities}
\label{sec:retention-defs}

\begin{definition}[Task-distribution shift]
\label{def:shift}
A \emph{shift} is a pair $(T_{\mathrm{train}},T_{\mathrm{shift}})$
of future task distributions. Its topology is fixed in three tiers:
(1)~\emph{support-fixed shift} (default): two distributions on a
common finite support $\mathcal{T}$, with total-variation distance;
the theorems of this section and of Sec.~\ref{sec:accounting} live
here. (2)~\emph{support-extending shift}: when
$\operatorname{supp}T_{\mathrm{shift}}\not\subseteq\mathcal{T}$,
the support is re-taken as
$\mathcal{T}'=\mathcal{T}\cup\operatorname{supp}
T_{\mathrm{shift}}$, and evaluating $\Gap$ on a new task requires
the $\tau_{\mathrm{new}}$ entry of the learning state
$\pi_M$---declaring the support is declaring coupling data, and the
specification is a per-device obligation (the key-redraw family of
Sec.~\ref{sec:gate} carries its specification inside the device
definition). On this tier total variation is uninformative (it
equals $1$ on disjoint point masses), and the theorems of
Sec.~\ref{sec:gate} rest on the declared coupling and imported
bounds instead. (3)~Metric strengthenings (Wasserstein-type) are
declared possible and not adopted.
\end{definition}

Two standing conventions apply. \emph{Support convention}: $\pi_M$ is a
learning state over the whole support ($\mathcal{T}$, or
$\mathcal{T}'$), and the consistency clauses of
Definition~\ref{def:memory}, in particular draw
exogeneity, are imposed on all of it. Exogeneity over the shifted
support is a substantive physical assumption. Specifically,
\emph{the shifted drawing mechanism does not read the training closure either}.
Realization-correlated adversarial shifts are thereby outside scope.
Such a world inspects the realized memory value $m$ and serves the
worst task. Distribution-level adversaries remain inside the
quantifiers: they choose $T_{\mathrm{shift}}$ knowing the design of $M$ and
$T_{\mathrm{train}}$. No collapse witness below uses a
realization-correlated shift. The distribution-level kind is exactly
what the shift witness (device SH) exercises. \emph{Budget convention}:
both arguments of the retention quantities are evaluated at the
\emph{same} $b$. Comparisons across budgets are the separate axis
of Main Theorem~\ref{thm:main-III}(iii). The notation $b=\infty$
abbreviates the protocol class without the drawdown-cap clause.
Each \emph{branch} of the gap then equals the supremum of its
finite-$b$ values (each branch is nondecreasing in $b$). The gap
itself carries no monotonicity in $b$. On $\flatstar$ supports,
$\V$ is $b$-independent (Lemma~\ref{lem:flat-identity}), whereas on
gate tasks, $\V$ at $b=\infty$ can lie strictly below its small-$b$
values (Main Theorem~\ref{thm:main-III}(iii)).

\begin{definition}[Retention gap and retention ratio]
\label{def:retention}
For a learning state $M$, a shift
$(T_{\mathrm{train}},T_{\mathrm{shift}})$, and a budget $b$,
\begin{equation}
\begin{aligned}
\Lgen(M)&:=\V(M;T_{\mathrm{train}},b)-\V(M;T_{\mathrm{shift}},b),\\
\rhogen(M)&:=\frac{\V(M;T_{\mathrm{shift}},b)}
{\V(M;T_{\mathrm{train}},b)} ,
\end{aligned}
\label{eq:Lgen}
\end{equation}
the gap form (loss under shift counted positive; units of work) and
the ratio form (dimensionless retention), the latter defined
\emph{only when} $\V(M;T_{\mathrm{train}},b)>0$.
\end{definition}

$\Lgen$ carries no sign constraint. Capital is $T$-relative
(Remark~\ref{rem:t-relativity}), and the shift may raise the value.
$\rhogen$ is not confined to $[0,1]$. The domain restriction is a
convention, not an emptiness claim. $\V_{\mathrm{train}}=0$ does
not mean ``nothing was learned''. Shift-specialized states with
$\V_{\mathrm{train}}=0$ and $\V_{\mathrm{shift}}>0$ exist
routinely, and the gap form retains their content in full:
$\Lgen=-\V_{\mathrm{shift}}<0$. Only the corner
$\V_{\mathrm{train}}=\V_{\mathrm{shift}}=0$ is genuinely
uninformative. The division of labor is fixed once. \emph{The gap
form is for accounting}. Its unit is work, and the
subject-difference identity of Sec.~\ref{sec:accounting} is
additive in it. \emph{The ratio form is for shape}. The
linear-versus-threshold contrast of Sec.~\ref{sec:gate} and the
orthogonality corollary are stated in it. Neither replaces the
other. $\Lgen$ is insensitive to the scale of
$\V_{\mathrm{train}}$, and $\rhogen$ has no additive
decomposition.

\begin{definition}[Task optimum]
\label{def:taskopt}
$\V^{*}(T,b):=\sup_M \V(M;T,b)$, the supremum over memories and
learning states conforming to Definition~\ref{def:memory}; it
enters only through Proposition~\ref{prop:taskopt}.
\end{definition}

Main Theorem~\ref{thm:main-I}(i) gives three basic properties. (G1)~For
support-fixed shifts,
$\lvert\Lgen(M)\rvert\le\lVert T_{\mathrm{train}}-
T_{\mathrm{shift}}\rVert_{\mathrm{TV}}\cdot
\sup_{\tau\in\mathcal{T}}\Gap(M;\tau,b)$; the constant depends on
$M$ and $b$ and is \emph{not} uniform over families of states.
(G2)~no sign constraint holds for $\Lgen$ (witness: the shift
device of Main Theorem~\ref{thm:main-III}(ii) with the roles of the
two tasks exchanged). (G3)~composition with updates is measured by
$\Lgen^{\mathrm{upd}}:=\Delta\V_{\mathrm{train}}-
\Delta\V_{\mathrm{shift}}=\Lgen(M')-\Lgen(M)$. This is the quantity
computed by the subject-difference identity of Sec.~\ref{sec:accounting}.

\emph{Terminology defense.}---The phrase \emph{thermodynamic
generalization} is occupied by Ref.~\cite{Boyd2025Overfitting},
whose definition is cost-side. It requires that dissipation not
diverge on tests drawn from the \emph{same} source. The quantities of
Definition~\ref{def:retention} differ in three respects. They are
value-side. Their subject is distribution shift, a concept absent
from that framework; budget and access are explicit arguments.
The primary term in this paper is \emph{value retention}, and
``thermodynamic generalization'' appears only as a related-work
distinction (Sec.~\ref{sec:related}). Ref.~\cite{Caraffa2026Dissipative}
uses ``generalization'' for a dissipative-balance steady-state
criterion. This is a formal analogy in which work extraction, task
value, and shift accounting do not appear. The work itself states, ``we
lack fundamental bounds relating computational resources to
learning capacity''; the present quantities are disjoint from that
usage. The \emph{work capacity} of
Ref.~\cite{Fiderer2025WorkCapacity} prices percept--action loops
(active updates). These loops lie outside the update class of
Definition~\ref{def:update} (Sec.~\ref{sec:value-blackwell}). No
bridge from any quantity in this paper to statistical
generalization (test risk, sample complexity) is claimed
(Sec.~\ref{sec:discussion}).

\subsection{Flat tasks: a state-uniform bound and the task
optimum}
\label{sec:retention-flat}

\begin{definition}[Flat task]
\label{def:flattask}
A task $\tau$ is a \emph{flat task} if it satisfies (F1)--(F4) of
Definition~\ref{def:flatstar} together with the \emph{$y$-clause}:
the tuple-level conditional distribution $p_\tau(x_\tau\mid y)$ is,
for each $y$, a point mass or a uniform distribution on a subset,
and $K_\tau\ge3$. The $y$-clause is the task-level part of
(F5$'$); (F5$'$) itself is the joint property that adds the
$(y,m)$-sections.
\end{definition}

\begin{lemma}[Flat-task upper bound: all states, all budgets]
\label{lem:flattask}
If $\tau$ is a flat task, then for \emph{every} memory $M$
conforming to Definition~\ref{def:memory}---including states that
violate (F5$'$)---and every $b\ge0$,
\begin{equation}
\Gap(M;\tau,b)\ \le\ \kT\,I(M;X_\tau\mid Y_\tau) .
\label{eq:flattask-bound}
\end{equation}
Proof in Appendix~\ref{app:equality}: the informed converse of
Lemma~\ref{lem:flat-identity} uses no (F5$'$) and holds for all
states; blind achievability is carried by the $y$-clause, which is
a task property independent of $\pi_M$. The $y$-clause is
essential, not a convenience: without it the blind side cannot
attain its ceiling at $b=0$ while a point-mass-informed state can,
and $\Gap>\kT\,H(X\mid Y)$ becomes possible.
\end{lemma}

\begin{proposition}[Task optimum of flat-task supports]
\label{prop:taskopt}
Let the support of $T$ consist of flat tasks and assume either
(i)~the manipulable registers of the tasks in the support are
pairwise disjoint physical registers (with the off-task components
of the optimizing state specified independent), or (ii)~a shared
world in which $p_\tau(w)$ is common to all $\tau$. Then
\begin{equation}
\V^{*}(T,b)=\kT\,\mathbb{E}_{\tau\sim T}
\bigl[H(X_\tau\mid Y_\tau)\bigr] ,
\label{eq:taskopt}
\end{equation}
independent of $b$. Proof in Appendix~\ref{app:equality}; in a
shared world with $\tau$-dependent $p_\tau(w)$ the copy
construction violates draw exogeneity and
Eq.~(\ref{eq:taskopt}) is not claimed, nor is a closed form on
gate-mixed supports.
\end{proposition}

\subsection{Fit--value alignment}
\label{sec:retention-alignment}

\begin{definition}[Fit--value alignment]
\label{def:alignment}
Let $F=\{f_1,\dots,f_r\}$ be a finite family of candidate $M$-local
updates, $\Phifit$ a fit functional, $T_{\mathrm{eval}}$ an
evaluation distribution, and $b$ a budget. The pair
$(\Phifit,\Delta\V)$ is \emph{aligned} (weak form) on $F$ if
$\operatorname*{arg\,max}_{f\in F}\Phifit(f)\subseteq
\operatorname*{arg\,max}_{f\in F}\Delta\V(f;T_{\mathrm{eval}},b)$;
it exhibits \emph{strict misalignment} if
$\operatorname*{arg\,max}_{f}\Phifit(f)\cap
\operatorname*{arg\,max}_{f}\Delta\V(f;T_{\mathrm{eval}},b)
=\emptyset$---no fit-best candidate is value-best. The collapse
statements below are all in the strict form (a silent fit measure,
constant on $F$, would otherwise be misread as misalignment), and
the positive domain is stated as an equality of values, the
strongest form. Two eligibility clauses: $\Phifit$ is a functional
of the joint law of the training closure $(M,M',D)$ (representatives:
the acquisition $\dpJD$, the world correlation $I(M';W)$), or a
work functional observable in the training run (in
Ref.~\cite{Boyd2022TML} work itself is the training signal);
functionals that read the future distribution $T_{\mathrm{shift}}$
are ineligible as fit measures---they are the value itself. And the
family $F$ carries a \emph{capacity constraint} (writable bits
strictly fewer than the record's bits) in every witness below, so
that full copy is excluded by a resource bound rather than by
fiat: the content of the collapse theorems is that \emph{when
capacity forces a choice, the fit ranking can misguide}---the
faithful counterpart of training constrained to a bounded memory
class in Ref.~\cite{Boyd2022TML}.
\end{definition}

\begin{proposition}[Correspondence proposition: the degenerate
domain, in two layers]
\label{prop:corresp}
Consider the degenerate setting:
($\alpha$)~a single $\flatstar$ task $\tau$, with $D$ the record of
$\tau$'s own manipulable medium (same source; with
$T_{\mathrm{eval}}=\delta_\tau$ there is no separate train/shift
slot, and this clause is the type-correct expression of ``same
source'');
($\beta$)~$D=X_\tau$ as a perfect copy (the record is a faithful
observation of the medium);
($\gamma_{\mathrm{acc}}$)~static all-read, no manipulation blockage;
($\delta$)~$b$ arbitrary;
and let the candidates be $M$-local updates of one and the same
initial state $M$, forming an (F5$'$)-stable family (all post-update
states satisfy (F5$'$); deterministic copy-type updates qualify).
\begin{enumerate}
\item[(A)] (\emph{Conditional exchange rate.}) For every candidate
$f$,
\begin{equation}
\V(M'_f;\delta_\tau,b)
=\kT\,I(M'_f;X_\tau\mid Y_\tau)
=\kT\,I(M'_f;D\mid Y_\tau) ,
\label{eq:corresp}
\end{equation}
independently of $b$: the value coincides with the
\emph{side-information-adjusted} record fit
$\Phifit^{Y}(f):=I(M'_f;D\mid Y_\tau)$ at the exchange rate $\kT$,
and alignment holds as an equality of values. \emph{No independence
between $X_\tau$ and $Y_\tau$ is assumed.}
\item[(B)] (\emph{Raw record stock, under joint neutrality.}) If in
addition ($\gamma$)~\emph{joint side-information neutrality} holds,
$(M,D)\perp Y_\tau$, then every candidate satisfies
$(M'_f,D)\perp Y_\tau$ (the fresh ancilla is jointly independent,
and $M'_f$ is a function of $(M,D,A)$), and
\begin{equation}
\V(M'_f;\delta_\tau,b)=\kT\,I(M'_f;D)=\kT\,\JD(M'_f) .
\label{eq:corresp-raw}
\end{equation}
Under a blank start ($M$ deterministic), neutrality reduces to
$D\perp Y_\tau$ (automatic when $X_\tau\perp Y_\tau$ or
$Y_\tau=\emptyset$), $\JD(M)=0$, and
$\Delta\V_f=\kT\,\dpJD(f)$: the raw acquisition ranking is the
value ranking.
\end{enumerate}
Proof in Appendix~\ref{app:equality}.
\end{proposition}

\begin{proposition}[Boundary witness: $X\perp Y$ does not suffice
for the raw layer]
\label{prop:corresp-boundary}
Weakening the hypothesis ($\gamma$) of
Proposition~\ref{prop:corresp}(B) to $X_\tau\perp Y_\tau$ makes
Eq.~\eqref{eq:corresp-raw} false. Witness: $X=(X_1,X_2)$ two uniform
bits (a single manipulable register), $Y$ one readable bit with
$X\perp Y$, $D$ a perfect copy of $X$, and $M$ a content copy of
$Y$---a world-correlated initial memory, which no admissibility
clause excludes; candidates $f_1=\mathbf{1}\{X=(0,0)\}$ (reading $D$
only) and $f_2=X_1\oplus M$ (a one-time pad). All hypotheses of
part (A) hold ($M$-local, (F5$'$)-stable, ($\alpha$)($\beta$)%
($\gamma_{\mathrm{acc}}$)($\delta$)), yet
\begin{equation*}
\begin{aligned}
I(M'_2;D) &= 0<\ln2=I(M'_2;D\mid Y),\\
I(M'_1;D) &= I(M'_1;D\mid Y)=h(\tfrac14) ,
\end{aligned}
\end{equation*}
with $h$ the binary entropy in nats, $h(\tfrac14)<\ln2$: the raw
record-stock argmax $\{f_1\}$ and the value argmax $\{f_2\}$ are
disjoint---strict misalignment strictly inside the domain of part
(A), whose conditional exchange rate remains exact on both
candidates. The driving mechanism is externally sourced
$Y$-correlation held by $M$, which cannot be manufactured from $D$;
joint neutrality is precisely the condition that closes this
channel. The witness is kept as a negative control in the numerical
suite (Appendix~\ref{app:devices}).
\end{proposition}

\begin{remark}[Limits of the correspondence---to be read with
Proposition~\ref{prop:corresp}]
\label{rem:corresp-limits}
The setting of Refs.~\cite{Boyd2022TML,Boyd2025Overfitting} and the
present one are different formal systems: (i)~unipartite ratchet
versus multipartite task tuple; (ii)~$\epsilon$-machine versus
learning state $\pi_M$; (iii)~sequence asymptotics versus one-shot
per draw; (iv)~the fit functionals differ in type---their
likelihood is a pathwise functional of realized data (an
estimator), our mutual information is an ensemble functional (not
an estimator), so the specifically maximum-likelihood content of
their equivalence has no image under the correspondence;
(v)~the record/world distinction does not exist in their
formalism (the data \emph{is} the tape), so clause ($\beta$) is
unbreakable there; (vi)~their ``efficient agent'' hypothesis
corresponds to (F5$'$)-stability here---both equivalences hold
inside a class of efficiently implementable candidates.
Proposition~\ref{prop:corresp} is therefore neither a re-derivation
nor a subsumption of their theorem: it is a correspondence that
isolates, in the present vocabulary, the structural features under
which their equivalence lives.
\end{remark}

\subsection{Main Theorem III}
\label{sec:retention-main}

\begin{maintheorem}[Value-retention alignment schema]
\label{thm:main-III}
Let the setting be that of Secs.~\ref{sec:setting}
and~\ref{sec:value}. Then:
\begin{enumerate}
\item[(i)] (\emph{Alignment in the degenerate domain, two
layers.}) Under the degenerate conditions
($\alpha$)($\beta$)($\gamma_{\mathrm{acc}}$)($\delta$) and
(F5$'$)-stability, value ranking aligns with the
side-information-adjusted record fit $I(M'_f;D\mid Y_\tau)$ through
the linear exchange rate $\kT$, with no record--side-information
independence assumed [Proposition~\ref{prop:corresp}(A)]; under the
additional joint neutrality ($\gamma$): $(M,D)\perp Y_\tau$, it
aligns with the raw acquired record stock $I(M'_f;D)$ (or any
strictly increasing transform of it)
[Proposition~\ref{prop:corresp}(B)]. The boundary between the two
layers is genuine: $X\perp Y$ alone does not restore the raw layer
(Proposition~\ref{prop:corresp-boundary}).
\item[(ii)] (\emph{S: shift collapse---dropping ($\alpha$).}) On
the device SH (shared world of two independent uniform bits
$X_A,X_B$; two symmetric $\flatstar$ tasks with
$\mathrm{Man}=\{X_A\}$ resp.\ $\{X_B\}$; $D$ a perfect two-bit
record; blank start; capacity one bit; candidates $f_A,f_B$ the two
one-bit copies; $T_{\mathrm{train}}=(1-\epsilon,\epsilon)$,
$T_{\mathrm{shift}}=(\epsilon,1-\epsilon)$,
$\epsilon\in(0,\tfrac12)$; Appendix~\ref{app:sbar}):
$\Delta\V(f_A;T_{\mathrm{train}},b)=(1-\epsilon)\kT\ln2
>\epsilon\,\kT\ln2=\Delta\V(f_B;T_{\mathrm{train}},b)$, while under
$T_{\mathrm{shift}}$ the value ranking strictly reverses (the
argmax sets are disjoint); moreover
$\rhogen(M_A)=\epsilon/(1-\epsilon)$, arbitrarily small inside the
domain $\V_{\mathrm{train}}>0$, and
$\Lgen(M_A)=(1-2\epsilon)\,\kT\ln2$ attains the bound (G1) with
equality. Here $\Phifit=\Delta\V(\cdot\,;T_{\mathrm{train}},b)$
is the training-distribution value itself (eligible as a work-type
fit functional: it reads $T_{\mathrm{train}}$, not
$T_{\mathrm{shift}}$), and the acquisition $\dpJD$ is tied at
$\ln2$ for both candidates, as is the world correlation: the
representative closure measures do not separate the candidates;
the evaluation distribution alone does.
\item[(iii)] (\emph{B: budget collapse---dropping ($\delta$).}) On
the device BG (an independent pair: a flat part of $n$ uniform bits
and an imported gate part with key length $m<n$, prize
$F_{\mathrm{cart}}$, horizon $K_{\mathrm{hor}}$, Type~I
realization; $T=(\tfrac12,\tfrac12)$ fixed; $D=(D_X,D_B)$ perfect
copies; capacity $n$ bits; candidates $f_X,f_B$ the two full
copies; parameters satisfying $n\ln2>2m\ln2+2\,C_G$ and
$\beta F_{\mathrm{cart}}(1-K_{\mathrm{hor}}2^{-m})-2\,C_G>n\ln2$,
with $C_G$ the finite uniform bookkeeping constant of
Sec.~\ref{sec:regime}---whether a given numerical quadruple
satisfies these inequalities depends on the unproven value of
$C_G$; e.g.\ $(n,m,\beta F_{\mathrm{cart}},K_{\mathrm{hor}})
=(80,20,100,3)$ qualifies iff $C_G<13.9$~nats, which holds under
the script convention of Appendix~\ref{app:verification} but is
not proven, while for every finite $C_G$ admissible quadruples
exist (take $n$ large); Appendix~\ref{app:sbar}):
(1)~the fit order is budget-independent,
$\dpJD(f_X)=n\ln2>m\ln2=\dpJD(f_B)$;
(2)~at $b=0$ the value order strictly reverses,
$\V(M_B;T,0)\ge\tfrac12\kT[\beta
F_{\mathrm{cart}}(1-K_{\mathrm{hor}}2^{-m})-2C_G]
>\V(M_X;T,0)=\tfrac12\,n\,\kT\ln2$;
(3)~at large budget the order recovers: for
$\beta b\ge m\ln2+C_G$, the gate-part gap lies in the hypothesis-%
free sandwich $[\,m\,\kT\ln2,\ 2m\,\kT\ln2+2\,C_G\,\kT\,]$, whose
upper end still gives $\V(M_B)<\V(M_X)$ under the stated parameter
condition (for reference: the exact equality
$\mathrm{Gap}_{\mathrm{gate}}=m\,\kT\ln2$ holds conditionally on a
calibration hypothesis of the companion framework, which is open
in general and constructively satisfied on cartridges; it is not
used);
(4)~hence the success of alignment reverses as a function of the
evaluation budget $b$---established at the two endpoints; no
monotonicity in $b$ or uniqueness of the reversal point is
claimed.
\item[(iv)] (\emph{A: access-route collapse---breaking the
joint-neutrality clause ($\gamma$).}) On the device AC$'$ (a single
task: $X=(X_1[2~\text{bits}],X_2[1~\text{bit}])$ a single
manipulable register, $Y$ a readable non-manipulable duplicate of
the content of $X_1$, so $I(X;Y)=2\ln2>0$; $D$ a perfect copy of
$X$, so ($\beta$) is held; $T=\delta_\tau$, $b$ arbitrary; blank
start; capacity two bits; candidates $f_1,f_2$ copying the $X_1$-
resp.\ $X_2$-part of $D$; Appendix~\ref{app:sbar}): the
representative fit measures prefer $f_1$
($\dpJD(f_1)=2\ln2>\ln2=\dpJD(f_2)$; likewise the world
correlation), while the value strictly reverses:
$\Delta\V(f_2)=\kT\ln2>0=\Delta\V(f_1)$---the larger acquisition
is $Y$-redundant, a re-purchase of what the blind side reads for
free. This occurs at the same $T$ and every $b$; extending the
manipulable set to $Z=(X,Y)$ restores the fit order,
$\V'(f_1)=2\kT\ln2>\V'(f_2)$.
\item[(v)] (\emph{R: record-excess collapse---dropping
($\beta$).}) On the device RJ (a single $\flatstar$ task with $X$
one uniform bit, $Y=\emptyset$; $J$ a junk register of $n\ge2$
uniform bits independent of $X$; $D=(\text{copy of }X,\,J)$, so the record
exceeds the faithful observation of the medium; $T=\delta_\tau$,
$b$ arbitrary; capacity $n$ bits; candidates $f_X,f_J$;
Appendix~\ref{app:sbar}): $\dpJD(f_J)=n\ln2>\ln2=\dpJD(f_X)$ while
$\Delta\V(f_X)=\kT\ln2>0=\Delta\V(f_J)$---strict reversal;
replacing every component of $D$ by $X$-content (record
faithfulization) restores agreement.
\item[(vi)] (\emph{Four-coordinate intervention theorem.}) For
each coordinate $C\in\{$S (shift), B (budget), A (access route), R
(record content)$\}$ there is a paired setting---same device
family, same candidate family, same fit functional---in which
changing coordinate $C$ \emph{alone} switches strict misalignment
to alignment: parts (ii)--(v) supply the pairs, with the
restoration operations
$T_{\mathrm{shift}}\to T_{\mathrm{train}}$; $b\to\infty$;
manipulable-set extension; record faithfulization. Here the
access-route coordinate is the pair
$(\mathrm{Man}_\tau,A_\tau)$ of manipulability and access
structure---in the formal task tuple
(Definition~\ref{def:task}) these are separate slots, and the
AC$'$ intervention moves only its $\mathrm{Man}_\tau$ member while
the narrow $A_\tau=(\mathrm{Acc}_0,R)$ stays all-read. (Gate
removal would also restore the budget pair, but it changes
$G_\tau$, not $b$, and is therefore not part of the
single-coordinate intervention.) The
fit--value ranking is thus separately sensitive to each of the
four operational coordinates. This part asserts \emph{no} uniform
necessity claim for the four conditions of
Proposition~\ref{prop:corresp}: the S, B, and R witnesses move the
letter of other clauses as a structural necessity of housing
several candidates (a multi-candidate record exceeds a single
faithful medium copy; the BG support contains a non-$\flatstar$
task), and for the shift coordinate no witness family can repair
this (candidates that a shift separates force a multi-medium
record, breaking the letter of ($\beta$)).
\item[(vii)] (\emph{Necessity of the neutrality coordinate.}) The
access witness AC$'$ of part (iv) satisfies ($\alpha$), ($\beta$),
($\delta$), a blank start, and (F5$'$) at all reachable states,
and uses the same fit functional as
Proposition~\ref{prop:corresp}(B) (the raw record stock, equal to
$\dpJD$ under the blank start); the only hypothesis of part (B)
that fails is joint neutrality ($\gamma$)---under a blank start,
$(M,D)\perp Y\iff D\perp Y$, broken here by the $X$--$Y$
redundancy. Hence ($\gamma$) cannot be dropped from the hypothesis
set of Proposition~\ref{prop:corresp}(B) if the alignment
guarantee is to be uniform over all devices and (F5$'$)-stable
candidate families. The necessity is necessity \emph{for the
uniform guarantee}, not for individual instances (on a
single-element family the argmax sets coincide trivially).
\end{enumerate}
\end{maintheorem}

Device definitions, proofs (arithmetic of
Corollary~\ref{cor:v-identity}, the uninformative-port lemma, and
the imported gate bounds), and the restoration numerics appear in
Appendix~\ref{app:sbar}. The noisy-record variant of (iv), whose
reversal claim is deliberately \emph{excluded} from the theorem, is
also recorded there. We emphasize the claim level: the four axes are
not asserted to classify all
misalignment phenomena; no if-and-only-if is claimed. The theorem is
a schema comprising a two-layer positive domain with its boundary
witness, four paired one-coordinate interventions, and a genuine
one-condition-drop necessity on the neutrality axis alone.

The four axes form the paper's interpretive backbone.
Alignment holds where training-side fit suffices. It collapses
where \emph{what to fit to} (shift), \emph{which part of the record
is the medium} (record excess), \emph{through which route the
correlation can be cashed} (access), and \emph{under what budget it
is cashed} (budget) become decisive independently of the goodness
of fit. The first three axes will reappear as ledger subjects in
Sec.~\ref{sec:accounting} (($\beta$)$\leftrightarrow g_b$,
($\gamma$)$\leftrightarrow g_c$, ($\alpha$)$\leftrightarrow$
re-posting between books). The budget axis lies outside the subject
ledger (Remark~\ref{rem:reposting});
Fig.~\ref{fig:alignment} charts the two-layer domain and the four
paired interventions.

\begin{figure}[t]
\begin{tikzpicture}[every node/.style={font=\footnotesize}]
\node[align=center,anchor=south] at (4.1,6.05)
{paired one-coordinate interventions:\\[1pt]
 change the named coordinate alone $\Rightarrow$\\
 ranking reverses, or alignment returns};
\draw (0.2,3.0) rectangle (8.0,5.7);
\node[align=center] at (4.1,4.35)
{two-layer alignment domain\\[2pt]
 (A) conditional rate $\kT\,I(M';D\mid Y)$;\\
 (B) raw stock under $(M,D)\perp Y$\\
 (Prop.~\ref{prop:corresp})};
\draw[<->] (0.9,3.0) -- (0.8,1.8);
\node[align=center,anchor=north,text width=1.4cm] at (0.8,1.65)
{S: shift\\ device SH};
\draw[<->] (2.7,3.0) -- (2.6,1.8);
\node[align=center,anchor=north,text width=1.4cm] at (2.6,1.65)
{B: budget\\ device BG};
\draw[<->] (4.8,3.0) -- (4.8,1.8);
\node[align=center,anchor=north,text width=2.2cm] at (4.8,1.65)
{A: access route\\ device AC$'$\\ ($\gamma$: necessity)};
\draw[<->] (7.2,3.0) -- (7.3,1.8);
\node[align=center,anchor=north,text width=1.5cm] at (7.3,1.65)
{R: record\\ device RJ};
\end{tikzpicture}
\caption{\label{fig:alignment}%
Fit--value alignment: the two-layer domain in which fit ranking
and value ranking coincide (Main Theorem~\ref{thm:main-III}(i)),
whose boundary is marked by the one-time-pad witness
(Proposition~\ref{prop:corresp-boundary}), and the four paired
interventions S/B/A/R (parts (ii)--(vi)): each double arrow is a
setting pair in which changing that coordinate alone reverses or
restores the ranking. The neutrality coordinate ($\gamma$, device
AC$'$) alone carries a one-condition-drop necessity witness (part
(vii)); no uniform necessity is claimed for the other three. The
four axes are not claimed to classify all misalignment phenomena.}
\end{figure}
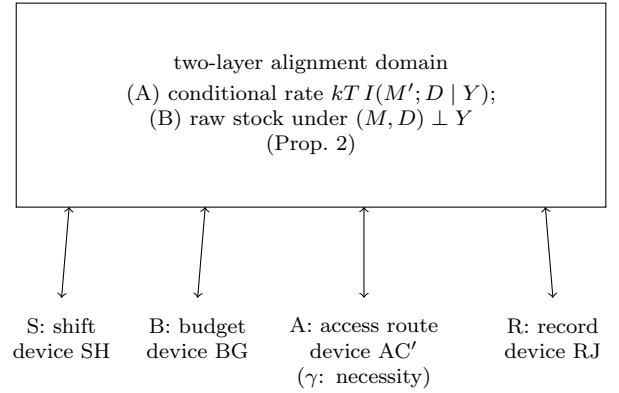

\section{Accounting consequences}
\label{sec:accounting}

The three subjects of Lemma~\ref{lem:decomp} translate into the
language of value retention. Call a subject \emph{absolute} if it
never contributes to $\Lgen^{\mathrm{upd}}$ (its subject difference
vanishes for every shift pair) and \emph{relative} otherwise. The
exact types are as follows.
\emph{Forgetting} $g_a=I(M;X_\tau\mid Y_\tau,M')$ is a
forgetting-type retention failure. It is $\tau$-dependent and can
be \emph{invisible in the training books}. If the support of
$T_{\mathrm{train}}$ does not contain the old task,
$g_a(\tau_{\mathrm{train}})=0$, while the loss surfaces only after
the shift. The corresponding device~F of
Appendix~\ref{app:devices} is the book reversal in which
overwriting an old memory costs nothing on the new book and
$\kT\ln2$ on the old one. This is the standard continual scenario
``train on the new task, the old one returns''. Hence
$g_a$ is relative. \emph{Waste} $g_b=I(M';D\mid X_\tau,Y_\tau,M)$
is adaptation to record excess (the ledger version of Main
Theorem~\ref{thm:main-III}(v)). On a shared world it is
$\tau$-constant and hence absolute, as the following lemma shows.
On non-shared worlds it is $\tau$-dependent and relative. The
example is device~DJ of Appendix~\ref{app:devices}: an acquisition
that is capital on $\tau_A$'s book is booked entirely as $g_b$ on
$\tau_B$'s. \emph{Junk overfitting} is acquisition dead on every
book, the component independent of all task variables. It is a
\emph{sub-event} of $g_b$, not the subject itself.
\emph{$Y$-contamination} $g_c=I(M';Y_\tau\mid M)$ charges the
whole $Y$-redundant component of the acquisition, wider than
literal copies of readable variables. It includes supply-side
contributions. There are $(X,Y,D)$ with $I(X;Y)>0$ on which a
perfectly hygienic learner cannot achieve $g_c=0$ together with
positive acquisition. The record itself is $Y$-redundant, and the
update kernel, which cannot read $Y$, cannot filter out the
redundant component. It is relative.

\begin{lemma}[Shared-world waste is $\tau$-constant]
\label{lem:sharedworld}
If all tasks in the support share one world register set $W$, with
$(X_\tau,Y_\tau)$ a partition of $W$ for each $\tau$, then
$g_b(\tau)=I(M';D\mid W,M)$ does not depend on $\tau$. (Proof in
Appendix~\ref{app:equality}: conditioning on $(X_\tau,Y_\tau)$
jointly is conditioning on $W$.)
\end{lemma}

The absolute/relative distinction is the content of the
translation. The classical word ``overfitting'' merges
(i)~acquisition that is dead on every future (junk; represented by
shared-world $g_b$) and (ii)~acquisition that is capital on the
training distribution and dies under the shift. The typed
accounting separates the two as vanishing versus surviving subject
differences in device-verifiable form. The dividing line depends
on the world-sharing structure, not on the subject's name. Even
$g_b$ is relative on non-shared worlds.

\begin{corollary}[Subject-difference identity for the retention
gap]
\label{cor:subject-diff}
Assume: (i)~$T_{\mathrm{train}}$ and $T_{\mathrm{shift}}$ are
distributions on a \emph{common finite support} $\mathcal{T}$,
every $\tau\in\mathcal{T}$ is a flat task, and both states satisfy
(F5$'$) on every $\tau\in\mathcal{T}$ ((F5$'$)-stability of the
update, in the two-state, all-$\tau$ quantification);
(ii)~$M\to M'$ is an $M$-local update
(Definition~\ref{def:update}); (iii)~draw exogeneity holds over
all of $\mathcal{T}$ (Definition~\ref{def:exogeneity}; in
particular the shifted drawing mechanism does not read the training
closure). Then
\begin{equation}
\Lgen^{\mathrm{upd}}
=\kT\Bigl(
\mathbb{E}_{\tau\sim T_{\mathrm{shift}}}
\bigl[g_a+g_b+g_c\bigr]
-\mathbb{E}_{\tau\sim T_{\mathrm{train}}}
\bigl[g_a+g_b+g_c\bigr]
\Bigr) .
\label{eq:subject-diff}
\end{equation}
Proof in Appendix~\ref{app:equality}. Hypothesis (iii) is
load-bearing: without it $\dpJD$ becomes $\tau$-dependent and
Eq.~(\ref{eq:subject-diff}) acquires the extra term
$\kT(\dpJD^{\mathrm{train}}-\dpJD^{\mathrm{shift}})$---the
exclusion of realization-correlated shifts
(Definition~\ref{def:shift}) is exactly what removes it. Scope:
noisy acquisitions (stochastic kernels, degraded copies) generally
violate (F5$'$) and lie outside the identity; the general dyadic
extension is the same open slot as in
Remark~\ref{rem:flat-scope}(v).
\end{corollary}

The per-update retention gap is thus \emph{exactly} the increase of
breakage subjects on the shifted book. The gross acquisition
$\dpJD$ is book-independent. Across books, what changes is the
subject posting of the acquisition \emph{and} the re-valuation of
already-held capital. The $g_a$ difference re-values existing
stock, not the current acquisition. Both enter the identity. On the
device SH, the same $\ln2$ of acquisition is posted as full
capital on $\tau_A$'s book and as full $g_c$ on
$\tau_B$'s. One task's capital is another task's $Y$-contamination,
the ledger version of the $T$-relativity
of capital (Remark~\ref{rem:t-relativity}).

\begin{remark}[Re-posting, and its fence]
\label{rem:reposting}
Within its stated range---flat-task supports, support-fixed shifts
(plus coupling-specified extensions)---%
Corollary~\ref{cor:subject-diff} says that shift-type overfitting
is accounting re-posting: no new breakage subject is needed.
\emph{Outside that range one may be needed}: for the gate-family
shifts of Sec.~\ref{sec:gate} Lemma~\ref{lem:decomp} does not
exist, the value collapse of the threshold theorem is not
decomposed into any subject, and in the budget-bound domain value
is not a linear exchange of information subjects at all (the
enablement premium of Sec.~\ref{sec:regime}); a $\kT\times$
(information subject) currency cannot carry what the shift
destroys there. The ledger decomposition of the gate domain---%
whether an enablement subject exists---is open
(Sec.~\ref{sec:discussion}). This is the precise meaning of
``($\delta$) lies outside the subjects'' in Main
Theorem~\ref{thm:main-III}.
\end{remark}

\begin{corollary}[Orthogonality of efficiency and retention]
\label{cor:orthogonality}
Fix the evaluation convention inside the statement: the numerator
$\Delta\V$ of $\etacap$ and the denominator
$\V_{\mathrm{train}}$ of $\rhogen$ are both evaluated on
$T_{\mathrm{train}}=\delta_{\tau_A}$, and only the numerator
$\V_{\mathrm{shift}}$ of $\rhogen$ varies with
$T_{\mathrm{shift}}$---$\etacap$ is a functional of the training
distribution and carries no shift argument. (Dropping this
convention and evaluating the $\Delta\V$ of $\etacap$ on
$T_{\mathrm{shift}}$ changes the truth value: the realized set
below then collapses to three points, losing $(1,0)$ and
$(\eta_0,0)$.) Then, for every $\eta_0\in(0,1]$ there are a device
and shifts---a shared
world $(X_A,X_B)$, the update a copy of $D_A$ into blank memory,
implemented reversibly ($\sigmaM=\ln2$) or with excess dissipation
$s=\ln2\,(1-\eta_0)/\eta_0$, the retention ratio evaluated against
$T_{\mathrm{shift}}=\delta_{\tau_A}$ (no shift) or
$\delta_{\tau_B}$ (orthogonal shift)---realizing all four pairs
\begin{equation*}
(\etacap,\rhogen)\in
\{(1,1),\,(1,0),\,(\eta_0,1),\,(\eta_0,0)\} ,
\end{equation*}
and partial shifts $T_{\mathrm{shift}}=(\delta,1-\delta)$ realize
$\rhogen=\delta$ continuously: the realized set fills the rectangle
$(0,1]\times[0,1]$. Hence no functional and no monotone relation
between $\etacap$ and $\rhogen$ exists on the device class
($\flatstar$ pairs, $M$-local updates, $\sigmaM>0$):
\emph{``overfitting equals low capitalization efficiency'' is false
in both directions.} Device and numerics in
Appendix~\ref{app:devices}; in the $(1,0)$ corner the four equality
conditions of Main Theorem~\ref{thm:main-II}(iv) hold $T$-a.s.\ on
$T_{\mathrm{train}}=\delta_{\tau_A}$.
\end{corollary}

The orthogonality follows from type. $\rhogen(M')$ is a
functional of $(\pi_{M'},T_{\mathrm{train}},T_{\mathrm{shift}},b)$
alone and never sees the implementation. The denominator of
$\etacap$ contains $\Sigmatot$ and never sees the shift. Two
quantities with these disjoint blind spots are structurally
expected to admit no functional relation. The corollary is the
witness, not the source. Its value is operational. A single device
refutes, in both directions, a conflation that occurs in the
literature. In particular, the $(1,0)$ corner provides a device
that overfits at thermodynamically perfect efficiency: full
capitalization on $T_{\mathrm{train}}$ coexists with total shift
fragility (Fig.~\ref{fig:rectangle}).

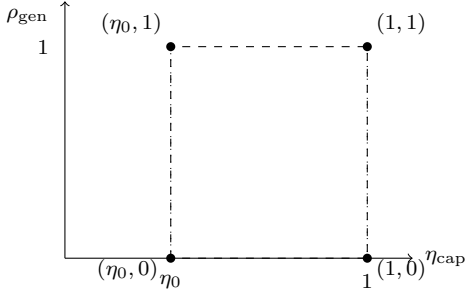
\begin{figure}[t]
\begin{tikzpicture}[scale=1.0]
\draw[->] (0,0) -- (4.6,0);
\draw[->] (0,0) -- (0,3.4);
\node[font=\footnotesize,anchor=west] at (4.65,0)
{$\etacap$};
\node[font=\footnotesize,anchor=east] at (-0.1,3.2)
{$\rhogen$};
\node[font=\footnotesize,anchor=north] at (4.0,-0.1) {$1$};
\node[font=\footnotesize,anchor=east] at (-0.1,2.8) {$1$};
\node[font=\footnotesize,anchor=north] at (1.4,-0.1)
{$\eta_0$};
\draw[dashed] (1.4,0) rectangle (4.0,2.8);
\fill (4.0,2.8) circle (0.06);
\fill (4.0,0) circle (0.06);
\fill (1.4,2.8) circle (0.06);
\fill (1.4,0) circle (0.06);
\node[font=\footnotesize,anchor=south west] at (4.0,2.85)
{$(1,1)$};
\node[font=\footnotesize,anchor=north west] at (4.0,0.1)
{$(1,0)$};
\node[font=\footnotesize,anchor=south east] at (1.4,2.85)
{$(\eta_0,1)$};
\node[font=\footnotesize,anchor=north east] at (1.35,0.1)
{$(\eta_0,0)$};
\draw[dotted] (4.0,0) -- (4.0,2.8);
\draw[dotted] (1.4,0) -- (1.4,2.8);
\end{tikzpicture}
\caption{\label{fig:rectangle}%
Rectangle witness in the $(\etacap,\rhogen)$ plane
(Corollary~\ref{cor:orthogonality}): for every $\eta_0\in(0,1]$
all four corners are realized by one device family, and partial
shifts fill the vertical segments continuously, so the realized
set contains the rectangle $(0,1]\times[0,1]$. No functional or
monotone relation between efficiency and retention exists on this
class; the $(1,0)$ corner is overfitting at perfect
capitalization efficiency.}
\end{figure}

\section{Gate-family application}
\label{sec:gate}

The dials $d$
(degeneracy) and $c$ (transfer) are control variables of value
retention \emph{in the gate-learning family of
Ref.~\cite{Sudo_Spectrum}}. The dial $c_0$ (effective
identification rate) enters only conjecturally, as a recovery
price (Sec.~\ref{sec:gate-outlook}); none of the three is claimed
to be an invariant of general distribution shift. On a common
horizontal axis, the main claim contrasts per-bit linear
retention in the $\flatstar$ family with a threshold transition
in the gate family.

\subsection{The $\flatstar$ baseline: linear}
\label{sec:gate-linear}

\begin{proposition}[TV-linear control on flat supports]
\label{prop:tvlinear}
If the support of $T$ consists of flat tasks (support-fixed shift),
then $\Lgen(M)\le\lVert T_{\mathrm{train}}-T_{\mathrm{shift}}
\rVert_{\mathrm{TV}}\cdot\sup_\tau\Gap(M;\tau,b)$, with equality on
the device SH (Main Theorem~\ref{thm:main-III}(ii)); and along the
mixing path $T_t=(1-t)\,T_{\mathrm{train}}+t\,T_{\mathrm{shift}}$
the value $\V$ is an affine function of $t$---no discontinuous
jumps on support-fixed paths. (The bound is Main
Theorem~\ref{thm:main-I}(i)(b); the equality case is Main
Theorem~\ref{thm:main-III}(ii).)
\end{proposition}

The bound alone does not exclude threshold shapes in the
\emph{ratio} form. The slope of $\rhogen$ against TV is
$\sup_\tau\Gap/\V_{\mathrm{train}}$; it is device-dependent and
unbounded. A state whose capital concentrates on light tasks can
reach $\rhogen=0$ under an arbitrarily small TV movement. The
same-axis pair carries the shape contrast:

\begin{proposition}[$\flatstar$ redraw: exact per-bit linearity]
\label{prop:perbit}
Let $\tau$ be a flat task with $X$ a single register of
$k_{\mathrm{eff}}$ uniform bits, $M$ a perfect copy of $X$, and
let the shift redraw $k_{\mathrm{res}}=k_{\mathrm{eff}}-c$ bits of
$X$ independently and uniformly (support-extending; coupling
specification: the persisting part matches, the remainder is
independent). Then
\begin{equation}
\begin{aligned}
\V_{\mathrm{train}} &= k_{\mathrm{eff}}\,\kT\ln2,\\
\V_{\mathrm{shift}} &= c\,\kT\ln2,\\
\rhogen &= \frac{c}{k_{\mathrm{eff}}}
=1-\frac{k_{\mathrm{res}}}{k_{\mathrm{eff}}} :
\end{aligned}
\label{eq:perbit}
\end{equation}
retention is exactly per-bit linear in the redraw depth
$k_{\mathrm{res}}$. (Arithmetic of Corollary~\ref{cor:v-identity}
on the copy state; numerics in Appendix~\ref{app:verification}.)
\end{proposition}

\subsection{The gate family: key redraw and the threshold}
\label{sec:gate-threshold}

The gate device is realized in \emph{Type II} form. This
declaration matters and is part of the device definition. In the
gate imported from Ref.~\cite{Sudo_MPU}, the key can either sit on
the initial value of a manipulable medium register (Type~I) or be
wired into a frozen, energy-decoupled parameter $\Theta$. The gate
passes iff the submitted $B$ equals $s(\Theta)$, with the medium $B$
starting uniform (Type~II). For lower bounds, Type~I is safe: the
conditional-extraction channel it opens only strengthens the
informed side (cf.\ Main Theorem~\ref{thm:main-III}(iii)). The
\emph{upper} bounds that carry the present theorem, however, require
Type~II. Because $\Theta$ is frozen and untouchable, the
$M$--$\Theta'$ correlation opens no extraction channel, and the
imported caps hold as stated (Appendix~\ref{app:sbar}).

\emph{Device KR (key-redraw family).} This Type~II gate has key
length $k$, prize $F_{\mathrm{cart}}$, horizon $K_{\mathrm{hor}}$,
and budget $b$. The pass map $s$ is uniformly $d$-fold
degenerate. Because the agent never touches the labels of $\Theta$
(Type~II), the normal-form reduction of Ref.~\cite{Sudo_Spectrum}
applies, and the device reduces to an effective key of
$k_{\mathrm{eff}}=k-\log_2 d$ bits. The shift $\Theta\to\Theta'$
preserves $c$ bits of the \emph{effective} coordinates and redraws
the remaining
\begin{equation}
k_{\mathrm{res}}:=k_{\mathrm{eff}}-c=k-\log_2 d-c
\label{eq:kres}
\end{equation}
independently and uniformly. The memory is $M=s(\Theta)$, a copy
of the pass state. The coupling specification (persisting bits
match; remainder and degeneracy coordinates independent) is part of
the device. Equation~(\ref{eq:kres}) is an \emph{import} of the
composed-exemption synthesis of Ref.~\cite{Sudo_Spectrum} in its
compatible case: $c_{\mathrm{eff}}=c$, and the transfer split refines
the degeneracy split. Here $c\le k_{\mathrm{eff}}$ by construction,
so the truncation $\min\{c,k_{\mathrm{eff}}\}$ is inactive. The
closed form in general position is open
there and is not claimed here.

\begin{theorem}[Gate threshold transition, Type II]
\label{thm:gatethreshold}
Write $C_G$ for the bookkeeping remainder of the imported gate
bounds (Sec.~\ref{sec:regime}): a constant assumed finite and
uniform in
$(k,k_{\mathrm{eff}},k_{\mathrm{res}},\beta b,
\beta F_{\mathrm{cart}},K_{\mathrm{hor}})$, whose proven numerical
upper bound is not available from the companion. Assume two
hypotheses:
\emph{(i) full suppression},
\begin{equation}
k_{\mathrm{eff}}\ln2\ \ge\ \beta b
+\ln\bigl(\beta F_{\mathrm{cart}}\,K_{\mathrm{hor}}\bigr)
\iff
\beta F_{\mathrm{cart}}\,c_{\mathrm{full}}\le1,
\label{eq:fullsup}
\end{equation}
with $c_{\mathrm{full}}:=K_{\mathrm{hor}}2^{-k_{\mathrm{eff}}}
e^{\beta b}$; and \emph{(ii) positivity},
\begin{equation}
\beta F_{\mathrm{cart}}\,(1-c_{\mathrm{full}})-\beta b-2\,C_G>0 ,
\label{eq:positivity}
\end{equation}
which certifies $\V_{\mathrm{train}}>0$ and hence the domain of
$\rhogen$ (Definition~\ref{def:retention}). Hypothesis (i) alone
does not: $(k_{\mathrm{eff}},\beta b,\beta F_{\mathrm{cart}},
K_{\mathrm{hor}})=(100,50,10,3)$ satisfies
Eq.~\eqref{eq:fullsup} while the lower bound of part (1) reads
$-40-2C_G<0$. Then on the device KR:
\begin{enumerate}
\item[(1)] $\V_{\mathrm{train}}\in
\kT\,\bigl[\beta F_{\mathrm{cart}}(1-c_{\mathrm{full}})-\beta b
-2\,C_G,\ \beta F_{\mathrm{cart}}+\beta b+C_G\bigr]$, with the
lower end positive by (ii);
\item[(2)] $\V_{\mathrm{shift}}\le\kT\,\bigl[\beta
F_{\mathrm{cart}}\min\bigl(1,K_{\mathrm{hor}}
2^{-k_{\mathrm{res}}}e^{\beta b}\bigr)+\beta b+C_G\bigr]$;
\item[(3)] $\V_{\mathrm{shift}}\ge\kT\,\bigl[\beta
F_{\mathrm{cart}}\min\bigl(1,K_{\mathrm{hor}}
2^{-k_{\mathrm{res}}}\bigr)-\beta F_{\mathrm{cart}}
\min(1,c_{\mathrm{full}})-\beta b-2\,C_G\bigr]$;
\item[(4)] consequently, as a function of $k_{\mathrm{res}}$: on
the side $k_{\mathrm{res}}\le\log_2 K_{\mathrm{hor}}$,
\begin{equation*}
\begin{aligned}
\rhogen
&\ge \frac{\beta F_{\mathrm{cart}}-1-\beta b-2C_G}
{\beta F_{\mathrm{cart}}+\beta b+C_G}\\
&=1-\frac{1+2\beta b+3C_G}
{\beta F_{\mathrm{cart}}+\beta b+C_G}\,;
\end{aligned}
\end{equation*}
for $k_{\mathrm{res}}\ge\log_2\bigl(K_{\mathrm{hor}}\,\beta
F_{\mathrm{cart}}\,e^{\beta b}\bigr)+t$,
\begin{equation*}
\rhogen\ \le\
\frac{2^{-t}+\beta b+C_G}
{\beta F_{\mathrm{cart}}(1-c_{\mathrm{full}})-\beta b-2C_G}\,.
\end{equation*}
The two shoulders are separated by
$\log_2(\beta F_{\mathrm{cart}})+\beta b/\ln2+t$ bits on the
$k_{\mathrm{res}}$ axis (the difference of the two boundary
positions; the $\log_2K_{\mathrm{hor}}$ terms cancel, and $t$ is
the free right-shoulder margin parameter, not a rounding
term)---a
threshold transition, on the same axis $k_{\mathrm{res}}$ on which
Proposition~\ref{prop:perbit} is exactly linear.
\end{enumerate}
Proof in Appendix~\ref{app:sbar}. Two premises carry it
(Appendix~\ref{app:sbar}): the posterior of the redrawn key given
the memory is uniform on the persisting fiber (exact, since
$\Theta$ is uniform and the persisting set is chosen independently
of $\Theta$), and the degeneracy is uniformly $d$-fold with Type~II
untouchable labels (non-uniform degeneracy is not treated, as in
Ref.~\cite{Sudo_Spectrum}).
\end{theorem}

\begin{proposition}[The $d$ dial, compatible case]
\label{prop:ddial}
As a comparison of device families with the \emph{effective-%
coordinate} persistence $c$ held fixed,
$k_{\mathrm{res}}=(k-\log_2 d)-c$ decreases in $d$, and the bounds
of Theorem~\ref{thm:gatethreshold} move toward retention as $d$
grows. Numerically, at
$(k,c,\beta F_{\mathrm{cart}},K_{\mathrm{hor}},b)=(44,36,1000,3,0)$
under the script convention $C_G=5.0$ of
Appendix~\ref{app:verification} (a convention, not a proven bound
on $C_G$; the numbers are substitution arithmetic, an
illustration): the lower bound on $\rhogen$ at
$d=16$ is $0.176$, above the upper
bound $0.017$ at $d=1$ (Appendix~\ref{app:devices}). \emph{This
separation of bounds is a numerical fact at the stated parameters
and stated convention,
not a general theorem}---the crossing of the bounds is parameter
dependent. The comparison is well posed only in effective
coordinates: fixing the persistence in raw key coordinates is
ill-posed under changes of $d$, and a transferred bit that lands in
a null direction of the degeneracy is transferred and worthless
(as recorded in Ref.~\cite{Sudo_Spectrum}); the exact form of
``degeneracy helps retention'' is that \emph{for persistence
transverse to the degeneracy directions}, coarser required
precision shortens the effective redraw depth.
\end{proposition}

Raising the budget dial $b$ shifts the transition window to the right by
$\beta b/\ln2$ bits (Appendix~\ref{app:devices}). The residual
guesswork of the partially informed side therefore becomes
purchasable.

\subsection{Outlook: the recovery price}
\label{sec:gate-outlook}

The dial $c_0$ of Ref.~\cite{Sudo_Spectrum} measures effective
identification leakage. It plays no role within the non-leaky gate
class used above. Its expected role concerns \emph{leaky} variants,
and we state this role strictly as a conjecture.

\begin{conjecture}[$c_0$ as a recovery price]
\label{conj:c0}
Consider leaky variants of the device KR in which a failed probe
discloses more than the membership bit---a family \emph{outside}
the non-leaky gate class of Ref.~\cite{Sudo_MPU}---and the
post-shift recovery problem: probe-charged update sequences
returning $\V$ from the collapsed right shoulder of
Theorem~\ref{thm:gatethreshold} to the training level. Then, in
the regime $K_{\mathrm{hor}}\ll m_{\mathrm{ident}}
(k_{\mathrm{res}})$ where the evaluation run itself cannot progress
identification: (i)~under membership-only disclosure the recovery
requires $\Omega(2^{k_{\mathrm{res}}})$ probes (the
channel-level-uninformative exponential invariance of
Ref.~\cite{Sudo_Spectrum}, translated to $\V$ through the
companion envelope); (ii)~under a constant identification rate
$c_0$, $O(k_{\mathrm{res}}/c_0)$ probes suffice (upper bound; the
matching lower bound is proved in Ref.~\cite{Sudo_Spectrum} only
for $\theta$-transitive devices, and fails for general devices by
explicit counterexample); (iii)~hence $c_0$ controls the
\emph{recovery price} while leaving the passive retention ratio
$\rhogen$ unchanged. The regime restriction in (iii) is essential:
when $K_{\mathrm{hor}}$ exceeds the identification horizon, leakage
feeds back into passive retention itself.
\end{conjecture}

We deliberately state this as a conjecture for two reasons.
First, the accounting convention that merges probe charges with the
search ledger $\Sigma_{\mathrm{search}}$ is a design shared with
the continual-learning framework (a companion in preparation,
Sec.~\ref{sec:discussion}). This convention should be fixed once,
not twice. Second, a work-ledger converse for leaky gates poses a
genuinely new proof obligation. The imported unlock bound is proved
for the non-leaky class. By contrast, the two leakage theorems of
Ref.~\cite{Sudo_Spectrum} are statements in probe-count and
guesswork currency, not work-ledger statements. We also record that
the imported oracle theorem is one-directional. An exponential
recovery price implies exponentially small $c_0$, but $c_0\to0$
does \emph{not} imply exponential cost. A bad-set device with
$c_0\to0$ and cheap recovery is exhibited in
Ref.~\cite{Sudo_Spectrum}. $c_0$ is not asserted to be a
theorem-level control variable of retention.

\section{Related work}
\label{sec:related}

Neighboring lineages supply ingredients that this paper imports,
while differing from it in identifiable ways. We collect the
positive imports and differences here and in
Tables~\ref{tab:correspondence} and~\ref{tab:claims}. We repeat the
single dated and scoped novelty statement from
Sec.~\ref{sec:intro}: within the fence of prior work surveyed up to
2026-07-12, we did not find a prior framework simultaneously
equipped with the same operational deletion value, evaluation
budget, access structure, search ledger, and four-axis alignment
map. Each ingredient separately has close relatives, tabulated
below.

\emph{Goldt--Seifert.}---The stochastic thermodynamics of
learning~\cite{Goldt2017PRL,Goldt2017NJP} established the
memory-side entropy-production ledger and its bound on information
acquisition. Their subsystem account $\Delta S(\omega)+\Delta Q$ is
the direct ancestor of $\sigmaM$ (Definition~\ref{def:sigmaM}), and
their efficiency $\eta\le1$ is the syntactic-bits analogue of
Main Theorem~\ref{thm:main-II}(iii). The numerator differs. That
work uses acquired information, whereas ours uses a deletion-based
task value with explicit $T$, budget, and access arguments.
Consequently, ``acquired but never capitalized'' (device LB$+$diss)
is separated by a theorem rather than by the choice of numerator
(Sec.~\ref{sec:regime}, B3).

\emph{Horowitz--Esposito, Hartich--Barato--Seifert.}---The
universal ledger identity of Lemma~\ref{lem:ledger-identity} is the
static-source degeneration of their bipartite information-flow
balances~\cite{Horowitz2014,Hartich2014} plus ancilla bookkeeping.
It is imported, not claimed. On that axis, we connect it to the
value-side breakage decomposition (Lemma~\ref{lem:decomp}) and
identify condition (f).

\emph{Shettell--Auff\`eves.}---Ref.~\cite{Shettell2026} defines a
bits-layer inferential efficiency (information gain over cumulative
memory erasure cost) and proves $\eta\le1$ with an equality
condition. It also decomposes the inefficiency into two correlation
subjects. At the bits layer, these results provide the closest
existing analogue of our $\le1$-plus-equality-plus-subjects
package. Their two subjects are the bits-layer analogues of our (b)
waste and (c) $Y$-contamination. That framework lacks the value
layer (informed--blind counterfactual, $T$, $b$), the forgetting
subject $g_a$, and the separation of implementation reversibility
($\Sigmatot=0$) as an independent equality condition. At the bits
layer, information gained but never capitalized cannot be
penalized. Our framework includes all three.

\emph{Kolchinsky--Wolpert.}---Semantic
information~\cite{Kolchinsky2018Semantic} measures the value of
information by a work difference on tasks. This is the lineage in
which $\V$ itself stands, together with the decision-theoretic value
of information. The counterfactual type differs: their
scramble-then-continue versus our deletion followed by
re-optimization from scratch. Their formulation has no budget slot.
The multiplier/ratio pair also differs. Their ``bang-per-bit''
$\kappa>1$ reappears in our regime map as the
gate-plus-finite-budget cell and collapses to $\le1$ in the
$\flatstar$ cell (Table~\ref{tab:correspondence}). Their mismatch
cost~\cite{Kolchinsky2017Mismatch} is the static cost-side neighbor
of Sec.~\ref{sec:retention}. In that work, a change in the input
distribution of a \emph{fixed process} changes the process's
\emph{dissipation}. Here, a movement of the \emph{task distribution}
changes the informed--blind \emph{asset value}, with the budget as
an argument. A distribution moves in both; the kind of distribution
and the charged quantity differ.

\emph{Boyd--Crutchfield--Gu(--Binder).}---Thermodynamic machine
learning~\cite{Boyd2022TML} established the affinity of extracted
work and log-likelihood. It also established the equivalence of
maximum-work and maximum-likelihood training inside its setting.
The sequel on thermodynamic
overfitting~\cite{Boyd2025Overfitting} carries the same-source test
analysis. Its exact asymptotic work rate (their Theorem~1) is usable
as a training-side baseline in our Proposition~\ref{prop:corresp}
settings. Both are positive imports, and
Proposition~\ref{prop:corresp} isolates the structural features
under which their equivalence holds
(Remark~\ref{rem:corresp-limits} lists the six formal differences).
What this paper adds is orthogonal to their axis. Their overfitting
is the small-sample failure mode within one source.
Ref.~\cite{Boyd2022TML} states, of the long-term effectiveness of
thermodynamic learning, ``we leave analyzing the long-term
effectiveness of thermodynamic learning to the future''. By
contrast, Main Theorem~\ref{thm:main-III}(ii) exhibits the
complementary mode: perfect fit, collapse under shift alone. This
mode requires axes for shift, budget, access, and record/world
separation. They are not expressible in the unipartite formalism,
which carries no access partition and no record/world distinction
(Remark~\ref{rem:corresp-limits}).

\emph{Still and successors; modularity.}---The dissipative price of
nonpredictive correlation~\cite{Still2012}, the cost--benefit ledger
of memory~\cite{Still2020}, and partially observable Szilard
engines~\cite{Still2021Partial,Daimer2023Observer} are the nearest
access-side relatives. The modularity theorem of
Ref.~\cite{Boyd2018Modularity} concerns correlation value destroyed
by access structure. Together, these works already exhibit the
\emph{phenomenon} that correlation can fail to convert into work.
The precise form of our separation claim therefore concerns the
statement, not the phenomenon. Main
Theorem~\ref{thm:main-I}(ii)--(iii) quantifies a
deletion-and-re-optimization value, per-task re-optimization, a
task \emph{distribution}, and a budget slot. None of those
frameworks contains this combination, so the theorem is not
expressible there (Sec.~\ref{sec:value-main}).

\emph{Value of data for bounded agents.}---Algorithmic
catalysis~\cite{Perrier2025PartI,Perrier2026PartII} is the closest
system to the capitalization concept: reusable computational
structure that lowers future costs. Its value is a cost reduction,
not an informed--blind work difference on a task distribution, and
budget, access, and replacement vocabulary are absent.
Epiplexity~\cite{Finzi2026Epiplexity} is the machine-learning-side
representative of value-of-data for computation-bounded observers.
It is the neighbor of the scope remark of
Sec.~\ref{sec:value-updates}. That remark concerns what re-encoding
is worth, precisely what $\V$ does not price. The bits-per-joule
accounting of Ref.~\cite{Takahashi2026Limits} shares the
degeneracy-of-%
accounting concern (Sec.~\ref{sec:discussion}) but has no value
currency or exchange laws. Ref.~\cite{Hasegawa2026Classifiers}
relates Bayes error to entropy production and activity. This
relation is an inference-side exchange law adjacent to our C2/C4
axis.

\emph{Asymptotics and run-time ledgers.}---Universal work
extraction~\cite{Watanabe2025Universal} shows informed and blind
optimal work rates coincide in the asymptotic i.i.d.\ limit. This
result underwrites the placement of $\V$ on the finite-run,
non-i.i.d., budgeted side. The regret--dissipation identity of
Ref.~\cite{Lumbreras2026RL} is a run-time excess-EP ledger
(informed-minus-achieved on a single fixed source). It occupies a
different slot from the update-side $\sigmaM$ and has no
informed--blind counterfactual, $T$, $b$, or capital side. The
minimax redundancy pricing of Ref.~\cite{Touzo2020MDL} is the
static worst-case cost-side neighbor. The resource-theoretic
interconvertibility of Ref.~\cite{Brandao2013Resource} is the
standing asymptotic contrast. The finite-run K1/K2 route asymmetry
of Sec.~\ref{sec:regime} is always to be read against this contrast.
Refs.~\cite{Caraffa2026Dissipative,Fiderer2025WorkCapacity} are
handled in the terminology defense of
Sec.~\ref{sec:retention-defs}. The activity axis of
Ref.~\cite{Fiderer2025WorkCapacity} is the recorded fourth slot of
Sec.~\ref{sec:value-blackwell}.

\emph{Honesty box.}---The inequality of Main
Theorem~\ref{thm:main-II}(iii) is a composition of
Sagawa--Ueda-type~\cite{Sagawa2010Generalized} and
Goldt--Seifert-type~\cite{Goldt2017PRL} steps and is not claimed as
new. The universal ledger identity is
imported~\cite{Horowitz2014,Hartich2014}. The collapse theorems of
Main Theorem~\ref{thm:main-III} are arithmetically elementary given
the extraction identity and the imported gate bounds. The claimed
contributions are: the operational quantity $\V$ (deletion plus
re-optimization, with $T$, access, and budget as arguments) and its
separation theorem at the stated quantifier; the accounting design
$\sigmaM$ with condition (f), the equality conditions, the regime
map, and the blank-start cumulative bound; and the alignment map
(problem setting, four axes with device witnesses, and the
subject-difference identity connecting them to the ledger).

\begin{table*}[t]
\caption{\label{tab:correspondence}%
Correspondence of efficiency-type quantities. GS:
Ref.~\cite{Goldt2017PRL}; KW: Ref.~\cite{Kolchinsky2018Semantic};
SA: Ref.~\cite{Shettell2026}; the last two rows are the adjacent
run-time and asymptotic results,
Refs.~\cite{Lumbreras2026RL,Watanabe2025Universal}. Lineages that
define no efficiency-type ratio (BCG/BCGB, Perrier) are compared in
Sec.~\ref{sec:related} and Table~\ref{tab:claims}.}
\begin{ruledtabular}
\begin{tabular}{p{1.9cm}p{3.0cm}p{2.9cm}p{2.4cm}p{2.6cm}p{2.4cm}}
Quantity & Numerator & Denominator & Direction & Counterfactual;
budget slot & ``Acquired but unusable'' \\
\colrule
$\etacap$ (this work) & $\Delta\V$: increment of informed$-$blind
optimal work on $T$ & measured search ledger $\sigmaM$
(implementation dependent) & $\le1$ is a \emph{regime theorem};
the failures outside $\flatstar{+}\sigmaM{+}$(f) are signatures &
deletion $+$ re-optimization from scratch; budget: yes ($b$ an
argument of $\V$) & separated by a \emph{theorem} ($\etacap=0$
from the extraction identity) \\
GS $\eta$ & acquired information $I(\sigma_T{:}\sigma)$
(syntactic bits) & weight-side EP $\Delta S(\omega)+\Delta Q$
(direct ancestor of $\sigmaM$) & $\le1$ (theorem in their
setting) & none (one real process); budget: no & built into the
definitional choice of the numerator \\
KW $\kappa_{\mathrm{stored}}$ & stored-information viability
difference & \emph{minimal} acquisition work $\kT\ln2\cdot I$ &
$>1$ possible (multiplier, ``bang-per-bit'') & scramble, then the
same protocol continues; budget: no & same \\
KW $\eta_{\mathrm{stored}}$ & stored semantic information &
acquired information $I_p(X_0;Y_0)$ & $\le1$ & same; budget: no &
same \\
SA $\eta$ & inferential information gain & cumulative memory
erasure cost (nearest to $\sigmaM$) & $\le1$ by construction,
with equality condition & none (one real process); budget: no &
same \\
regret--dissipation \cite{Lumbreras2026RL} & cumulative regret &
cumulative run-time dissipation (an identity) & identity &
informed minus achieved on one fixed source; budget: no & run-time
slot; no capital side \\
asymptotic collapse~\cite{Watanabe2025Universal} &
\multicolumn{2}{l}{informed $=$ blind optimal work rates in the
i.i.d.\ limit} & collapse & asymptotic; budget: no & underwrites
the finite-run, budgeted domain of $\V$ \\
\end{tabular}
\end{ruledtabular}
\end{table*}

\begin{table*}[t]
\caption{\label{tab:claims}%
Claims, hypotheses, breaking witnesses, and prior attribution.
Devices in Appendices~\ref{app:devices} and~\ref{app:sbar}.}
\begin{ruledtabular}
\begin{tabular}{p{3.5cm}p{4.1cm}p{3.9cm}p{4.5cm}}
Claim & Hypotheses & Breaking / witness devices & Prior
attribution \\
\colrule
Separation (Main Thm.~\ref{thm:main-I}(ii)(iii)) & $M$-local
updates; draw exogeneity; $\forall n\,\exists$ quantifier & LB
(witness); cartridge variant & phenomenon-level relatives:
Refs.~\cite{Still2012,Kolchinsky2018Semantic,Boyd2018Modularity};
the statement's quantifier is not expressible there \\
$\flatstar$ extraction identity (Main
Thm.~\ref{thm:main-II}(i)) & (F1)--(F5$'$); $q=p$ & gates break it
(G); deep dyadic atoms break it at small $b$ & SU
lineage~\cite{Sagawa2010Generalized}; per-task equality,
$b$-independence, pathwise domain here \\
Ledger identity (Main Thm.~\ref{thm:main-II}(ii)) & any $M$-local
update $+$ implementation & --- (identity) & static-source
degeneration of Refs.~\cite{Horowitz2014,Hartich2014} \\
$\etacap\le1$ (Main Thm.~\ref{thm:main-II}(iii)) & $\flatstar$
support; (F5$'$)-stable $M$-local update; (f); $\sigmaM$ account &
E ($\Sigmatot$ account); G
(gate$+$finite $b$); CX ((f) broken) & composition of
Refs.~\cite{Sagawa2010Generalized,Goldt2017PRL}; the regime and
its boundary are the content \\
Equality conditions (Main Thm.~\ref{thm:main-II}(iv)) & as above &
E (achievement); W: (b); F: (a)$+$(c); Y: (c); E$+$ex: (e)
divergence & bits-layer
analogue in Ref.~\cite{Shettell2026}; value layer, forgetting
subject, reversibility separation here \\
Blank-start cumulative (Main Thm.~\ref{thm:main-II}(v)) & blank
start; flat support & CX two-stage ($\eta_{\mathrm{cum}}=1$) &
--- \\
Alignment schema (Main Thm.~\ref{thm:main-III}) & two-layer
domain: ($\alpha$)($\beta$)($\gamma_{\mathrm{acc}}$)($\delta$) $+$
(F5$'$)-stability for the conditional rate; joint neutrality
$(M,D)\perp Y$ for the raw layer & one-time-pad boundary witness;
SH/BG/AC$'$/RJ (paired interventions, one coordinate each; AC$'$
also the $\gamma$-necessity witness) & equivalence inside their
setting:
Ref.~\cite{Boyd2022TML}; same-source overfitting:
Ref.~\cite{Boyd2025Overfitting}; shift/budget/access/record axes
here \\
Subject-difference identity
(Cor.~\ref{cor:subject-diff}) & common flat support; two-state
(F5$'$); exogeneity over $\mathcal{T}$ & realization-correlated
shifts excluded & --- \\
Orthogonality (Cor.~\ref{cor:orthogonality}) & flat pairs;
$\sigmaM>0$ & rectangle device family & refutes a conflation
current in the literature, both ways \\
Threshold transition (Thm.~\ref{thm:gatethreshold}) & Type II KR;
full suppression $+$ positivity; uniform finite $C_G$; compatible
case & leaky variants (outside scope;
Conjecture~\ref{conj:c0}) & gate bounds imported from
Ref.~\cite{Sudo_MPU} (Appendix~\ref{app:fixedI}, incl.\ B3);
$k_{\mathrm{res}}$ imported from Ref.~\cite{Sudo_Spectrum} \\
\end{tabular}
\end{ruledtabular}
\end{table*}

\section{Discussion and limitations}
\label{sec:discussion}

\emph{Not a theory of statistical generalization.}---Every quantity
in this paper is a work value on a stated task distribution, under
stated budget and access. The retention quantities of
Sec.~\ref{sec:retention} compare two such values. No
sample-complexity, risk-bound, or i.i.d.-asymptotic content is
claimed. Neither $\Lgen$ nor $\rhogen$ reduces to or estimates a
generalization error. We make no claim that value retention
explains or predicts the statistical generalization of any learning
system. The correspondence with machine-learning phenomena asserted
anywhere in this paper is a correspondence of accounting structure
(Sec.~\ref{sec:intro}).

\emph{The U-restriction, and two self-assessments.}---All protocol
classes carry the imported Assumption-U restriction
(Sec.~\ref{sec:setting-asterisk}). The reduction question
remains open in the companion framework. Thus, $\V$ is a
restricted-class--relative quantity, and every interpretation-level
statement in this paper (``capital'', ``full capitalization'')
carries that asterisk. The theorems are unconditional on the
restricted class. Two further self-assessments delimit what $\V$
measures. \emph{Budget-latent learning}: updates exist with
$\Delta\V(b)=0$ but $\Delta\V(b')>0$ for $b'>b$. Such an update acquires
correlation whose cashing route is closed at the evaluation
budget. The learning predicate of
Definition~\ref{def:learning} is relative to $b$
(Sec.~\ref{sec:regime} is built from this fact). \emph{Preparedness
has value zero}: because the blind side knows the public task tuple
and re-optimizes per task, and because draws are exogenous, the knowledge
of $T$, of the statistics of the world, and of pre-positioned
strategies is constitutively free. $\V$ does not measure it. A
variant that priced preparedness would change the blind
counterfactual. Such variants, like risk-averse (non-affine)
functionals of $T$ (Remark~\ref{rem:t-relativity}), are explicit
non-goals of this paper.

\emph{$\flatstar$ scope.}---The exactness package of
Sec.~\ref{sec:ledger} is proved on the $\flatstar$ class. Its
boundary is recorded rather than smoothed over
(Remark~\ref{rem:flat-scope}): multi-register manipulable sets with
internal correlation are out of scope; partial-read variants of
(F4) are undeveloped; outside (F5$'$) the extraction identity fails
at small budgets (a budget-floor restoration for general dyadic
states is a recorded open slot); and on gate tasks both the
identity and the linear bound fail. The map of Table~\ref{tab:regime}
charts that enablement regime rather than hiding it.

\emph{The gate ledger is open.}---Within flat supports,
Corollary~\ref{cor:subject-diff} reduces shift-type overfitting to
re-posting of the three ledger subjects. For the gate family, no
per-task decomposition exists in this paper. The collapse of
Theorem~\ref{thm:gatethreshold} is not attributed to any subject.
In the budget-bound domain, value carries an enablement premium
that no $\kT\times$(information subject) currency can express.
Whether an enablement subject can be defined, and thereby yield a
ledger decomposition of the gate domain, is open.

\emph{Repeated-deployment returns are untreated.}---The cumulative
bound of Main Theorem~\ref{thm:main-II}(v) is an acquisition-side
ledger over update histories. It is not a return-on-investment
statement for repeated deployment of the same capital. A
``thermodynamics of deployment revenue'', defined as multi-run evaluation
with task and budget sequences that extends the transfer inventory of
Ref.~\cite{Sudo_Spectrum}, is a future extension, contingent on
the continual framework below. Here we only mark the slot.

\emph{Accounting degeneracy (boundary paragraph).}---A companion
line of work, kept out of this paper by design, concerns the
degeneracy of the accounting itself. On macroscopically
implementable device classes, the information accounting admits
degeneracies under which the observed learning function or
statistical performance does not identify the physical dissipation.
Under these conditions, no thermodynamics-specific term appears
in the predictive structure of learning theory, and dissipation is
non-identifiable from function alone. We assert here neither a new
universal energy floor, nor the existence of arbitrarily low-cost
implementations, nor any claim about the energy costs of real AI
systems. The full treatment, with its own fence, is deferred to a
separate short report.

\emph{Continual learning.}---Sequences of tasks, budgets, and
updates (overwriting, replay, transfer, and the reachability of
value under sustained operation) are the subject of a companion
in preparation. The present $\V$ is the single-evaluation
functional that that framework iterates.

\appendix

\section{Devices and numerical verification}
\label{app:devices}

\subsection{Device LB (ledger-blocked correlation)}
\label{app:dev-lb}

\emph{Registers.} $X_1,X_2$: one uniform bit each, independent (the
working media of the two tasks). $Y=(Y_1,\dots,Y_n)$: i.i.d.\ fair
bits independent of $(X_1,X_2)$---world registers outside every
manipulable set. $M_{\mathrm{core}}$: a perfect copy of $X_1$ (one
bit of genuine capital, so that $\Delta\V=0$ is exhibited in a
nondegenerate situation with $\V>0$). All energies degenerate
($E\equiv0$).

\emph{Tasks.} $\mathcal{T}=\{\tau_1,\tau_2\}$ in the type of
Definition~\ref{def:task}: $\tau_1=(W=(X_1,X_2,Y)$ [shared world],
$p(w)=$ uniform product, $E\equiv0$, $G=\emptyset$,
$\mathrm{Man}=\{X_1\}\cup\text{ancillas}$, $K_\tau=3$, static
all-read$)$; $\tau_2$ identical with
$\mathrm{Man}=\{X_2\}\cup\text{ancillas}$. $T_n=(1/2,1/2)$;
$b\ge0$ arbitrary (the informed protocols require no upfront
investment; the budget axis is deliberately not exercised here).
Giving the access structure a $Y$ read port changes nothing:
$X_1\perp Y$, so adding $Y$ to the read set does not lower
$H(X_1\mid\text{read})$---the read-family scan of the verification
confirms this invariance.

\emph{Record and update sequence.} $D:=$ the observation record of
$Y$ ($D_j$ a copy of $Y_j$; read-only during updates;
$I(D;W)=n\ln2>0$). Updates:
$M_k:=(M_{\mathrm{core}},\text{copies of }Y_1,\dots,Y_k)$, each
step one controlled-copy from $D_k$ into fresh memory (an
admissible growth update).

\emph{Exact gap values.} For each $\tau_i$: informed supremum
$=\kT[\ln\lvert X_i\rvert-H(X_i\mid M_k)]=\kT\,I(M_k;X_i)$, blind
supremum $=0$. The converse (no admissible protocol exceeds these)
is the three-step argument: (1)~the conditional chain bound
(Lemma~\ref{lem:condchainprime}) applied to the $M=m$ conditional
law---in this device all variables are independent given $m$
($X_1\mid m$ deterministic or uniform; $X_2,Y$ independent of
$m$), so only the $X_1$ component of the conditional total
correlation contributes; (2)~invariance of non-manipulable
marginals: designed mechanisms act on
$\mathrm{Man}\cup\text{ancillas}$ only (mechanism locality),
$G=\emptyset$, so the $Y$-joint and the other $X$ are untouched
---correlation-creating mechanisms can be written, but their
account-credit is exactly offset in the ledger, which the chain
bound disposes of wholesale; (3)~auxiliary netting: imported
auxiliaries net $\le0$ under the terminal-balance convention
($\beta$) (Lemma~\ref{lem:auxnetting}). Achievability: a
conditional permutation ($X_1$ conditioned on
$M_{\mathrm{core}}$) followed by Szilard extraction.

\emph{Consequences.} $I(M_k;D)=k\ln2$ strictly increasing;
$I(M_k;W)=(k+1)\ln2$; $\Gap(M_k;\tau_1,b)=\kT\ln2$ ($k$-%
independent), $\Gap(M_k;\tau_2,b)=0$; hence
$\V(M_k;T_n,b)=\tfrac12\kT\ln2$ constant---$\Delta\V=0$ (Main
Theorem~\ref{thm:main-I}(ii)(iii)).

\emph{The two blocking conditions.} (i)~$X_i\perp Y$: the acquired
correlation targets registers independent of the working media,
closing the laundering route (using the $M$--$Y$ correlation to
reset $X_i$ cheaply and re-extract). (ii)~$Y\notin\mathrm{Man}_\tau$
for every $\tau$ in the support: the cashing route is severed.
Read through (i), LB implements the classical statement that
memorizing environmental noise creates no value; read through (ii),
the correlation is genuine correlation with world variables of
every supported task, blocked purely by the absence of a
manipulation route.

\emph{$T$-relativity contrast.} Appending
$\tau_3$ (extraction from $\mathrm{Man}=\{Y\}$, the bundled
single-register reading of Remark~\ref{rem:t-relativity}) with
$T'=(1/3,1/3,1/3)$ gives
$\Gap(M_k;\tau_3,b)=\kT\,I(M_k;Y)=k\,\kT\ln2$, so
$\V(M_k;T',b)$ increases at every step: the same update sequence
capitalizes nothing under $T$ and every step under $T'$
(Remark~\ref{rem:t-relativity}).

\emph{Cartridge variant.} Replacing the copies of $Y$ by copies of
keys whose gates are absent from the support yields the same
conclusions ($I(M;B)$ grows; no cartridge task in $T$; $\Delta\V=0$):
what LB realizes by manipulation blockage, the variant realizes by
gate absence---the blockage has multiple physical realizations,
and this robustness is what Sec.~\ref{sec:regime} (B3) builds on.

\subsection{The ledger device suite (equality and its
non-achievements)}
\label{app:dev-ledger}

\emph{Device E (reversible copy; equality achievement).}
$W=\{X_1\}$ (one uniform bit); a single $\flatstar$ task
($E\equiv0$, $G=\emptyset$, $\mathrm{Man}=\{X_1\}\cup$anc,
$K=3$, static all-read); $T=\delta_{\tau}$; $D=$ perfect copy of
$X_1$; $M_0$ blank. Update: controlled-copy of $D$ into fresh
memory (growth form) or into a pre-erased cell (overwrite form);
both deterministic permutations on degenerate registers:
$\Sigmatot=0$, $\sigmaM=\ln2$, $\Delta\V=\kT\ln2$. Under the
$\Sigmatot$ account this is the B1 witness ($\Delta\V>0$ at
$\Sigmatot=0$); under $\sigmaM$ it achieves $\etacap=1$ with all
four equality conditions: no forgetting (blank start), no waste
($D$ is a perfect copy of $X_1$: $g_b=H(D\mid X_1)=0$), no
$Y$-contamination ($Y=\emptyset$), reversible implementation.

\emph{Device E$+$ex (excess).} Device E with the implementation
made finite-time: excess dissipation $s>0$, $\sigmaM=\ln2+s$,
$\etacap=\ln2/(\ln2+s)<1$---the (e) violation in isolation; as
$s\to0^{+}$ it also provides the divergent family
$\Delta\V/(\kT\Sigmatot)=\ln2/s$ for the $\Sigmatot$ account.

\emph{Device W (waste).} $D=(X_1,J)$ with $J$ an independent junk
bit; update copies all of $D$ (two bits, reversible):
$g_b=\ln2$ ignites alone; $\etacap=\ln2/(2\ln2)=1/2$.

\emph{Device F (forgetting plus conversion).} Shared world
$(X_1,X_2)$; two symmetric tasks $\tau_1,\tau_2$
($\mathrm{Man}=\{X_1\}$ resp.\ $\{X_2\}$), $T=(1/2,1/2)$;
$M_0=$ copy of $X_1$; $D=$ copy of $X_2$; update: overwrite (a
non-growth $M$-local update erasing the content of $M_0$ and
writing $D$). Books: on $\tau_1$, forgetting
$g_a(\tau_1)=\ln2$ \emph{and} $Y$-contamination
$g_c(\tau_1)=\ln2$ (for $\tau_1$, $X_2$ is a readable environment
variable), averaging to $g_a=g_c=\ln2/2$; on $\tau_2$ the same
acquisition is capital ($\Delta\Gap_{\tau_2}=+\kT\ln2$). Net:
$\Delta\V=0$, $\sigmaM=\ln2$, $\etacap=0$, $\V$ unchanged at
$\tfrac12\kT\ln2$. One task's capital is another task's
$Y$-contamination: the subjects are task relative, and no
cross-subject cancellation occurs in the $T$-average (all subjects
are nonnegative). Under the shift reading, device F is the book
reversal of Sec.~\ref{sec:accounting}: train book blind to the
forgetting, old book charged in full.

\emph{Device Y (preparedness).} $D=$ copy of $Y_1$, a readable
non-manipulable world bit independent of $X$; update copies $D$:
$g_c=\ln2$ ignites alone; $\etacap=0$---a copy of a readable
variable is what the blind side holds for free (the efficiency
version of the preparedness self-assessment of
Sec.~\ref{sec:discussion}).

\emph{Device G (gate; enablement).} Key $B$ of $n$ uniform bits;
tamper-proof pass/fail gate releasing $F_{\mathrm{cart}}$; horizon
$K_{\mathrm{hor}}$; budget $b$; $D=$ perfect copy of $B$;
$M_0$ blank; update a reversible copy of the $n$ bits,
$\sigmaM=n\ln2$. Informed value $=F_{\mathrm{cart}}$ by zero-work
conditional swap at any $b\ge0$; blind value $\le
F_{\mathrm{cart}}\min(1,K_{\mathrm{hor}}2^{-n}e^{\beta b})+b
+C_G\,\kT$ (Appendix~\ref{app:fixedI}). Hence
$\etacap\to\infty$ as $\beta F_{\mathrm{cart}}\to\infty$ at fixed
$(n,K_{\mathrm{hor}},b)$ with
$K_{\mathrm{hor}}2^{-n}e^{\beta b}<1$ (partial suppression,
Appendix~\ref{app:fixedI}, B1). Illustration at
$(n,\beta F_{\mathrm{cart}},\beta b,K_{\mathrm{hor}})=(10,100,0,3)$:
$\beta\Delta\V\ge 99.7-C_G$, i.e.\ $\etacap\ge14.2$ under the
script convention $C_G=1.0$ (Appendix~\ref{app:verification};
substitution arithmetic, not a theorem).

\emph{Device LB$+$diss (uncapitalized dissipation).} Device LB with
the copy re-implemented in overwrite form: each stage first erases
the target cell (quasistatic isothermal erasure: heat $\kT\ln2$,
$\Sigmatot=0$) and then writes $D_k$ reversibly:
$\sigmaM=\ln2$ per stage (plus $\Sigmatot>0$ if finite-time);
$\Delta I(M;D)=\ln2$ per stage; $\Delta\V=0$ (LB blockage);
$\etacap=0$. Against device E, the raw-acquisition ratio
$\dpJD/\sigmaM$ equals $1$ on both---the surrogate cannot separate
them; $\etacap$ separates them as $1$ versus $0$
(Sec.~\ref{sec:regime}, B3).

\emph{Device CX (recycler; (f) violation).} A single $\flatstar$
task: $W=\{X\}$ with $X=(X_1,X_2)$ two uniform bits as one
manipulable register, $Y=\emptyset$. $D=(D_1,D_2,D_J)$: copies of
$X_1,X_2$, and an independent fair coin. $M_0=$ copy of $D_J$ (a
junk-loaded state, one controlled-copy from blank). Non-growth
update: (i)~matched erasure of the $M$ cell against control $D_J$
(un-copy: the cell becomes deterministic), (ii)~write $D_1$ into
the freed cell, (iii)~absorb one blank ancilla and write $D_2$.
All stages deterministic permutations, $Q_{\mathrm{bath}}=0$:
$\sigmaM=[2\ln2-\ln2-0]+0=\ln2$, $\Sigmatot=0$, $\dpJD=2\ln2$,
recycling credit $I(M_0;D\mid M_1)=\ln2$;
$\Delta\V=2\kT\ln2$ and $\etacap=2$. Lemmas~\ref{lem:updateDP}
and~\ref{lem:decomp} hold on CX unimpaired; only the (f)-form of
the ledger identity breaks. One-bit variant with excess $s$:
$\sigmaM=s$, $\Delta\V=\kT\ln2$, $\etacap=\ln2/s\to\infty$.
Two-stage blank-start sequence (absorb junk: $\etacap=0$,
$\sigmaM=\ln2$; then recycle: $\etacap=2$, $\sigmaM=\ln2$):
cumulative $\eta_{\mathrm{cum}}=2\ln2/2\ln2=1$ exactly---the
books of the history are honest (Main
Theorem~\ref{thm:main-II}(v)).

\emph{Supply-side scope of $g_c$.} The subject
$g_c=I(M';Y_\tau\mid M)$ charges the whole $Y$-redundant component
of the acquisition, not only literal copies: on tasks with
$I(X;Y)>0$ and $D$ a copy of $X$, a perfectly hygienic learner
(no forgetting, no junk, $M'=$ copy of $D$) still incurs
$g_c=I(X;Y)>0$, and since the update kernel cannot read $Y$, no
admissible learner achieves $\etacap=1$ on such
$(X,Y,D)$---equality achievability is a joint property of the
record format and the task correlation structure, not of learner
hygiene alone. The devices above sit at the $X\perp Y$ corner
where this does not activate.

\subsection{The retention device suite}
\label{app:dev-retention}

The alignment witnesses SH, BG, AC$'$, RJ are defined with their
proofs in Appendix~\ref{app:sbar}. The remaining retention devices:

\emph{Device KR (key redraw; Type II).} As defined in
Sec.~\ref{sec:gate-threshold}: wired key $\Theta$ (frozen,
energy-decoupled), pass iff $B=s(\Theta)$, uniform $d$-fold
degenerate pass map, prize $F_{\mathrm{cart}}$, horizon
$K_{\mathrm{hor}}$, budget $b$; $M=$ copy of $s(\Theta)$; shift
redraws $k_{\mathrm{res}}$ effective bits (persisting bits match;
remainder and degeneracy coordinates independent---the
support-extension coupling obligation of
Definition~\ref{def:shift}, discharged inside the device
definition).

\emph{Device DJ (non-shared worlds; relativity of waste).} Two
tasks with disjoint worlds $W_{\tau_A}=\{X_A\}$,
$W_{\tau_B}=\{X_B\}$ (independent uniform bits); $D=$ copy of
$X_A$ (outside the world of $\tau_B$); update: copy $D$ into blank
$M$. Books: $g_b(\tau_A)=0$ but $g_b(\tau_B)=\ln2$---the same
acquisition is capital on $\tau_A$'s book and entirely waste on
$\tau_B$'s, and $\Lgen^{\mathrm{upd}}=\kT\ln2$ is carried in full
by the $g_b$ difference: waste, too, is relative on non-shared
worlds (Sec.~\ref{sec:accounting}).

\emph{Rectangle devices (Corollary~\ref{cor:orthogonality}).}
Shared world $(X_A,X_B)$, two symmetric tasks; update $=$ copy of
$D_A$ (record of $X_A$) into blank memory; implementation
reversible ($\sigmaM=\ln2$, $\etacap=1$) or with excess
$s=\ln2\,(1-\eta_0)/\eta_0$ ($\etacap=\eta_0$); evaluation
$T_{\mathrm{train}}=\delta_{\tau_A}$ with
$T_{\mathrm{shift}}=\delta_{\tau_A}$ ($\rhogen=1$),
$\delta_{\tau_B}$ ($\rhogen=0$), or $(\delta,1-\delta)$
($\rhogen=\delta$). In the $(1,0)$ corner the four equality
conditions of Main Theorem~\ref{thm:main-II}(iv) hold $T$-a.s.\ on
$T_{\mathrm{train}}$.

\emph{Restoration checks.} For each alignment witness, restoring
the named condition restores argmax agreement:
SH: $T_{\mathrm{shift}}\to T_{\mathrm{train}}$; RJ: replace all
components of $D$ by $X$-content; AC$'$: extend the manipulable
set to $Z=(X,Y)$; BG: $b\to\infty$ (hypothesis-free sandwich) or
remove the gate. Together with the positive domain (five
(F5$'$)-stable candidates realizing the equalities of
Proposition~\ref{prop:corresp} with argmax agreement), these close
the correspondence table numerically.

\subsection{Numerical verification}
\label{app:verification}

Three pure-Python scripts (standard library only), re-executed
directly in the main verification pass, cover the devices above;
all exit with status 0.

\emph{Coverage discipline (stated exactly).} None of the scripts
performs a numerical search over multi-stage adaptive protocols:
such suprema are not finitely enumerable, and the converse load is
carried entirely by the analysis (the three-step argument of
Appendix~\ref{app:dev-lb}; the conditional chain bound). What the
numerics verify is: (a)~the \emph{arithmetic of the converse
values}---the analytic upper bounds re-evaluated independently
from the joint distributions; (b)~\emph{achievability
transformations}---the copy/permutation state maps simulated
explicitly, with the extraction-stage work evaluated by the
identity $\kT[\ln\lvert X\rvert-H(X\mid\text{read})]$ (a
semi-analytic step); and (c)~\emph{saturation within the read
family}---for the LB family, all $2^{k+1}$ read subsets $S$
evaluated, confirming that the informed value saturates at
$\kT\ln2$ on subsets containing $M_{\mathrm{core}}$, is $0$
without it, and that the blind value is $0$ independent of $k$.
Gap evaluations on $\flatstar$ devices use the right-hand side of
Eq.~(\ref{eq:flat-identity}); this represents $\Gap$ only on
(F5$'$)-conforming states, and the scripts \emph{machine-check}
(F5$'$) on every reached state (point-mass/uniform-on-subset test);
noisy (BSC) states appear only as (F5$'$)-violation witnesses.
Gate-task rows are value-substitution arithmetic in the imported
analytic bounds, with the bookkeeping constant $C_G$ fixed inside
each script by convention ($C_G$ evaluated at $1.0$ in the
ledger-suite
gate check and conservatively at $5.0$ throughout the
retention suite) and margins checked as inequalities. These script
values are \emph{conventions}, not proven upper bounds on $C_G$;
every gate-task numerical display in this paper is substitution
arithmetic under them, and no theorem rests on the numbers
(Sec.~\ref{sec:regime}, Theorem~\ref{thm:gatethreshold}).

\emph{Script 1} (device LB; $N_Y=6$, 256 atoms, deviations
$\le1.1\times10^{-14}$): monotone growth of $I(M_k;D)$ to
$n\ln2$; $I(M_k;W)=(k+1)\ln2$; pinned gaps
$(\kT\ln2,0)$ for all $k$; constancy of $\V$ under $T$ versus
strict growth under $T'$; and the single-task anchor
(Main Theorem~\ref{thm:main-I}(i)(e)) by read-set scan.

\emph{Script 2} (ledger suite; 14 checks): the universal and
four-term ledger identities on all ten devices, with the (f)-form
on the nine (f)-conforming devices and $\Sigmatot\ge0$ throughout;
equality achievement on device E (both implementations) with all
subjects zero; the non-achievement pattern
(W $\to g_b=\ln2$ only; F $\to g_a=g_c=\ln2/2$;
Y $\to g_c=\ln2$ only); LB$+$diss versus E under the
raw-acquisition surrogate (both $1.0$) and under $\etacap$
($0$ vs $1$); the B1 contrast; the device-G arithmetic
($\etacap\ge14.2$ under the script convention $C_G=1.0$; the check
condition is the qualitative $\etacap>1$ only); the five-term decomposition
Eq.~(\ref{eq:full-ledger}) on all devices, degenerating to four
terms under (f); Lemma~\ref{lem:updateDP} on all
device--task pairs; the per-task equality of
Lemma~\ref{lem:decomp} on all pairs \emph{including} the
(f)-violating CX; device CX as permanent negative control
($\etacap=2$; failure of the (f)-form; validity of the universal
form) and the two-stage cumulative chain
$\eta_{\mathrm{cum}}=1$; and the (F5$'$)-necessity witness: the
deep dyadic atom $p=(1/2,1/4,1/8,1/8)$ yields blind uniformization
trajectory works $(\ln2,0,-\ln2,-\ln2)\,\kT$---the negative
branches violating the $b=0$ pathwise cap while the expectation
matches---against a flat-conditional control with all branches
nonnegative; and the partial-suppression boundary of the device-G
unboundedness: at $(n,K_{\mathrm{hor}},b)=(1,2,0)$ the suppression
parameter reaches one, a non-adaptive blind protocol submits both
keys within the horizon and collects the prize's leading term, and
the slope $(1-c)$ vanishes---confining the fixed-content
divergence to $c<1$ (the B1 premise).

\emph{Script 3} (retention suite; 84 checks in eight sections):
(F5$'$) machine checks on all reached states throughout; the
two-layer correspondence
equalities and argmax agreement on five candidates, together with
the one-time-pad boundary witness of
Proposition~\ref{prop:corresp-boundary} as a negative
control---(F5$'$)-stability of the world-correlated initial memory
and both candidates, survival of the conditional exchange rate on
both, disjoint raw-stock and value argmax sets, and the blank-start
contrast restoring agreement; the SH gap
matrix, strict reversal, $\rhogen=\epsilon/(1-\epsilon)$, the (G1)
equality, and restoration; AC$'$ and RJ reversals with
manipulable-set and record restorations, the legacy
environment-record variant, and the BSC (F5$'$)-violation witness;
the BG fit order, parameter conditions, $b=0$ reversal, hypothesis-%
free recovery sandwich, reference conditional equality, and gate
removal; the rectangle corners, the $(1,0)$-corner subject
audit, and the partial-shift continuum; the KR certificates for the
two separated hypotheses of Theorem~\ref{thm:gatethreshold}---full
suppression and positivity---plus the witness quadruple
$(100,50,10,3)$ on which (i) holds while the training lower bound
is negative (hypothesis (i) alone does not certify the domain; the
positivity certificate requires (ii)), the closed-form left-shoulder
bound of part (4), threshold shoulders, window monotonicity, the
$d$/$b$ dials (including the numerical separation of
Proposition~\ref{prop:ddial} under the stated convention), the
hypothesis-domain guard (at fixed
$(k_{\mathrm{eff}},K_{\mathrm{hor}},b)$, hypothesis (i) forces
$\beta F_{\mathrm{cart}}\le1/c_{\mathrm{full}}$), and the
shoulder-separation arithmetic
($\log_2(\beta F_{\mathrm{cart}})+\beta b/\ln2+t$, the
$\log_2K_{\mathrm{hor}}$ terms cancelling); the
subject postings, the
subject-difference identity, junk as sub-event, the device-F book
reversal, the device-DJ $\tau$-dependence, and the shared-world
constancy of Lemma~\ref{lem:sharedworld}; and the per-bit
linearity of Proposition~\ref{prop:perbit}.

\section{Equality proof and proof infrastructure}
\label{app:equality}

\subsection{The blind branch does not depend on $M$: full
argument}
\label{app:proof-blind}

A blind protocol has no read port on $M$ (deletion), $M$ is
energetically decoupled (Definition~\ref{def:memory}(i)), and by
clauses (ii)--(iv) it enters no dynamics: every stage map of a
blind protocol factorizes as $\mathrm{id}_M\otimes\Lambda_W$, and
the non-$M$ variables evolve autonomously. Hence
$\mathbb{E}[W_{\mathrm{ext}}]$ of every
$P\in\mathrm{Prot}_b(A_\tau-M,p)$ is a functional of the non-$M$
marginals alone, and the blind branch does not depend on the
coupling structure $\pi_M$. The one remaining conceivable
channel---importing an auxiliary whose initial state is correlated
with $M$---is closed not by pricing (the no-unpriced-resources
clause only prices imports, it does not forbid them) but by type:
the coupling at $t_0$ is task data, not an object of protocol
choice, and mid-run auxiliaries are adjoined as fresh registers in
product state ($\mathrm{id}\otimes\rho$); no preparation map acting
on $M$ exists in the class. This is the justification for writing
$\Wblind(T,b)$ without an $M$ argument.

\subsection{The conditional chain bound, stated}
\label{app:proof-chain}

The converse below rests on a conditional form of the chain bound
of the companion framework, with the total-correlation term
retained. We state it once, so that the weight-bearing path
depends on a statement rather than on the interior of an imported
proof.

\begin{lemma}[Conditional chain bound, retained form]
\label{lem:condchainprime}
Let $\tau$ be a task, $M$ a memory conforming to
Definition~\ref{def:memory} (hence frozen during runs), and
$P\in\mathrm{Prot}_b(A_\tau+M,p)$ under convention ($\alpha$)
(Sec.~\ref{sec:setting-imports}). Then for every memory value $m$,
\begin{equation}
\begin{aligned}
\beta\,\mathbb{E}[W_{\mathrm{ext}}\mid M=m]
&\le-\beta\sum_i\Delta F^{(m)}_{\mathrm{loc},i}\\
&\quad+TC(t_0\mid m)-TC(t_K\mid m),
\end{aligned}
\label{eq:condchain}
\end{equation}
where $F_{\mathrm{loc},i}$ are the local free energies, $TC$ is the
total correlation of the register set, and both are evaluated on
the $M=m$ conditional law. This is the chain bound of the companion
framework~\cite{Sudo_MPU} applied to the conditional law, with the
telescoping term $\sum_k(\mathrm{dec}-\mathrm{inc})
=TC(t_0)-TC(t_K)$ of its proof retained rather than discarded.
\end{lemma}

\begin{lemma}[Auxiliary netting]
\label{lem:auxnetting}
On a flat landscape ($E\equiv0$), the entropy contribution of any
imported auxiliary obeys
$\Delta S_{\mathrm{aux}}\le\ln\lvert A\rvert-S(\rho_{\mathrm{aux}})
\le\ln\lvert A\rvert-H_{\min}(\rho_{\mathrm{aux}})
=D_{\max}(\rho_{\mathrm{aux}}\Vert\pi_{\mathrm{flat}})$: the
auxiliary's $\Delta S$ credit cannot exceed its one-shot
$D_{\max}$ charge, so its net contribution under the terminal-%
balance convention ($\beta$) is $\le0$. The bound holds for
stochastic and per-realization (branch-dependent) preparations
$\rho_m$; and since fresh auxiliaries are adjoined in product
state, $\rho_m$ is constant within each branch, so the auxiliary
contributes zero to $TC(t_0\mid m)$.
\end{lemma}

\subsection{Proof of the extraction identity
(Lemma~\ref{lem:flat-identity})}
\label{app:proof-flat}

\emph{Converse (branch upper bounds; valid for every state, no
use of (F5$'$)).} Write $X:=X_\tau$, $Y:=Y_\tau$.

\emph{Step 1.} $M$ is frozen (derived from
Definition~\ref{def:memory}(ii)$+$(iv)); under convention
($\alpha$), apply Lemma~\ref{lem:condchainprime} to the $M=m$
conditional law.

\emph{Step 2 (degenerate local free energies).} By (F1),
$F_{\mathrm{loc},i}=-\kT\,S_i$, so
$-\beta\sum_i\Delta F^{(m)}_{\mathrm{loc},i}
=\sum_i\Delta S^{(m)}_i$. Split the variables into three groups:
for $X$ (manipulable), $\Delta S^{(m)}_X\le\ln\lvert X\rvert
-S(X\mid m)$ (uniformization is the ceiling); for $Y$
(non-manipulable), designed mechanisms act only on
$\mathrm{Man}\cup\text{ancillas}$ (mechanism locality) and (F2)
removes fixed gates, so the $Y$-\emph{joint} distribution is
invariant through the run and $\Delta S^{(m)}_i=0$ for all
$i\in Y$; for auxiliaries, Lemma~\ref{lem:auxnetting} gives a net
$\le0$.

\emph{Step 3 (cancellation of the correlation terms).} By the
chain rule for total correlation,
$TC(t_0\mid m)=TC_Y(t_0\mid m)+I(X;Y\mid m)$ (one register $X$
against the group $Y$). By monotonicity of $TC$ under taking
subsystems, $TC(t_K\mid m)\ge TC_Y(t_K\mid m)$, and by Step~2 the
$Y$-joint is invariant, so $TC_Y(t_K\mid m)=TC_Y(t_0\mid m)$.
Hence $TC(t_0\mid m)-TC(t_K\mid m)\le I(X;Y\mid m)$.

\emph{Step 4 (assembly).}
$\beta\,\mathbb{E}[W_{\mathrm{ext}}\mid m]\le\ln\lvert X\rvert
-S(X\mid m)+I(X;Y\mid m)=\ln\lvert X\rvert-S(X\mid Y,m)$; taking
$\mathbb{E}_m$, the informed branch is
$\le\kT[\ln\lvert X\rvert-H(X\mid Y,M)]$. The same argument
without the $M$-conditioning (the blind branch is a functional of
the non-$M$ marginals, Appendix~\ref{app:proof-blind}) gives
blind $\le\kT[\ln\lvert X\rvert-H(X\mid Y)]$.

\emph{Achievability under (F5$'$): all budgets, pathwise
nonnegative, exact.} Informed: a conditional permutation with
parent set $(m,y)$ aligns the support of the conditional
distribution of $X$ (one stage, deterministic, zero work on the
degenerate landscape), and an isothermal expansion from the uniform
support $S_{m,y}$ to the full alphabet extracts, on \emph{every}
trajectory, $w_{\mathrm{ext}}(\omega)=\kT\ln(\lvert X\rvert/\lvert
S_{m,y}\rvert)\ge0$ (one stage; $K_\tau\ge3$ suffices). By
(F5$'$) each branch distribution is a point mass or uniform, so the
branch average attains $\kT[\ln\lvert X\rvert-H(X\mid M,Y)]$
exactly. Blind: the same two stages with parent $y$ only attain
$\kT[\ln\lvert X\rvert-H(X\mid Y)]$. The blind value attains the
blind converse, so the \emph{gap} closes as an equality, each
branch pinned separately.

\emph{Budget independence.} The converse nowhere involves $b$. But
that alone does not give $b$-independence of the gap: $\Gap$ is a
difference, and a budget that shrinks only the blind supremum
\emph{increases} the gap. What carries the claim is the pathwise
nonnegativity of the achieving protocols: with the running account
$E_k(\omega)$ read as cumulative net drawdown (extraction moves it
negative), every positive-probability trajectory of both protocols
satisfies $E_k(\omega)\le0\le b$ at every stage, so both branches
attain their ceilings at every budget including $b=0$. This is
precisely where (F5$'$) enters, and where deep dyadic atoms break
the equality (Definition~\ref{def:flatstar}). $\blacksquare$

The two halves have unequal scope, and the asymmetry is
load-bearing elsewhere: the informed converse holds for every
state and underlies Main Theorem~\ref{thm:main-I}(iv) (the
post-update state of a D-free update need not satisfy (F5$'$)) and
the flat-task bound below, whose blind side is pinned by the
$y$-clause; achievability of the gap equality is confined to
(F5$'$).

\subsection{Conditional inheritance of joint independence}
\label{app:proof-inherit}

\begin{lemma}[Joint independence is inherited under
conditioning]
\label{lem:jointindep}
If $A\perp(W,M,D)$ jointly, then for every partition of
$(W,M,D)$ into subtuples $(B_1,B_2)$: $A\perp B_1\mid B_2$.
Consequently $p(a\mid x,y,m,d)=p(a)$, and the kernel marginalized
over the ancilla, $\tilde K(m'\mid m,d):=\sum_a p(a)K(m'\mid
m,d,a)$, does not depend on $(x,y)$---yielding simultaneously the
Markov chains $X\!-\!D\!-\!M'$ given $(M,Y)$ and
$Y\!-\!(M,D)\!-\!M'$.
\end{lemma}

\begin{proof}
$p(a,b_1,b_2)=p(a)\,p(b_1,b_2)$ implies
$p(a,b_1\mid b_2)=p(a)\,p(b_1\mid b_2)$ directly. Stochastic
kernels are covered because bath randomness is already folded into
the kernel's dependence on $(m,d,a)$
(Definition~\ref{def:update}).
\end{proof}

\subsection{Proof of update data processing
(Lemma~\ref{lem:updateDP})}
\label{app:proof-dp}

By Lemma~\ref{lem:flat-identity} applied to both states
((F5$'$)-stability),
$\Delta\Gap_\tau=\kT[\,I(M';X\mid Y)-I(M;X\mid Y)\,]$.

\emph{Supply cap.} Chain rule:
$I(M';X\mid Y)\le I((M',M);X\mid Y)=I(M;X\mid Y)+I(M';X\mid M,Y)$,
so $\Delta\Gap_\tau/\kT\le I(M';X\mid M,Y)$; the Markov chain
$X\!-\!D\!-\!M'$ given $(M,Y)$ (Lemma~\ref{lem:jointindep}) and
data processing give $I(M';X\mid M,Y)\le I(D;X\mid M,Y)$.

\emph{Acquisition cap.} The same Markov property gives
$I(M';X\mid M,Y)\le I(M';D\mid M,Y)$. Since also
$Y\!-\!(M,D)\!-\!M'$ is Markov, $I(M';Y\mid M,D)=0$, and the two
expansions of $I(M';(D,Y)\mid M)$ give
\begin{equation*}
\begin{aligned}
I(M';D\mid M,Y)
&=I(M';D\mid M)-I(M';Y\mid M)\\
&\le\dpJD .\qquad\blacksquare
\end{aligned}
\end{equation*}

\subsection{Proof of the ledger identity
(Lemma~\ref{lem:ledger-identity})}
\label{app:proof-ledger}

Let the terminal system be $(M',D)$, with $M'$ the totality of
memory-side registers---$M$ together with all absorbed ancillas,
excluding the read-only $D$, the world, and the bath---so that the
disjoint pair $(M',D)$ exhausts the terminal system (the totality
convention of Definition~\ref{def:implementation}(iii)).
From the definition of the implementation entropy production,
$\Sigmatot=[S(M',D)-S(M,D,A)]+\Delta S_{\mathrm{bath}}$. Initially
$A$ is jointly independent of $(M,D)$, so
$S(M,D,A)=S(M)+S(D)-I(M;D)+S(A)$; terminally
$S(M',D)=S(M')+S(D')-I(M';D)$ with $S(D')=S(D)$ ($D$ read-only).
Substituting, $\Sigmatot=\sigmaM-\Delta I(M;D)$, which is
Eq.~(\ref{eq:ledger-universal}); the four-term form is the chain
rule $I((M,M');D)=I(M;D)+I(M';D\mid M)=I(M';D)+I(M;D\mid M')$.
$\blacksquare$

\subsection{Proof of the per-task decomposition
(Lemma~\ref{lem:decomp})}
\label{app:proof-decomp}

Write $X:=X_\tau$, $Y:=Y_\tau$. Three chain-rule identities
superpose:
(1)~the two expansions of $I((M',M);X\mid Y)$:
$I(M';X\mid Y)-I(M;X\mid Y)=I(M';X\mid M,Y)-I(M;X\mid Y,M')
=I(M';X\mid M,Y)-g_a$;
(2)~the two expansions of $I(M';(X,D)\mid M,Y)$ with the Markov
property $X\!-\!D\!-\!M'$ given $(M,Y)$ (hence
$I(M';X\mid D,M,Y)=0$):
$I(M';X\mid M,Y)=I(M';D\mid M,Y)-I(M';D\mid X,M,Y)
=I(M';D\mid M,Y)-g_b$;
(3)~the two expansions of $I(M';(D,Y)\mid M)$ with the Markov
property $Y\!-\!(M,D)\!-\!M'$ (hence $I(M';Y\mid M,D)=0$):
$I(M';D\mid M,Y)=I(M';D\mid M)-I(M';Y\mid M)=\dpJD-g_c$.
Substituting into
$\Delta\Gap_\tau=\kT[\,I(M';X\mid Y)-I(M;X\mid Y)\,]$
(Lemma~\ref{lem:flat-identity} at both states) assembles
Eq.~(\ref{eq:decomp}). $\blacksquare$

(Lemma~\ref{lem:updateDP} is the truncation of this equality:
dropping all subjects gives the acquisition cap; the two Markov
properties are what turn the inequality chain into an equality
with nonnegative residues.)

\subsection{Proof of Main Theorem~\ref{thm:main-II}, parts
(iii)--(v)}
\label{app:proof-mainII}

\emph{(iii).} By the $T$-averaged acquisition cap
(Lemma~\ref{lem:updateDP}), $\Delta\V\le\kT\,\dpJD$; by the
(f)-form of Lemma~\ref{lem:ledger-identity},
$\dpJD=\sigmaM-\Sigmatot\le\sigmaM$. Compose. $\blacksquare$

\emph{(iv).} $T$-average Lemma~\ref{lem:decomp}: $\dpJD$ is
$\tau$-independent by draw exogeneity
(Definition~\ref{def:exogeneity} and
Lemma~\ref{lem:exo-invariance}: $\tau\perp(M,D,\text{ancillas})$
makes the joint law of $(M,M',D)$ common to all $\tau$), so
$\mathbb{E}_\tau[\dpJD]=\dpJD$ and, substituting the (f)-form,
\begin{equation*}
\begin{aligned}
\frac{\Delta\V}{\kT}
&=\dpJD-\mathbb{E}_\tau[g_a+g_b+g_c]\\
&=\sigmaM-\Sigmatot-\mathbb{E}_\tau[g_a+g_b+g_c] .
\end{aligned}
\end{equation*}
Thus $\etacap=1\iff\mathbb{E}_\tau[g_a+g_b+g_c]+\Sigmatot=0$; all
terms are nonnegative, so the sum vanishes iff each term does,
i.e., iff (a)(b)(c) hold at every $\tau$ with $T(\tau)>0$ and
(e) holds. Without (f), substituting the four-term form instead
gives $\etacap=1\iff\mathbb{E}_\tau[g_a+g_b+g_c]+\Sigmatot
=I(M;D\mid M')$. The provenance of the four nonnegativities is
asymmetric, and so is the failure mode: the three subjects are
conditional mutual informations, nonnegative unconditionally,
while $\Sigmatot\ge0$ alone rests on the physical clauses of
Definition~\ref{def:implementation}(iv) (fresh equilibrated bath,
local detailed balance, no feedback); in implementation classes
with pre-correlated baths only the forward implication breaks.
$\blacksquare$

\emph{(v).} (1)~The universal identity telescopes:
$\Sigma_{\mathrm{search}}=\sum_k\sigmaM^{(k)}
=[I(M_n;D)-I(M_0;D)]+\sum_k\Sigmatot^{(k)}\ge I(M_n;D)$, using
$I(M_0;D)=0$ (blank start) and $\Sigmatot^{(k)}\ge0$.
(2)~$M_n$ is a measurable function of $(M_0,D,\text{all
ancillas})$ with $M_0$ deterministic and the ancillas jointly
independent of $(W,M,D)$ (the per-update freshness composes: each
ancilla is jointly independent of the entire history at its
introduction, so the pooled randomness is independent of $W$ given
$D$); hence $M_n\perp W_\tau\mid D$ for every $\tau$, and
$I(M_n;X_\tau\mid Y_\tau)\le I(M_n;W_\tau)\le I(M_n;(D,W_\tau))
=I(M_n;D)+I(M_n;W_\tau\mid D)=I(M_n;D)$.
(3)~$\V(M_0)=0$ and
$\V(M_n)\le\kT\,\mathbb{E}_\tau[I(M_n;X_\tau\mid Y_\tau)]$: since
$M_0$ is blank, the $\flatstar$ property of the support reduces at
the task level to the $y$-clause, so every supported task is a
flat task and the state-uniform bound of Lemma~\ref{lem:flattask}
applies to $M_n$ \emph{whether or not} (F5$'$) survived the
sequence (only the converse half is used). Composing (1)--(3)
gives Eq.~(\ref{eq:cumulative}). $\blacksquare$

\subsection{Proof of the flat-task bound and the task optimum}
\label{app:proof-flattask}

\emph{Lemma~\ref{lem:flattask}.} The informed converse of
Appendix~\ref{app:proof-flat} uses no (F5$'$) and gives, for every
$M$, informed $\le\kT[\ln\lvert X\rvert-H(X\mid Y,M)]$. Blind
achievability is carried by the $y$-clause: the blind conditional
distributions $p_\tau(x\mid y)$ are task properties independent of
$\pi_M$, all point masses or uniform, so the two-stage protocol of
Appendix~\ref{app:proof-flat} attains
$\kT[\ln\lvert X\rvert-H(X\mid Y)]$ at every budget, pathwise
nonnegative. Subtract. Essentiality of the $y$-clause: without it
the blind side cannot attain its ceiling at $b=0$, while an
informed state with point-mass branches can, so
$\Gap>\kT\,H(X\mid Y)$ becomes possible. $\blacksquare$

\emph{Proposition~\ref{prop:taskopt}.} ($\le$)
Lemma~\ref{lem:flattask} applies to every $M$ in the supremum, and
$I(M;X\mid Y)\le H(X\mid Y)$. ($\ge$) Take $M^{*}$ the complete
copy of all $X_\tau$ (in a shared world, of the world register
list). Its conditional distributions are point masses, so (F5$'$)
holds automatically and Lemma~\ref{lem:flat-identity} applies as
an equality, giving
$\V(M^{*})=\kT\,\mathbb{E}_\tau[H(X_\tau\mid Y_\tau)]$.
Admissibility of $M^{*}$: under (i), the off-task components are
specified independent; under (ii), the common $p_\tau(w)$ makes
the $M^{*}$-marginal $\tau$-independent; in both cases draw
exogeneity (Definition~\ref{def:memory}(b)) holds, with $D$ taken
as a trivial register. (In a shared world with $\tau$-dependent
$p_\tau(w)$, the copy construction violates exogeneity and the
closed form is not claimed.) $\blacksquare$

\subsection{Proof of the two-layer correspondence
(Proposition~\ref{prop:corresp})}
\label{app:proof-corresp}

\emph{(A).} By (F5$'$)-stability, the extraction identity
(Lemma~\ref{lem:flat-identity}, via
Corollary~\ref{cor:v-identity}) applies to each candidate state:
$\V(M'_f;\delta_\tau,b)=\kT\,I(M'_f;X_\tau\mid Y_\tau)$,
independently of $b$. By ($\beta$), $D$ and $X_\tau$ have
identical content, so
$I(M'_f;X_\tau\mid Y_\tau)=I(M'_f;D\mid Y_\tau)$. No further step
is taken: in particular, no conditioning is removed. $\blacksquare$

\emph{(B).} $M'_f=f(M,D,A)$ with $A$ jointly independent of
$(W,M,D)$ (Definition~\ref{def:update}). From $(M,D)\perp Y_\tau$
and the joint independence of $A$, the tuple $(M,D,A)$ is
independent of $Y_\tau$; $(M'_f,D)$ is a function of that tuple,
hence $(M'_f,D)\perp Y_\tau$, and
$I(M'_f;D\mid Y_\tau)=I(M'_f;D)$. Substitute into (A). Under a
blank start, $(M,D)\perp Y_\tau$ reduces to $D\perp Y_\tau$ and
$\JD(M)=0$, so $\Delta\V_f=\kT\,\JD(M'_f)=\kT\,\dpJD(f)$.
(The failed route recorded for contrast: from $D\perp Y_\tau$
alone---e.g.\ from $X_\tau\perp Y_\tau$ and ($\beta$)---the
identity $I(M'_f;D\mid Y_\tau)=I(M'_f;D)$ does \emph{not} follow
when $M$ carries $Y$-correlation, since conditioning can create
mutual information; Proposition~\ref{prop:corresp-boundary} makes
this failure explicit, with numerical regression in the retention
suite.) $\blacksquare$

\subsection{Proof of the subject-difference identity
(Corollary~\ref{cor:subject-diff}) and
Lemma~\ref{lem:sharedworld}}
\label{app:proof-subjectdiff}

\emph{Corollary~\ref{cor:subject-diff}.} Under hypothesis (i),
Lemma~\ref{lem:decomp} holds at every $\tau\in\mathcal{T}$:
$\Delta\Gap_\tau=\kT[\dpJD-g_a(\tau)-g_b(\tau)-g_c(\tau)]$.
Average against $T_{\mathrm{train}}$ and $T_{\mathrm{shift}}$ and
subtract: by hypothesis (iii) the joint law of $(M,M',D)$ is
common to all $\tau\in\mathcal{T}$, so $\dpJD$ is $\tau$-%
independent and cancels in the difference, leaving
Eq.~(\ref{eq:subject-diff}). If (iii) fails---a drawing mechanism
that reads the training closure---$\dpJD$ becomes $\tau$-dependent
and the identity acquires the residual term
$\kT(\dpJD^{\mathrm{train}}-\dpJD^{\mathrm{shift}})$; the
exclusion of realization-correlated shifts
(Definition~\ref{def:shift}) is what removes it. $\blacksquare$

\emph{Lemma~\ref{lem:sharedworld}.} On a shared world,
$(X_\tau,Y_\tau)$ is a partition of the same register set $W$ for
each $\tau$, so conditioning on $(X_\tau,Y_\tau)$ jointly is
conditioning on $W$:
$g_b(\tau)=I(M';D\mid X_\tau,Y_\tau,M)=I(M';D\mid W,M)$,
$\tau$-independent. $\blacksquare$

\section{S/B/A/R witness details}
\label{app:sbar}

\subsection{The uninformative-port lemma}
\label{app:proof-noport}

\begin{lemma}[Uninformative port]
\label{lem:noport}
If $M$ is jointly independent of \emph{all} variables of a task
$\tau$ (world registers, gate keys, cartridge states), then
$\Gap(M;\tau,b)=0$ for every $b$.
\end{lemma}

\begin{proof}
Under $M\perp(\text{all task variables})$, conditioning on $M=m$
leaves the task-side joint law unchanged; each $m$-branch protocol
of the informed side corresponds one-to-one, by constant
substitution, to an element of the blind class with the identical
objective functional, so
$\mathbb{E}_m[\sup]=\sup$ at the same value---no convexity gain
from the mixture arises. The lemma does not use $\flatstar$ and
applies to gate tasks.
\end{proof}

\subsection{Witness S: device SH and the proof of Main
Theorem~\ref{thm:main-III}(ii)}
\label{app:witness-s}

\emph{Device SH.} Shared world $W=(X_A,X_B)$, one uniform bit
each, independent, $p$ common to both tasks.
$\mathcal{T}=\{\tau_A,\tau_B\}$:
$\tau_A=(W,\text{uniform},E\equiv0,G=\emptyset,
\mathrm{Man}=\{X_A\}\cup\text{anc},K=3,\text{static all-read})$,
$\flatstar$ with $Y_{\tau_A}=\{X_B\}$; $\tau_B$ symmetric.
$D=(D_A,D_B)$ perfect copies; $M_0$ blank; capacity one bit;
$F=\{f_A,f_B\}$ with $f_i$ the controlled-copy of $D_i$. All
reached states satisfy (F5$'$) (machine-checked;
Appendix~\ref{app:verification}).

\emph{Proof.} By Corollary~\ref{cor:v-identity} ((F5$'$) holds)
and $T$-affinity, all values are arithmetic:
$\Delta\V(f_A;T_{\mathrm{train}},b)=(1-\epsilon)\kT\ln2$,
$\Delta\V(f_B;T_{\mathrm{train}},b)=\epsilon\,\kT\ln2$, and the
rankings reverse under
$T_{\mathrm{shift}}=(\epsilon,1-\epsilon)$; the argmax sets are
singletons and disjoint. $\rhogen(M_A)
=\epsilon\,\kT\ln2/[(1-\epsilon)\kT\ln2]=\epsilon/(1-\epsilon)$.
For the (G1) equality:
$\lVert T_{\mathrm{train}}-T_{\mathrm{shift}}\rVert_{\mathrm{TV}}
=1-2\epsilon$ and $\sup_\tau\Gap(M_A;\tau,b)=\kT\ln2$, whose
product equals $\Lgen(M_A)=(1-2\epsilon)\kT\ln2$. Both candidates
acquire $\dpJD=\ln2$ (tie), and the world correlation is likewise
tied: the representative training-side measures do not separate
them. $\blacksquare$

\subsection{Witness B: device BG and the proof of Main
Theorem~\ref{thm:main-III}(iii)}
\label{app:witness-b}

\emph{Device BG (Type I).} Two independent parts. Flat part: $X$,
$n$ uniform bits as a single manipulable register ($K=3$, static
all-read). Gate part: an imported gate of the cartridge type
(Appendix~\ref{app:fixedI}), Type~I realization, with key $B$ of
$m<n$ uniform bits, prize $F_{\mathrm{cart}}$, horizon
$K_{\mathrm{hor}}$. $\mathcal{T}=\{\tau_{\mathrm{flat}},
\tau_{\mathrm{gate}}\}$ ($\tau_{\mathrm{gate}}$ is not
$\flatstar$), $T=(\tfrac12,\tfrac12)$; $D=(D_X,D_B)$ perfect
copies; capacity $n$ bits; $F=\{f_X,f_B\}$ the two full copies.
Parameter conditions (with $C_G$ the uniform bookkeeping constant
of Sec.~\ref{sec:regime}; whether a numerical quadruple satisfies
them depends on the unproven value of $C_G$---Main
Theorem~\ref{thm:main-III}(iii)): $n\ln2>2m\ln2+2\,C_G$ and
$\beta F_{\mathrm{cart}}(1-K_{\mathrm{hor}}2^{-m})-2\,C_G
>n\ln2$.

\emph{Proof.} (1)~is copy arithmetic. (2)~$\V(M_X;T,0)$: flat
part by Lemma~\ref{lem:flat-identity} (point-mass state, (F5$'$)
holds), $\Gap=n\,\kT\ln2$; gate part by Lemma~\ref{lem:noport}
($M_X\perp(B,\text{cartridge})$), $\Gap=0$; total
$\tfrac12 n\,\kT\ln2$. $\V(M_B;T,0)$: flat part zero
(Lemma~\ref{lem:noport}); gate part: informed
$\ge F_{\mathrm{cart}}-C_G\,\kT$ at $b=0$ (statement B1,
Appendix~\ref{app:fixedI}; the Type~I conditional-extraction
channel can only strengthen this lower bound, so it is
conservative), blind $\le
F_{\mathrm{cart}}K_{\mathrm{hor}}2^{-m}+0+C_G\,\kT$ (the imported
ledger corollary~\cite{Sudo_MPU}); subtract. The stated parameter
condition makes $\V(M_B;T,0)>\V(M_X;T,0)$ strict.
(3)~Hypothesis-free sandwich for the gate-part gap at
$\beta b\ge m\ln2+C_G$: lower end $m\,\kT\ln2$ by the exact
per-realization cancellation (the saturation direction of the
cartridge tightness result, free of auxiliary hypotheses); upper
end by the unconditional chain-bound ceiling informed
$\le F_{\mathrm{cart}}+\kT\,TC(t_0)+C_G\,\kT
=F_{\mathrm{cart}}+m\,\kT\ln2+C_G\,\kT$ against the blind
achievability blind $\ge F_{\mathrm{cart}}-m\,\kT\ln2$ (purchase
the key state by erasure at cost $m\,\kT\ln2\le b$, then swap);
hence $\mathrm{Gap}_{\mathrm{gate}}\in[m\,\kT\ln2,\
2m\,\kT\ln2+2\,C_G\,\kT]$, and $n\ln2>2m\ln2+2\,C_G$ keeps
$\V(M_B)<\V(M_X)$. (The BA-conditional exact equality
$\mathrm{Gap}_{\mathrm{gate}}=m\,\kT\ln2$ is reference only; the
hypothesis is open in general and constructively satisfied on
cartridges.) (4)~is (2) versus (3); only the endpoints are
established. $\blacksquare$

\subsection{Witness A: device AC$'$, proof, and the two
variants}
\label{app:witness-a}

\emph{Device AC$'$ (shadow design).} A single task:
$X=(X_1[\text{2 bits}],X_2[\text{1 bit}])$ one manipulable
register (joint permutations available, $K=3$); $Y=$ a readable,
non-manipulable duplicate of the \emph{content} of $X_1$---so
$X$ and $Y$ are correlated, $I(X;Y)=2\ln2>0$. $D=$ perfect copy of
$X$ (clause ($\beta$) held). $T=\delta_\tau$, $b$ arbitrary;
$M_0$ blank; capacity two bits;
$F=\{f_1,f_2\}$ copying the $X_1$- resp.\ $X_2$-part of $D$. All
states (F5$'$) (machine-checked; no noise is needed).

\emph{Proof.} Corollary~\ref{cor:v-identity} arithmetic:
$\Delta\V(f_1)=\kT\,I(X_1;X\mid Y)=\kT\,I(X_1;X\mid X_1)=0$ and
$\Delta\V(f_2)=\kT\,I(X_2;X\mid Y)=\kT\,I(X_2;X_2)=\kT\ln2$,
against $\dpJD(f_1)=2\ln2>\ln2=\dpJD(f_2)$. For the restoration:
on $\tau'$ with $\mathrm{Man}=\{Z\}$, $Z=(X,Y)$ a single register
and $Y_{\tau'}=\emptyset$, the same arithmetic gives
$\V'(f_1)=2\kT\ln2>\kT\ln2=\V'(f_2)$. $\blacksquare$

\emph{Variants.} (i)~\emph{Environment-record variant}:
$D=(\text{copy of }X,\text{copy of }Y)$ with $Y$ independent
readable $n$ bits, $F=\{f_X,f_Y\}$: fit prefers $f_Y$ ($n\ln2$ vs
$\ln2$), value prefers $f_X$ ($\kT\ln2$ vs $0$)---a compound of
record excess ($D\supsetneq$ medium record) and the access axis,
and the direct argmax form of the $g_c$ subject.
(ii)~\emph{Noisy record}: $D_X=$ a binary-symmetric-channel
copy of $X$ with flip probability $\epsilon>0$. Copying it
\emph{violates} (F5$'$) (the conditional atoms
$(1-\epsilon,\epsilon)$ are neither point masses nor uniform), so
Corollary~\ref{cor:v-identity} does not apply; what remains
provable is the all-state upper bound
$\Delta\V(f_X)\le\kT(\ln2-h(\epsilon))$ for every $b$
(Lemma~\ref{lem:flattask}), while attainment for
$b\ge b^{*}(\epsilon)=\kT\ln\bigl(1/(2\epsilon)\bigr)$ is a
conjecture (the same open slot as the general-dyadic budget floor,
Remark~\ref{rem:flat-scope}(v)). \emph{The reversal claim for the
noisy variant is therefore excluded from Main
Theorem~\ref{thm:main-III}}, and only the upper bound and the
conjecture are recorded here.

\subsection{Witness R: device RJ and proof}
\label{app:witness-r}

\emph{Device RJ.} A single $\flatstar$ task: $X$ one uniform bit
(manipulable), $Y=\emptyset$, $K=3$; $J$: $n\ge2$ uniform bits
independent of $X$ (junk---$n\ge2$ makes the fit preference
strict). $D=(\text{copy of }X,\,J)$---clause
($\beta$) broken: the record exceeds the faithful observation of
the medium. $T=\delta_\tau$, $b$ arbitrary; capacity $n$ bits;
$F=\{f_X,f_J\}$. All states (F5$'$).

\emph{Proof.} $\dpJD(f_J)=n\ln2>\ln2=\dpJD(f_X)$;
$\Delta\V(f_X)=\kT\ln2$ (Corollary~\ref{cor:v-identity}) and
$\Delta\V(f_J)=0$ (Lemma~\ref{lem:noport}: the $f_J$-state is
jointly independent of $X$). Replacing every component of $D$ by
$X$-content restores agreement: the two argmax sets coincide (both
candidates then tie at $\ln2$ on fit and on value---weak alignment
is restored; no strict preference exists after the
restoration). $\blacksquare$

\subsection{Proof of the gate threshold transition
(Theorem~\ref{thm:gatethreshold})}
\label{app:proof-threshold}

\emph{(1)~Training side.} Lower end: the informed branch reads
$M=s(\Theta)$ and drives the known value to the target by a
zero-work conditional permutation (it draws no external work, so
it runs at $b=0$; statement B1 of Appendix~\ref{app:fixedI},
whose Type~II reading is unchanged), giving informed
$\ge F_{\mathrm{cart}}-C_G\,\kT$; the blind branch is capped by
the
imported ledger corollary at
$F_{\mathrm{cart}}c_{\mathrm{full}}+b+C_G\,\kT$; subtract. Upper
end: supply boundedness---prize plus budget plus bookkeeping. Here
the Type~II declaration is load-bearing: conditioned on $M=m$ the
medium $B$ remains uniform ($B\perp M$ at $t_0$), so no
conditional nonequilibrium resource exists to extract beyond the
prize; in a Type~I realization this conditional channel exists
and would break the upper end.

\emph{(2)~Shift side, upper bound.} The per-realization informed
branch after the shift is the $M=m$ conditional law. The
conditional prior of the redrawn key is exact: persisting bits
match, the remainder is uniform---the largest atom is
$w_0=2^{-k_{\mathrm{res}}}$ (uniform posterior: $\Theta$ uniform
and the persisting set chosen independently of $\Theta$). Apply
statement B3 of Appendix~\ref{app:fixedI} with view $v=m$: part
(a) gives
$P(\mathrm{unlock}\mid m)\le
K_{\mathrm{hor}}2^{-k_{\mathrm{res}}}e^{\beta b}$, and part (b)
($\theta$-wise ledger conversion) gives
$\mathbb{E}[W\mid m]\le F_{\mathrm{cart}}\min(1,\cdot)+b
+C_G\,\kT$;
take $\mathbb{E}_m$ and subtract blind $\ge0$ (the null protocol
is admissible).

\emph{(3)~Shift side, lower bound (achievability).} With the
persisting $c$ bits the informed side narrows the candidate space
to $2^{k_{\mathrm{res}}}$ and presents, over the
$K_{\mathrm{hor}}$ opportunities, one distinct candidate per
opportunity (the same opportunity semantics as the imported bound;
when $K_{\mathrm{hor}}>2^{k_{\mathrm{res}}}$ the list exhausts the
space and the min saturates at $1$). The tactic is non-adaptive---a fixed candidate list, no
flag reads---so failed branches cost exactly zero work and the
bookkeeping (a $C_G$ entry) is a single total, not proportional to
$K_{\mathrm{hor}}$. The success probability is exactly
$\min(1,K_{\mathrm{hor}}2^{-k_{\mathrm{res}}})$ (distinct
candidates, uniform posterior, disjoint pass events); the blind
side is charged at the ledger-corollary cap.

\emph{(4)}~is the arithmetic of the ratio of (1)--(3): on the left
shoulder, $\min(1,K_{\mathrm{hor}}2^{-k_{\mathrm{res}}})=1$ and
full suppression bounds
$\beta F_{\mathrm{cart}}\,c_{\mathrm{full}}\le1$, so parts
(3)/(1) give the displayed closed form; on the right shoulder,
parts (2)/(1) with
$K_{\mathrm{hor}}2^{-k_{\mathrm{res}}}e^{\beta b}
\le 2^{-t}/(\beta F_{\mathrm{cart}})$ give the displayed upper
bound; positivity (ii) keeps every denominator
positive. $\blacksquare$

\emph{The two premises, restated} (they carry the proof and bound
its scope): (i)~\emph{uniform posterior}---exact here since
$\Theta$ is uniform and the persisting set is $\Theta$-%
independent; (ii)~\emph{uniform $d$-fold degeneracy with Type~II
untouchable labels}---licensing the normal-form reduction to
$k_{\mathrm{eff}}$; non-uniform degeneracy is not treated, as in
Ref.~\cite{Sudo_Spectrum}.

\section{Fixed-I statements B1/B2/B3 (self-contained)}
\label{app:fixedI}

This appendix restates, self-containedly, the imported gate
statements of the companion framework~\cite{Sudo_MPU} used in
Secs.~\ref{sec:regime}, \ref{sec:retention},
and~\ref{sec:gate}: the two amplification statements B1/B2 and the
partial-information bound B3. The restatement is complete as a set
of
statements---nothing below relies on any compiled artifact of the
companion---and the premises are kept strictly apart: B1 uses
a prize-independent partial-suppression condition and no
threshold; B2 uses the exact full-suppression threshold; B3 is the
conditional (partial-information) form on the Type~II family.
Proofs
are in the companion.

\emph{Setting (cartridge family, Type I realization).} A key
register $B$ of $n$ uniform bits; a tamper-proof, pass/fail-only
gate that accepts at most $K_{\mathrm{hor}}$ submissions per run,
reveals only pass/fail, and on pass releases a prize---a cartridge
of free energy $F_{\mathrm{cart}}$; a per-run pathwise budget $b$
(Sec.~\ref{sec:setting-budget}); the datum is $D$, a perfect copy
of $B$, so the correlation content is $I(D;B)=n\ln2$, held fixed
throughout while $F_{\mathrm{cart}}$ is free. Write
\begin{equation}
c:=K_{\mathrm{hor}}\,2^{-n}e^{\beta b}
\label{eq:fixedI-c}
\end{equation}
(the suppression parameter of the companion statements, local to
this appendix; unrelated to the transfer dial $c$ of
Sec.~\ref{sec:gate}).
The blind side (every read port on $D$ deleted; re-optimization
from scratch) obeys the imported ledger bound
$\mathbb{E}[W_{\mathrm{ext}}]\le
F_{\mathrm{cart}}\min(1,c)+b+C_G\,\kT$. Every remainder of this
appendix is written with the same constant $C_G$ as the body: a
per-side bookkeeping remainder (the flag-and-residual accounting
of the ledger bound on the blind side, the implementation cost of
the conditional permutation on the informed side, each at most
$C_G\,\kT$), finite and uniform in $F_{\mathrm{cart}}$ and in the
device parameters as part of the imported statements (see the side
note after B1), while a proven numerical upper bound is
not---which is why the body's finite numerical displays are
script-convention illustrations.

\begin{theorem}[B1: fixed-content unboundedness under partial
suppression]
\label{thm:fixedI-B1}
Fix a finite $b$, a horizon $K_{\mathrm{hor}}$, and a key length
$n$ for which
\begin{equation}
c=K_{\mathrm{hor}}\,2^{-n}e^{\beta b}<1
\label{eq:fixedI-partial}
\end{equation}
---\emph{partial suppression}, a condition independent of
$F_{\mathrm{cart}}$. On the cartridge family with $I=n\ln2$ held
fixed and $F_{\mathrm{cart}}$ free: the blind side is capped at
$c\,F_{\mathrm{cart}}+b+C_G\,\kT$, while the informed side reads $D$
in each realization and drives the known value to the target by a
zero-work conditional permutation, so that
$\mathbb{E}[W_{\mathrm{ext}}]\ge F_{\mathrm{cart}}-C_G\,\kT$
\emph{already at $b=0$}---$C_G<\infty$ the per-side bookkeeping
remainder of the setting above, uniform in $F_{\mathrm{cart}}$.
Hence
\begin{equation}
\begin{aligned}
\Gap&\ge (1-c)\,F_{\mathrm{cart}}-b-2\,C_G\,\kT,\\
\Gap&\xrightarrow[F_{\mathrm{cart}}\to\infty]{}\infty\\[-2pt]
&\qquad\text{at fixed }I=n\ln2 ,
\end{aligned}
\label{eq:fixedI-growth}
\end{equation}
and consequently the value of data admits no upper bound by any
function of its mutual information alone: for every
$f:[0,\infty)\to\mathbb{R}$ there is a task in this family, with
$I=n\ln2$ unchanged, on which
$\Gap/(\kT\ln2)>f\bigl(I(D;B)\bigr)$ (the correlation content of
the setting above, held fixed at $n\ln2$).
\end{theorem}

Two side notes carried with the statement. \emph{The proof uses no
threshold}: it consists of the ledger cap with
$F_{\mathrm{cart}}\min(1,c)=c\,F_{\mathrm{cart}}$ ($c$ independent
of $F_{\mathrm{cart}}$), the budget-independence of the informed
swap, and a subtraction. \emph{The constant $C_G$ is uniform in
$F_{\mathrm{cart}}$}: both bookkeeping contributions---the
flag-and-residual accounting of the ledger bound and the informed
swap's implementation cost---are fixed by the device structure at
given $n$ and do not grow with the prize. B1 does \emph{not}
silence the blind side: under partial suppression alone, free
search wins with probability $K_{\mathrm{hor}}2^{-n}$ at zero work
and collects a constant fraction of the prize; suppressing the
blind side to its budget floor is a stronger, prize-dependent
demand---the second statement.

\begin{theorem}[B2: exact full-suppression threshold]
\label{thm:fixedI-B2}
Fix $b$ and $K_{\mathrm{hor}}$, and take the cartridge family of
Theorem~\ref{thm:fixedI-B1}. Whenever the key length clears the
\emph{full-suppression threshold}
\begin{equation}
n\ln2\ \ge\ \beta b+\ln(\beta F_{\mathrm{cart}})
+\ln K_{\mathrm{hor}}
\iff
\beta\,c\,F_{\mathrm{cart}}\le1 ,
\label{eq:fixedI-threshold}
\end{equation}
the blind side is capped at its budget floor,
\begin{equation*}
\begin{aligned}
\mathbb{E}[W_{\mathrm{ext}}]
&\le F_{\mathrm{cart}}K_{\mathrm{hor}}2^{-n}e^{\beta b}
  +b+C_G\,\kT\\
&\le b+(1+C_G)\,\kT
\end{aligned}
\end{equation*}
(the threshold makes the first term at most $\kT$; $C_G$ as in
B1), and
$\Gap\ge F_{\mathrm{cart}}-b-(1+2\,C_G)\,\kT$.
\end{theorem}

The threshold (\ref{eq:fixedI-threshold}) is exact in the
slope-bound sense: one nat of key per nat of budget, per log-nat
of prize, per log-count of opportunities---unit exchange rate in
each argument, with no Pinsker slack.

\begin{theorem}[B3: partial-information unlock and ledger bound,
Type II]
\label{thm:fixedI-B3}
Take the gate family in its Type~II realization: the key is a
frozen, energetically decoupled parameter $\Theta$ (labels
untouchable by the agent), the gate passes iff the submitted medium
state matches $s(\Theta)$, with at most $K_{\mathrm{hor}}$
submissions per run, pass/fail disclosure only, and no leakage
beyond the membership bit. Let an agent hold partial information on
the pass state: conditioned on its view $v$ at the start of the
run, the largest atom of its conditional prior over pass states is
$w_0(v)=2^{-H_{\min}(s(\Theta)\mid v)}$. Then:
(a)~(\emph{unlock bound, conditional}) $P(\mathrm{unlock}\mid v)\le
K_{\mathrm{hor}}\,w_0(v)\,e^{\beta b}$;
(b)~(\emph{ledger conversion, $\theta$-wise})
$\mathbb{E}[W_{\mathrm{ext}}\mid v]\le
F_{\mathrm{cart}}\min\bigl(1,K_{\mathrm{hor}}\,w_0(v)\,
e^{\beta b}\bigr)+b+C_G\,\kT$, with the bookkeeping remainder
$C_G$ finite and uniform in
$(w_0,F_{\mathrm{cart}},K_{\mathrm{hor}},\beta b)$ (the
flag-and-residual accounting is fixed by the device structure and
does not grow with the prize). The fully blind case
$w_0=2^{-n}$ recovers the ledger bound of the setting above. In
the companion, the device axiom of the unlock bound takes $w_0$ as
the largest atom of the agent's conditional prior---covering the
partially informed side---and the proof of the ledger corollary
proceeds $\theta$-wise, hence conditionally on the view; the
Type~II gate clause supplies $w_0$ directly in the min-entropy
form above.
\end{theorem}

\emph{Usage in this paper.} B1 carries the enablement route of the
regime map (the row labeled (B2) in Sec.~\ref{sec:regime}; the
regime-map row labels are unrelated to the B1/B2 statement names of
this appendix) and the $b=0$ reversal
of the budget witness (Appendix~\ref{app:witness-b}); its
partial-suppression premise (\ref{eq:fixedI-partial}) is carried
explicitly wherever the unboundedness is invoked---outside it
($c\ge1$) the blind side, too, can collect the prize. B2, translated
to the Type~II key-redraw family through the effective key length
$k_{\mathrm{eff}}$, is the full-suppression premise of
Theorem~\ref{thm:gatethreshold}. B3 is the load-bearing import of
the shift-side upper bound of Theorem~\ref{thm:gatethreshold}
(part (2) of its proof, Appendix~\ref{app:proof-threshold}), with
$v$ the retained memory $m$ and $w_0=2^{-k_{\mathrm{res}}}$ on the
uniform persisting fiber. All statements are single-fixed-%
device statements; no network or cascade form is used anywhere in
this paper. Through the single-task anchor of Main
Theorem~\ref{thm:main-I}(i)(e), the gaps above are the
$\kT\ln2$-denominated values of the companion framework on this
family.

\section{Notation table}
\label{app:notation}

\begin{table*}[t]
\caption{\label{tab:notation}%
Notation contract. Primary terms are fixed here; the body uses these
terms and no synonyms. The word ``fitting'' from an early working
document is rendered throughout as \emph{record-correlation increase}
(the predicate $\Delta I(M;D)>0$); $\Phifit$ always denotes the
training-side fit functional.}
\begin{ruledtabular}
\begin{tabular}{p{4.3cm}p{3.4cm}p{8.2cm}}
Symbol & Primary term & Definition and scope notes \\
\colrule
$\Phifit(f)$ & training-side fit &
Any functional ranking candidate updates from the training side
(training loss, likelihood, training work, value on the training
distribution). \\
$\JD(M)=I(M;D)$ & record-correlation stock &
Memorization; not identical to training-side fit. \\
$\dpJD = I(M';D\mid M)$ & gross acquisition &
Distinguished from the signed increment $\dJD = I(M';D)-I(M;D)$;
$\dJD$ telescopes along update sequences, $\dpJD$ in general does not. \\
\(\begin{gathered}
\V(M;T,b)={}\\[-2pt]
\Winformed(M;T,b)\\[-2pt]
{}-\Wblind(T,b)
\end{gathered}\) & capital value &
Blind deletes the $M$-read port and re-optimizes from scratch under the
same tasks, access structure, and evaluation budget $b$. \\
learning & learning (capitalization) &
An admissible $M$-local update with $\Delta \V>0$. ``Updates that do not
touch data are in general not learning'' is stated as a theorem only in
the fully proved $\flatstar$ scope. \\
$\sigmaM=\Delta I(M;D)+\Sigmatot$ & search ledger &
Memory-side subsystem account; not a pure-dissipation quantity. \\
$\etacap=\Delta \V/(\kT\,\sigmaM)$ & capitalization efficiency &
Defined for $\sigmaM>0$. In Corollary~\ref{cor:orthogonality} its
numerator is evaluated on $T_{\mathrm{train}}$; $\etacap$ carries
no shift argument. \\
\(\begin{gathered}
\Lgen=V_{\mathrm{train}}-V_{\mathrm{shift}},\\[-2pt]
\rhogen=V_{\mathrm{shift}}/V_{\mathrm{train}}
\end{gathered}\) & value retention &
Primary term is \emph{value retention}; ``thermodynamic generalization''
appears only as a related-work distinction. $\Lgen$ carries no sign
constraint; $\rhogen$ is defined only for $V_{\mathrm{train}}>0$ and is
not confined to $[0,1]$. \\
$(M,D)\perp Y_\tau$ & joint side-information neutrality &
The hypothesis of the raw record-stock layer
[Proposition~\ref{prop:corresp}(B)]; strictly stronger in effect
than $X_\tau\perp Y_\tau$
(Proposition~\ref{prop:corresp-boundary}). \\
$C_G$ & gate bookkeeping constant &
The bookkeeping remainder of the imported gate bounds, assumed
finite and uniform in the device parameters; no proven numerical
upper bound is available, so finite numerical gate displays are
script-convention illustrations
(Appendix~\ref{app:verification}). \\
\end{tabular}
\end{ruledtabular}
\end{table*}

\bibliographystyle{apsrev4-2}
\bibliography{refs}

\end{document}